\documentclass[onecolumn]{IEEEtran}

\usepackage{amsthm}
\usepackage{amssymb,amsmath,amsfonts,latexsym, enumerate,url, cases}
\usepackage{graphicx,color}
\usepackage{booktabs}
\usepackage{threeparttable}
\usepackage{array}
\usepackage{float}
\numberwithin{equation}{section}
\usepackage{hyperref}
\usepackage{newtxtext}
\usepackage[all]{xy}
\usepackage[
backend=biber,
style=ieee,
sorting=nyt,
defernumbers=false
]{biblatex}

\newtheorem{theorem}{Theorem}[section] %
\newtheorem{lemma}[theorem]{Lemma} %
\newtheorem{corollary}[theorem]{Corollary} %
\newtheorem{proposition}[theorem]{Proposition} %
\newtheorem{definition}[theorem]{Definition} %
\newtheorem{example}[theorem]{Example}
\usepackage{nicematrix}
\usepackage{longtable}
\begin{document}
	
	\title{Galois Hulls of Generalized Roth-Lempel Codes and Their Applications to EAQECCs}
	\author{ Xuefei Wu, 
		 Qi Liu, 
		 Yingchun Chen, 
		 Haiyan Zhou$^*$%
		 \thanks{Xuefei Wu, Qi Liu, 
		 	Yingchun Chen and 
		 	Haiyan Zhou are with the School of Mathematics, Nanjing Normal University, Nanjing, China (e-mail: 240901020@njnu.edu.com; QuLiu@aliyun.com;
		 	240901002@njnu.edu.com;
		 	zhouhy@njnu.edu.com).}
		 }
	\maketitle	
	\begin{abstract}
		The dimension of a Galois hull is an important parameter in the construction of entanglement-assisted quantum error-correcting codes. In this paper, we study generalized Roth-Lempel (GRL) codes with prescribed Galois hull dimensions over finite fields of arbitrary characteristic. We develop a common construction method based on normalized Lagrange coefficients and suitable choices of column multipliers, and obtain six explicit constructions from multiplicative cosets, trace fibers, additive-subspace cosets, and mixed additive--multiplicative fibers. For extension sizes $s=2$ and $s=3$, we construct MDS and AMDS GRL codes with Galois hulls of specific dimensions. The resulting codes provide EAQECCs with explicit dimension, entanglement consumption, and relatively large minimum distance. Taking zero-dimensional hulls also gives Hermitian LCD GRL codes and the corresponding maximally entangled EAQECCs.
	\end{abstract}
	
	
	\section{Introduction}\label{sec:introduction}
	
	Let $q=p^e$ and $0\le\ell\le e-1$. The $\ell$-Galois hull of a linear code $C$ over $\mathbb F_q$ is $\operatorname{Hull}_\ell(C)=C\cap C^{\perp_\ell}$, where $C^{\perp_\ell}$ is defined using the inner product $\langle\boldsymbol x,\boldsymbol y\rangle_\ell=\sum_i x_i y_i^{p^\ell}$. This includes the Euclidean and Hermitian hulls as special cases. A zero-dimensional hull gives a Galois linear complementary dual code, while other prescribed hull dimensions allow different choices of dimension and entanglement consumption in entanglement-assisted quantum error-correcting codes (EAQECCs). In particular, an $[n,k,d]_q$ linear code with an $h$-dimensional Galois hull yields an EAQECC encoding $k-h$ logical qudits and consuming $n-k-h$ maximally entangled pairs, with minimum distance at least $d$~\cite{9559992}.
	
	Quantum error-correcting codes (QECCs) play an important role in protecting quantum information against errors arising during storage and transmission. The development of quantum error correction began with Shor's nine-qubit code in 1995~\cite{Shor1995} and Steane's construction in 1996~\cite{Steane1996}. Subsequently, the stabilizer formalism and related algebraic constructions established a close connection between classical coding theory and quantum error correction. For more information, we refer the reader to \cite{Gottesman1996,KnillLaflamme1997,CRSS1998,AshikhminKnill2001}.
	
	In 1996, Calderbank and Shor~\cite{CalderbankShor1996} and Steane~\cite{Steane1996} proposed the Calderbank--Shor--Steane (CSS) construction, which produces quantum codes from classical linear codes satisfying suitable duality conditions. To remove the restriction of self-orthogonality or duality, Brun, Devetak, and Hsieh introduced entanglement-assisted quantum error-correcting codes (EAQECCs) in 2006~\cite{BrunDevetakHsieh2006}. Later developments related the required entanglement to ranks of matrices and, consequently, to the dimensions of classical hulls~\cite{WildeBrun2008,GuendaJitmanGulliver2018,GalindoHernandoMatsumotoRuano2019}. 
	
	Reed--Solomon codes~\cite{ReedSolomon1960} and generalized Reed--Solomon (GRS) codes are particularly suitable for quantum constructions because of their MDS property and flexible evaluation points and column multipliers. Early quantum MDS constructions from GRS codes and related self-orthogonal variants can be found in~\cite{GrasslBethRoetteler2004,LiXingWang2008,JinLingLuoXing2010,JinXing2014}. The study of prescribed hull dimensions further enlarged this framework. Luo, Cao, and Chen~\cite{LuoCaoChen2019} and Fang et al.~\cite{FangFuLiZhu2020} constructed MDS codes with flexible Euclidean or Hermitian hull dimensions. Fan and Zhang~\cite{FanZhang2017} introduced the Galois inner product, and Liu and Pan~\cite{LiuPan2020} further studied Galois hulls of linear codes. In 2021, Cao~\cite{9559992} systematically studied GRS codes with prescribed Galois hull dimensions and derived several families of MDS EAQECCs. Subsequent developments include constructions based on self-orthogonal GRS codes, Goppa representations, larger lengths, and intersections of GRS codes~\cite{FangJinLuoMa2022,WuLiYang2022,LiZhuLi2023,li2023several,Cao2023,LiZhu2024,WanZhu2025,LiuChen2025}. These works provide the main GRS background for constructing EAQECCs with prescribed Galois hull dimensions.
	
	Roth-Lempel (RL) codes, introduced by Roth and Lempel in 1989~\cite{roth1989construction}, provide another algebraic framework for constructing codes with strong distance properties. In recent years, RL codes have attracted renewed attention because of their MDS, AMDS, NMDS, LCD, and related duality properties~\cite{WuHyunLee2021,HanFan2023,WuHengLiDing2024}. In 2025, Li, Jiang, and Liu~\cite{li2025new} and Liang, Wan, and Liao~\cite{LiangWanLiao2026GRL} introduced generalized Roth-Lempel (GRL) codes by appending to a GRS generator matrix a short block determined by a nonsingular matrix $A_s$. Subsequent work investigated MDS, AMDS, NMDS, self-dual, and extended GRL codes~\cite{LiangLiao2026NMDS,LiangLiao2026Extended}. More recently, Liang et al.~\cite{Liang2026LCD} constructed Euclidean and Hermitian LCD GRL codes and derived EAQECCs. Liu et al.~\cite{liu2026generalized} further established NMDS characterizations for GRL codes with $s=2$ and $s=3$ and constructed Hermitian self-orthogonal GRL codes and associated quantum codes.
	
	These developments motivate the study of GRL codes whose Galois hull dimensions can be prescribed throughout explicit ranges while their classical minimum distances remain determined. The appended coordinates make this problem different from the corresponding GRS problem: a choice of multipliers that controls the evaluation part must also be compatible with the conditions imposed by $A_s$. Moreover, in boundary cases, the matrix relation used to obtain the desired hull dimension restricts the available extension matrices. The distance conditions and the hull conditions must therefore be considered together.

	In this paper, we address this problem over $\mathbb F_q$, where $q=p^e>4$, $p$ is a prime, $\ell>0$, and $2\ell\mid e$. Our main contributions are as follows.
	\begin{enumerate}
		\item We develop a common hull construction for GRL codes by normalizing the Lagrange coefficients and reducing the Galois-dual conditions to polynomial identities. In the ordinary ranges, the method works with an arbitrary nonsingular extension matrix and gives every hull dimension $0\le h\le k-s-1$ or $0\le h\le k-s$, according to the coefficient condition. We also treat additional boundary cases with suitable extension matrices.
		\item We apply the common method to six explicit constructions using multiplicative, additive, and mixed evaluation sets. For $s=2$ and $s=3$, we combine these constructions with subset-sum and determinant criteria for the MDS and AMDS properties. The full evaluation set gives $[q+2,k,q-k+2]_q$ AMDS GRL codes, including special choices with hull dimension $k-1$.
		\item We derive EAQECCs from the resulting classical codes and give explicit formulas for their encoded dimension and entanglement consumption. Associated with the AMDS conditions, we obtain several families of AMDS EAQECCs, see Table \ref{tab:EA-amds-examples}.
		\item By taking $h=0$, we obtain $\ell$-Galois LCD GRL codes and maximally entangled EAQECCs. In the Hermitian case $q=p^{2\ell}$, these constructions yield EAQECCs over $\mathbb F_{p^\ell}$; comparison with existing Hermitian LCD GRL constructions also gives families outside their cited parameter ranges, see Table \ref{tab:LCD-examples}.
	\end{enumerate}
	
	
	The rest of this paper is organized as follows. Section~\ref{sec:preliminaries} recalls the necessary definitions and the polynomial description of GRL duals. Section~\ref{sec:construction} presents the common construction and the six explicit evaluation-set constructions. Sections~\ref{sec:s2} and~\ref{sec:s3} study the MDS and AMDS properties together with prescribed Galois hulls for $s=2$ and $s=3$, respectively. Section~\ref{sec:eaqecc} gives the applications to EAQECCs. Section~\ref{sec:LCD} treats LCD GRL codes and the associated maximally entangled EAQECCs. Section~\ref{sec:conclusion} concludes the paper.
	
	\section{Preliminaries}\label{sec:preliminaries}
	
	Let $q=p^e$, where $p$ is a prime and $e$ is a positive integer, and
	let $\mathbb F_q$ denote the finite field with $q$ elements. For a
	positive integer $k$, put
	$$
	\mathbb F_q[x]_{<k}
	=\{f(x)\in\mathbb F_q[x]:\deg f(x)\le k-1\}.
	$$
	
	An $[n,k,d]_q$ linear code is a $k$-dimensional subspace of
	$\mathbb F_q^n$ with minimum Hamming distance $d$. The Singleton bound
	gives $d\le n-k+1$. A code is called maximum distance separable (MDS)
	if $d=n-k+1$, and almost MDS (AMDS) if $d=n-k$. It is called near-MDS
	(NMDS) if both $C$ and its Euclidean dual $C^\perp$ are AMDS.
	
	For a vector, matrix, or code $X$ over $\mathbb F_q$, the notation
	$X^{(p^r)}$ means that the map $x\mapsto x^{p^r}$ is applied
	entrywise. If $f(x)=\sum_j f_jx^j$, we write
	$$
	f^{p^r}(x)=(f(x))^{p^r}=\sum_j f_j^{p^r}x^{jp^r}.
	$$
	
	\begin{definition}\label{def:Galois}
		Let $0\le\ell\le e-1$. For
		$\boldsymbol x=(x_1,\ldots,x_n)$ and
		$\boldsymbol y=(y_1,\ldots,y_n)$ in $\mathbb F_q^n$, the
		$\ell$-Galois inner product is defined by
		$$
		\langle\boldsymbol x,\boldsymbol y\rangle_\ell
		=\sum_{i=1}^n x_i y_i^{p^\ell}.
		$$
		The $\ell$-Galois dual and the $\ell$-Galois hull of a linear code
		$C\subseteq\mathbb F_q^n$ are, respectively,
		$$
		C^{\perp_\ell}
		=\{\boldsymbol x\in\mathbb F_q^n:
		\langle\boldsymbol c,\boldsymbol x\rangle_\ell=0
		\text{ for all }\boldsymbol c\in C\}
		$$
		and
		$$
		\operatorname{Hull}_\ell(C)=C\cap C^{\perp_\ell}.
		$$
		When $\ell=0$, these notions reduce to the Euclidean inner product,
		dual, and hull. When $e$ is even and $\ell=e/2$, they reduce to the
		Hermitian inner product, dual, and hull.
	\end{definition}
	
	\begin{lemma}\label{lem:galois-euclidean-dual}\cite{li2023several}
		For every linear code $C$ over $\mathbb F_q$,
		$$
		C^{\perp_\ell}=(C^\perp)^{(p^{e-\ell})}.
		$$
	\end{lemma}
	
	\begin{lemma}\label{lem:galois-hull-rank}\cite{9559992}
		Let $G$ be a generator matrix of an $[n,k]_q$ linear code $C$.
		Then
		$$
		\dim\bigl(\operatorname{Hull}_\ell(C)\bigr)
		=k-\operatorname{rank}\left(
		G\bigl(G^{(p^\ell)}\bigr)^{T}
		\right).
		$$
	\end{lemma}
	
	\begin{proposition}\label{prop:minimum-distance-rank}
		Let $G$ be a generator matrix of an $[n,k]_q$ linear code $C$.
		Then the minimum distance of $C$ is $d$ if and only if the following
		conditions hold:
		\begin{enumerate}
			\item[(1)] There exist $n-d$ columns of $G$ whose rank is less
			than $k$.
			\item[(2)] Every submatrix formed by $n-d+1$ columns of $G$ has
			rank $k$.
		\end{enumerate}
	\end{proposition}
	
	Let $\boldsymbol a=(a_1,\ldots,a_n)$ be a vector of pairwise distinct
	elements of $\mathbb F_q$, and let
	$\boldsymbol v=(v_1,\ldots,v_n)\in(\mathbb F_q^*)^n$. The generalized
	Reed--Solomon code associated with $\boldsymbol a$ and $\boldsymbol v$
	is
	$$
	\operatorname{GRS}_k(\boldsymbol a,\boldsymbol v)
	=\left\{
	(v_1f(a_1),\ldots,v_nf(a_n)):
	f(x)\in\mathbb F_q[x]_{<k}
	\right\}.
	$$
	Its standard generator matrix is
	$$
	G_k(\boldsymbol a,\boldsymbol v)
	=\bigl(v_j a_j^{i-1}\bigr)_{
		1\le i\le k,\,1\le j\le n}.
	$$
	The extended generalized Reed--Solomon code is defined by
	$$
	\operatorname{GRS}_k(\boldsymbol a,\boldsymbol v,\infty)
	=\left\{
	(v_1f(a_1),\ldots,v_nf(a_n),f_{k-1}):
	f(x)=\sum_{j=0}^{k-1}f_jx^j\in\mathbb F_q[x]_{<k}
	\right\}.
	$$
	
	For $1\le i\le n$, put
	\begin{equation}\label{eq:ui}
		u_i=\prod_{\substack{1\le j\le n\\j\ne i}}
		(a_i-a_j)^{-1}.
	\end{equation}
	These numbers are called the \emph{Lagrange coefficients} associated with $\boldsymbol a$. Indeed, Lagrange interpolation gives
	$$
	f(x)=\sum_{i=1}^n f(a_i)u_i\prod_{j\ne i}(x-a_j),\qquad \deg f<n,
	$$
	so $u_i$ is the coefficient of $x^{n-1}$ in the $i$-th interpolation basis polynomial. Comparing leading coefficients gives $\sum_{i=1}^n u_i a_i^j=0$ for $0\le j\le n-2$ and $\sum_{i=1}^n u_i a_i^{n-1}=1$.
	
	\begin{lemma}\label{lem:grs-galois-dual}\cite{li2023several}
		Let $f(x)\in\mathbb F_q[x]_{<k}$, $C=\operatorname{GRS}_k(\boldsymbol a,\boldsymbol v)$ as above.
		\begin{enumerate}
			\item[(1)] The codeword
			$\boldsymbol c=(v_1f(a_1),\ldots,v_nf(a_n))$ of $C$ belongs to
			$\operatorname{GRS}_k(\boldsymbol a,\boldsymbol v)^{\perp_\ell}$
			if and only if there exists $g(x)\in\mathbb F_q[x]$ with
			$\deg g(x)\le n-k-1$ such that
			$$
			v_i^{p^\ell+1}f^{p^\ell}(a_i)=u_i g(a_i),
			\qquad 1\le i\le n.
			$$
			
			\item[(2)] The codeword
			$$
			\boldsymbol c=(v_1f(a_1),\ldots,v_nf(a_n),f_{k-1})
			$$
			belongs to
			$\operatorname{GRS}_k(\boldsymbol a,\boldsymbol v,\infty)^{\perp_\ell}$
			if and only if there exists $g(x)=\sum_jg_jx^j\in\mathbb F_q[x]$
			with $\deg g(x)\le n-k$ such that
			$$
			v_i^{p^\ell+1}f^{p^\ell}(a_i)=u_i g(a_i),
			\qquad 1\le i\le n,
			$$
			and
			$$
			f_{k-1}^{p^\ell}=-g_{n-k}.
			$$
		\end{enumerate}
	\end{lemma}
	
	We next recall the generalized Roth-Lempel codes introduced in
	\cite{li2025new}. Let $1\le s\le k\le n$ and let
	$A_s\in\operatorname{GL}_s(\mathbb F_q)$. Define
	$$
	E_s=
	\begin{pmatrix}
		O_{(k-s)\times s}\\
		A_s
	\end{pmatrix}.
	$$
	
	\begin{definition}\label{def:GRL}
		The generalized Roth-Lempel code
		$\operatorname{GRL}_k(\boldsymbol a,\boldsymbol v,A_s)$ is the
		$[n+s,k]_q$ linear code generated by
		$$
		G_k(\boldsymbol a,\boldsymbol v,A_s)
		=\left(G_k(\boldsymbol a,\boldsymbol v)\mid E_s\right).
		$$
		Equivalently, if
		$f(x)=\sum_{j=0}^{k-1}f_jx^j\in\mathbb F_q[x]_{<k}$ and
		$$
		\boldsymbol s_k(f)=(f_{k-s},\ldots,f_{k-1}),
		$$
		then its corresponding codeword is
		$$
		\boldsymbol c_f
		=\bigl(v_1f(a_1),\ldots,v_nf(a_n),
		\boldsymbol s_k(f)A_s\bigr).
		$$
	\end{definition}
	
	When $s=1$ and $A_1=(1)$, Definition~\ref{def:GRL} gives the extended
	GRS code. 
	
	When $\boldsymbol v=(1,\cdots,1)$, $s=2$ and
	$$
	A_2=
	\begin{pmatrix}
		0&1\\
		1&\zeta
	\end{pmatrix},
	$$
	it gives the classical Roth-Lempel code
	\cite{roth1989construction}.
	
	For $j\ge0$, define
	\begin{equation}\label{eq:anj}
		a_{n-1}^{(j)}
		:=\sum_{i=1}^n u_i a_i^{n-1+j}.
	\end{equation}
	In particular, $a_{n-1}^{(0)}=1$. 
	
	Let
	\begin{equation}\label{eq:Ts}
		T_s=\bigl(t_{r,t}\bigr)_{1\le r,t\le s},
		\qquad
		t_{r,t}=
		\begin{cases}
			0, &r+t\le s,\\
			a_{n-1}^{(r+t-s-1)}, &r+t\ge s+1.
		\end{cases}
	\end{equation}
	For $s\ge3$, this matrix can be written as
	$$
	T_s=
	\begin{pmatrix}
		0&0&\cdots&0&1\\
		0&0&\cdots&1&a_{n-1}^{(1)}\\
		\vdots&\vdots&&\vdots&\vdots\\
		0&1&\cdots&a_{n-1}^{(s-3)}&a_{n-1}^{(s-2)}\\
		1&a_{n-1}^{(1)}&\cdots&a_{n-1}^{(s-2)}&a_{n-1}^{(s-1)}
	\end{pmatrix}.
	$$
	For the two smaller cases, $T_1=(1)$ and
	$$
	T_2=
	\begin{pmatrix}
		0&1\\
		1&a_{n-1}^{(1)}
	\end{pmatrix}.
	$$
	Since the anti-diagonal entries of $T_s$ are all equal to $1$,
	$T_s$ is nonsingular. Put
	\begin{equation}\label{eq:Bs}
		B_s=-T_s\bigl(A_s^{-1}\bigr)^{\mathsf T}.
	\end{equation}
	Obviously, $B_{s}$ is also nonsingular.
	
	\begin{proposition}\label{prop:grl-parity-check}\cite{li2025new}
		A parity-check matrix of
		$C=\operatorname{GRL}_k(\boldsymbol a,\boldsymbol v,A_s)$ is
		$$
		H=
		\left(
		G_{n+s-k}(\boldsymbol a,
		(u_1v_1^{-1},\ldots,u_nv_n^{-1}))
		\,\middle|\,
		\begin{pmatrix}
			O_{(n-k)\times s}\\
			B_s
		\end{pmatrix}
		\right).
		$$
		Consequently,
		$$
		C^\perp
		=\left\{
		\bigl(u_1v_1^{-1}g(a_1),\ldots,u_nv_n^{-1}g(a_n),
		\boldsymbol s_{n+s-k}(g)B_s\bigr):
		g(x)\in\mathbb F_q[x]_{<n+s-k}
		\right\},
		$$
		where
		$$
		\boldsymbol s_{n+s-k}(g)
		=(g_{n-k},\ldots,g_{n-k+s-1}).
		$$
	\end{proposition}

	Notice that
	$$
	\det\bigl(A_s^{(p^\ell)}\bigr)
	=\bigl(\det A_s\bigr)^{p^\ell}.
	$$
	Hence $A_s^{(p^\ell)}$ is nonsingular whenever $A_s$ is
	nonsingular.

	The following characterization will be used throughout the subsequent
	constructions.
	
	\begin{proposition}\label{prop:II.2}
		Let
		$C=\operatorname{GRL}_k(\boldsymbol a,\boldsymbol v,A_s)$, and let
		$\boldsymbol c_f\in C$ be the codeword corresponding to
		$f(x)\in\mathbb F_q[x]_{<k}$. Then
		$\boldsymbol c_f\in C^{\perp_\ell}$ if and only if there exists
		$g(x)\in\mathbb F_q[x]$ with
		$\deg g(x)\le n-k+s-1$ such that
		\begin{equation}\label{eq:grl-galois-dual}
			\begin{cases}
				v_i^{p^\ell+1}f^{p^\ell}(a_i)=u_i g(a_i),
				&1\le i\le n,\\
				\boldsymbol s_k(f)^{(p^\ell)}A_s^{(p^\ell)}
				=\boldsymbol s_{n+s-k}(g)B_s.&
			\end{cases}
		\end{equation}
	\end{proposition}
	
	\begin{proof}
		By Lemma~\ref{lem:galois-euclidean-dual},
		$\boldsymbol c_f\in C^{\perp_\ell}$ if and only if
		$\boldsymbol c_f^{(p^\ell)}\in C^\perp$. Applying
		Proposition~\ref{prop:grl-parity-check} and comparing the first $n$
		coordinates gives
		$$
		v_i^{p^\ell}f^{p^\ell}(a_i)
		=u_iv_i^{-1}g(a_i),
		\qquad 1\le i\le n,
		$$
		which is equivalent to the first line of
		\eqref{eq:grl-galois-dual}. Comparing the last $s$ coordinates gives
		the second line. This proves the result.
	\end{proof}
	
	To balance the coefficient in \eqref{eq:grl-galois-dual}, we use the following extension of \cite[Lemma III.2]{9559992} to arbitrary characteristic.
	\begin{lemma}\label{lem:III.2}
		Let $q=p^e$, where $p$ is a prime, and let $\ell$ be a positive integer such that $\ell\mid e$. Then, for every $u\in\mathbb F_{p^\ell}^*$, there exists $v\in\mathbb F_q^*$ satisfying
		$v^{p^\ell+1}=u$ if and only if either $p=2$ or $2\ell\mid e$.
	\end{lemma}
	
	\begin{proof}
		If $p$ is odd, the result follows directly from \cite[Lemma III.2]{9559992}. Now assume that $p=2$. Since the Frobenius map $x\mapsto x^2$ is an automorphism of $\mathbb F_{2^\ell}$, for any $u\in\mathbb F_{2^\ell}^*$, there exists $v\in\mathbb F_{2^\ell}^*$ such that $v^2=u$; in fact, one may take $v=u^{2^{\ell-1}}$. Since $\ell\mid e$, we have $\mathbb F_{2^\ell}\subseteq\mathbb F_q$, and
		$v^{2^\ell+1}=v^{2^\ell}v=v^2=u$.
		Hence the required $v$ always exists when $p=2$. This completes the proof.
	\end{proof}
	Under $\ell\mid e$ with $\ell>0$, we also have $p^\ell+1\mid q-1$ if and only if $2\ell\mid e$, since $q-1\equiv(-1)^{e/\ell}-1\pmod{p^\ell+1}$ and $p^\ell+1>2$. Thus we retain $2\ell\mid e$ in the subsequent constructions in both odd and even characteristic.
	
	The following inclusion also holds in arbitrary characteristic.
	\begin{lemma}\label{prop:complementary-subfield}\cite{li2023several}
		Let $q=p^e$, where $p$ is a prime and $0<\ell<e$. Suppose that $2(e-\ell)\mid e$. Then
		$\mathbb F_{p^{e-\ell}}\subseteq\mathbb F_q$ and
		$$
		\mathbb F_{p^{e-\ell}}^{*}\subseteq
		\left(\mathbb F_q^{*}\right)^{p^\ell+1}.
		$$
		In particular, for every $c\in\mathbb F_{p^{e-\ell}}^{*}$, there exists
		$v\in\mathbb F_q^{*}$ such that $v^{p^\ell+1}=c$.
	\end{lemma}
	\begin{proof}
		Since $2(e-\ell)\mid e$, we have $(e-\ell)\mid\ell$. For $c\in\mathbb F_{p^{e-\ell}}^*$, Lemma~\ref{lem:III.2} gives $v\in\mathbb F_q^*$ with $v^{p^{e-\ell}+1}=c$. Hence $v^{p^\ell+1}=(v^{p^{e-\ell}+1})^{p^\ell}=c^{p^\ell}=c$.
	\end{proof}
	The following observation shows that the $\ell$-Galois hull and the
	$(e-\ell)$-Galois hull always have the same dimension.
	\begin{proposition}\label{prop:complementary-hull}\cite{li2023several}
		Let $C$ be an $[n,k]_q$ linear code over $\mathbb F_q$, where $q=p^e$. For any $1\le \ell\le e-1$,
		\[
		\dim\bigl(\operatorname{Hull}_{\ell}(C)\bigr)
		=
		\dim\bigl(\operatorname{Hull}_{e-\ell}(C)\bigr).
		\]
	\end{proposition}

	\section{Main construction}\label{sec:construction}
	\renewcommand{\thesubsection}{\mbox{\thesection-\arabic{subsection}}}
	\renewcommand{\thesubsectiondis}{\arabic{subsection}.}
	In this section, we first present a general construction strategy that will be used throughout the subsequent explicit constructions. We then establish several theorems based on different coset constructions.
	
	\subsection{A common construction strategy}
	
	Let $q=p^e>4$, where $p$ is a prime, and assume that $\ell>0$ and $2\ell\mid e$.
	Then $(\mathbb F_q^*)^{p^\ell+1}$ has $(q-1)/(p^\ell+1)>1$ elements, which ensures the choices of $\beta$ and $\eta$ below. The condition $q>4$ excludes the degenerate case $q=4$, $\ell=1$, in which this group is $\{1\}$.
	Let $\boldsymbol a=(a_1,\ldots,a_n)$ be a vector of distinct elements
	of $\mathbb F_q$, and put
	$u_i=\prod_{1\le j\le n,\,j\ne i}(a_i-a_j)^{-1}$. The following
	results isolate the common arguments used in all the explicit
	constructions.
	
	\begin{proposition}\label{thm:common-nonzero}
		Assume that $a_i\ne0$ and
		$a_i^{-1}u_i\in\mathbb F_{p^\ell}^{*}$ for $1\le i\le n$.
		Let $A_s\in\operatorname{GL}_s(\mathbb F_q)$ be arbitrary. Then,
		for every $s+1\le k\le\left\lfloor\frac{n+p^\ell}{p^\ell+1}\right\rfloor$
		and every $0\le h\le k-s-1$, there exists an $[n+s,k]_q$
		generalized Roth-Lempel code $C$ with extension matrix $A_s$ such
		that $\dim(\operatorname{Hull}_{\ell}(C))=h$.
	\end{proposition}
	
	\begin{proof}
		By Lemma~\ref{lem:III.2}, choose $v_i\in\mathbb F_q^{*}$ such that
		$v_i^{p^\ell+1}=a_i^{-1}u_i$ for $1\le i\le n$. Set
		$z=k-s-h-1$ and choose $\beta\in\mathbb F_q^{*}$ such that
		$\gamma:=\beta^{p^\ell+1}\ne1$. Let $C$ be obtained from
		$\operatorname{GRL}_k(\boldsymbol a,\boldsymbol v,A_s)$ by replacing
		$v_1,\ldots,v_z$ with $\beta v_1,\ldots,\beta v_z$.
		
		Let a codeword represented by $f(x)$ with $\deg f(x)\le k-1$
		belong to $\operatorname{Hull}_{\ell}(C)$. By
		Proposition~\ref{prop:II.2}, there exists $g(x)\in\mathbb F_q[x]$
		with $\deg g(x)\le n-k+s-1$ such that
		\begin{equation}\label{eq:common-nonzero}
			\begin{cases}
				\gamma a_i^{-1}u_i f^{p^\ell}(a_i)=u_i g(a_i), &1\le i\le z,\\
				a_i^{-1}u_i f^{p^\ell}(a_i)=u_i g(a_i), &z+1\le i\le n,\\
				(f_{k-s}^{p^\ell},\ldots,f_{k-1}^{p^\ell})A_s^{(p^\ell)}
				=(g_{n-k},\ldots,g_{n-k+s-1})B_s.
			\end{cases}
		\end{equation}
		
		The assumption on $k$ gives $p^\ell(k-1)\le n-k$. For
		$i=z+1,\ldots,n$, we have $f^{p^\ell}(a_i)=a_i g(a_i)$.
		Thus $f^{p^\ell}(x)-xg(x)$ has at least
		$n-z=n-k+s+h+1\ge n-k+s+1$ distinct zeroes, while its degree is at
		most $n-k+s$. Therefore $f^{p^\ell}(x)=xg(x)$. In particular,
		$f_0=0$ and $\deg g(x)\le n-k-1$, so the last equation of
		\eqref{eq:common-nonzero} gives
		$f_{k-s}=\cdots=f_{k-1}=0$. For $1\le i\le z$, comparison with the
		first equation gives $f(a_i)=0$. Hence
		$f(x)=xc(x)\prod_{i=1}^{z}(x-a_i)$ with $\deg c(x)\le h-1$, and
		thus $\dim(\operatorname{Hull}_{\ell}(C))\le h$.
		
		Conversely, for every
		$f(x)=xc(x)\prod_{i=1}^{z}(x-a_i)$ with $\deg c(x)\le h-1$, take
		$g(x)=x^{-1}f^{p^\ell}(x)$. Then the equations in
		Proposition~\ref{prop:II.2} are satisfied. Thus
		$\dim(\operatorname{Hull}_{\ell}(C))\ge h$, and therefore
		$\dim(\operatorname{Hull}_{\ell}(C))=h$.
	\end{proof}
	
	\begin{proposition}\label{thm:common-direct}
		Let $\delta\in\mathbb Z_{>0}$ and $\lambda\in\mathbb F_q^{*}$. Assume that $\lambda a_i^{\delta-1}u_i\in\mathbb F_{p^\ell}^{*}$ for $1\le i\le n$. Let $A_s\in\operatorname{GL}_s(\mathbb F_q)$ be arbitrary. Then, for every $s<k\le\left\lfloor\frac{n+p^\ell-\delta}{p^\ell+1}\right\rfloor$ and every $0\le h\le k-s$, there exists an $[n+s,k]_q$ generalized Roth-Lempel code $C$ with extension matrix $A_s$ such that $\dim(\operatorname{Hull}_{\ell}(C))=h$.
	\end{proposition}
	
	\begin{proof}
		By Lemma~\ref{lem:III.2}, choose $v_i\in\mathbb F_q^{*}$ such that $v_i^{p^\ell+1}=\lambda a_i^{\delta-1}u_i$ for $1\le i\le n$. Set $z=k-s-h$ and choose $\beta\in\mathbb F_q^{*}$ such that $\gamma:=\beta^{p^\ell+1}\ne1$. Let $C$ be obtained from $\operatorname{GRL}_k(\boldsymbol a,\boldsymbol v,A_s)$ by replacing $v_1,\ldots,v_z$ with $\beta v_1,\ldots,\beta v_z$.
		
		Let a codeword represented by $f(x)$ with $\deg f(x)\le k-1$ belong to $\operatorname{Hull}_{\ell}(C)$. By Proposition~\ref{prop:II.2}, there exists $g(x)\in\mathbb F_q[x]$ with $\deg g(x)\le n-k+s-1$ such that
		\begin{equation}\label{eq:common-direct}
			\begin{cases}
				\gamma\lambda a_i^{\delta-1}u_i f^{p^\ell}(a_i)=u_i g(a_i), &1\le i\le z,\\
				\lambda a_i^{\delta-1}u_i f^{p^\ell}(a_i)=u_i g(a_i), &z+1\le i\le n,\\
				(f_{k-s}^{p^\ell},\ldots,f_{k-1}^{p^\ell})A_s^{(p^\ell)}
				=(g_{n-k},\ldots,g_{n-k+s-1})B_s.
			\end{cases}
		\end{equation}
		
		The assumption on $k$ gives $\delta-1+p^\ell(k-1)<n-k$, and hence $\deg(\lambda x^{\delta-1}f^{p^\ell}(x))\le n-k-1$. For $i=z+1,\ldots,n$, we have $g(a_i)=\lambda a_i^{\delta-1}f^{p^\ell}(a_i)$. Thus $g(x)-\lambda x^{\delta-1}f^{p^\ell}(x)$ has at least $n-z=n-k+s+h\ge n-k+s$ distinct zeroes, while its degree is at most $n-k+s-1$. Therefore $g(x)=\lambda x^{\delta-1}f^{p^\ell}(x)$. In particular, $\deg g(x)\le n-k-1$, and hence $g_{n-k}=\cdots=g_{n-k+s-1}=0$. Since $A_s^{(p^\ell)}$ is nonsingular, the last equation of \eqref{eq:common-direct} gives $f_{k-s}=\cdots=f_{k-1}=0$. For $1\le i\le z$, comparison with the first equation of \eqref{eq:common-direct} gives $f(a_i)=0$. Hence $f(x)=c(x)\prod_{i=1}^{z}(x-a_i)$ with $\deg c(x)\le h-1$, and thus $\dim(\operatorname{Hull}_{\ell}(C))\le h$.
		
		Conversely, for every $f(x)=c(x)\prod_{i=1}^{z}(x-a_i)$ with $\deg c(x)\le h-1$, take $g(x)=\lambda x^{\delta-1}f^{p^\ell}(x)$. Then the equations in Proposition~\ref{prop:II.2} are satisfied. Thus $\dim(\operatorname{Hull}_{\ell}(C))\ge h$, and therefore $\dim(\operatorname{Hull}_{\ell}(C))=h$.
	\end{proof}
	
	Actually, the range of $k$ in Proposition~\ref{thm:common-direct} can be enlarged by one possible boundary value if $A_s$ is chosen appropriately. Indeed, the original range corresponds to $\delta-1+p^\ell(k-1)<n-k$, while the only additional value allowed by the enlarged range occurs when $\delta-1+p^\ell(k-1)=n-k$.

	\begin{corollary}\label{cor:common-direct-extended}
		Under the assumptions of Proposition~\ref{thm:common-direct}, for every $s<k\le\left\lfloor\frac{n+p^\ell-\delta+1}{p^\ell+1}\right\rfloor$ and every $0\le h\le k-s$, there exists an $[n+s,k]_q$ generalized Roth-Lempel code $C$ such that $\dim(\operatorname{Hull}_{\ell}(C))=h$. For the additional boundary case $\delta-1+p^\ell(k-1)=n-k$, one may choose $\eta\in(\mathbb F_q^{*})^{p^\ell+1}$ with $\eta\ne\lambda$ and take $A_s\in\operatorname{GL}_s(\mathbb F_q)$ satisfying $A_s^{(p^\ell)}A_s^T=-\eta I_s$.
	\end{corollary}
	
	\begin{proof}
		It remains only to consider the boundary case $\delta-1+p^\ell(k-1)=n-k$. Choose $\eta\in(\mathbb F_q^{*})^{p^\ell+1}$ with $\eta\ne\lambda$. Since $-1\in(\mathbb F_q^{*})^{p^\ell+1}$ by Lemma~\ref{lem:III.2}, one may take $A_s=\mu I_s$ with $\mu^{p^\ell+1}=-\eta$, so that $A_s^{(p^\ell)}A_s^T=-\eta I_s$.
		
		Use the same choice of $\boldsymbol v$, $z$ and $\beta$ as in the proof of Proposition~\ref{thm:common-direct}. The same argument gives $g(x)=\lambda x^{\delta-1}f^{p^\ell}(x)$. In the present case, $g_{n-k}=\lambda f_{k-1}^{p^\ell}$ and $g_{n-k+1}=\cdots=g_{n-k+s-1}=0$. Using $B_s=-T_s(A_s^{-1})^{T}$, the last equation of Proposition~\ref{prop:II.2}, after right multiplication by $A_s^T$, becomes
		\[
		\eta(f_{k-s}^{p^\ell},\ldots,f_{k-1}^{p^\ell})
		=(g_{n-k},0,\ldots,0)T_s.
		\]
		Since the first row of $T_s$ is $(0,\ldots,0,1)$, we obtain $f_{k-s}=\cdots=f_{k-2}=0$ and $(\eta-\lambda)f_{k-1}^{p^\ell}=0$. As $\eta\ne\lambda$, we have $f_{k-s}=\cdots=f_{k-1}=0$. The remainder of the proof is identical to that of Proposition~\ref{thm:common-direct}, and hence $\dim(\operatorname{Hull}_{\ell}(C))=h$.
	\end{proof}
	
	\begin{proposition}\label{thm:common-plus-one-nonzero}
		Assume that $a_j\ne0$ and $a_j^{-1}u_j\in\mathbb F_{p^\ell}^{*}$ for $1\le j\le n$. Suppose that $\frac{n}{p^\ell+1}\in\mathbb Z_{>0}$ and put $k=\xi+1+\frac{n}{p^\ell+1}$. If $0\le \xi\le\min\left\{p^\ell-1,\left\lfloor\frac{\frac{n}{p^\ell+1}-1}{p^\ell}\right\rfloor\right\}$ and $\xi+1\le s\le\min\{\xi+p^\ell,k-1\}$, then there exists an $[n+s,k]_q$ generalized Roth-Lempel code $C$ with a $(k-s)$-dimensional $\ell$-Galois hull.
	\end{proposition}
	
	\begin{proof}
		By Lemma~\ref{lem:III.2}, choose $v_j\in\mathbb F_q^{*}$ such that $v_j^{p^\ell+1}=a_j^{-1}u_j$ for $1\le j\le n$. Choose $\mu\in\mathbb F_q^{*}$ such that $\mu^{p^\ell+1}=-1$, and take
		$$
		A_s=\begin{pmatrix}
			I_{s-\xi-1}&0&O\\
			0&\mu&O\\
			O&-\mu^{-p^\ell}(a_{n-1}^{(1)},\ldots,a_{n-1}^{(\xi)})^T&I_{\xi}
		\end{pmatrix}.
		$$
		where $a_{n-1}^{(j)}$ is defined in \eqref{eq:anj}.
		Clearly, $\det A_s=\mu\ne0$. Moreover, we have
		$$
		A_s^{(p^\ell)}A_s^T=
		\begin{pmatrix}
			I_{s-\xi-1}&0&O\\
			0&-1&-(a_{n-1}^{(1)},\ldots,a_{n-1}^{(\xi)})\\
			O&-((a_{n-1}^{(1)})^{p^\ell},\ldots,(a_{n-1}^{(\xi)})^{p^\ell})^T&
			I_{\xi}-((a_{n-1}^{(1)})^{p^\ell},\ldots,(a_{n-1}^{(\xi)})^{p^\ell})^T
			(a_{n-1}^{(1)},\ldots,a_{n-1}^{(\xi)})
		\end{pmatrix}.
		$$
		
		Let $C=\operatorname{GRL}_k(\boldsymbol a,\boldsymbol v,A_s)$. Consider a codeword represented by $f(x)$ belonging to $\operatorname{Hull}_{\ell}(C)$. By Proposition~\ref{prop:II.2}, there exists $g(x)$ with $\deg g(x)\le n-k+s-1$ such that $f^{p^\ell}(a_j)=a_jg(a_j)$ for $1\le j\le n$. Since $p^\ell(k-1)\le n-1$ and $s\le k-1$, we have
		$
		f^{p^\ell}(x)=xg(x).
		$
		Hence $f_0=0$. Moreover, since $-1+p^\ell(k-\xi-1)=n-k+\xi$ and 
		$
		-1+p^\ell(k-\xi)>n-k+s-1\ge n-k+\xi\ge n-k>-1+p^\ell(k-\xi-2),
		$
		we have $f_{k-\xi}=\cdots=f_{k-1}=0$, and among $g_{n-k},\ldots,g_{n-k+s-1}$ only
		$
		g_{n-k+\xi}=f_{k-\xi-1}^{p^\ell}
		$
		can be nonzero.
		
		Using $B_s=-T_s(A_s^{-1})^T$, the last $s$ coordinates give
		$$
		(f_{k-s}^{p^\ell},\ldots,f_{k-1}^{p^\ell})A_s^{(p^\ell)}A_s^T=-(g_{n-k},\ldots,g_{n-k+s-1})T_s.
		$$
		Since the $(\xi+1)$-st row of $T_s$ is $(0,\ldots,0,1,a_{n-1}^{(1)},\ldots,a_{n-1}^{(\xi)})$, by the choice of $A_s$ we obtain
		$$
		f_{k-s}=\cdots=f_{k-\xi-2}=0,
		$$
		and there is no further restriction on $f_{k-\xi-1}$. Therefore $f_1,\ldots,f_{k-s-1}$ and $f_{k-\xi-1}$ can be chosen freely, and hence $\dim(\operatorname{Hull}_{\ell}(C))\le k-s$.
		
		Conversely, choose $f_1,\ldots,f_{k-s-1}$ and $f_{k-\xi-1}$ arbitrarily and let all the other coefficients be zero. Taking $g(x)=x^{-1}f^{p^\ell}(x)$, we have $\deg g(x)\le n-k+\xi\le n-k+s-1$, and the above equations are satisfied. Therefore $\dim(\operatorname{Hull}_{\ell}(C))=k-s$.
	\end{proof}

	\begin{proposition}\label{thm:common-plus-one-shifted}
		Let $\delta\in\mathbb Z_{>0}$ and $\lambda\in\mathbb F_{p^\ell}^{*}$. Suppose that $\lambda a_j^{\delta-1}u_j\in\mathbb F_{p^\ell}^{*}$ for $1\le j\le n$. Assume that $\frac{n-\delta}{p^\ell+1}\in\mathbb Z_{>0}$ and put $k=\xi+1+\frac{n-\delta}{p^\ell+1}$. If $0\le \xi\le\min\left\{p^\ell-1,\left\lfloor\frac{n-\delta}{p^\ell(p^\ell+1)}\right\rfloor\right\}$ and $\xi+1\le s\le\min\{\xi+p^\ell,k-1\}$, then there exists an $[n+s,k]_q$ generalized Roth-Lempel code $C$ with a $(k-s+1)$-dimensional $\ell$-Galois hull.
	\end{proposition}
	
	\begin{proof}
		By Lemma~\ref{lem:III.2}, choose $v_j\in\mathbb F_q^{*}$ such that $v_j^{p^\ell+1}=\lambda a_j^{\delta-1}u_j$ for $1\le j\le n$. Choose $\mu\in\mathbb F_q^{*}$ such that $\mu^{p^\ell+1}=-\lambda$, and take
		$$
		A_s=\begin{pmatrix}
			I_{s-\xi-1}&0&O\\
			0&\mu&O\\
			O&-\lambda\mu^{-p^\ell}(a_{n-1}^{(1)},\ldots,a_{n-1}^{(\xi)})^T&I_{\xi}
		\end{pmatrix}.
		$$
		Clearly, $\det A_s=\mu\ne0$. Moreover, the $(s-\xi)$-th row of $A_s^{(p^\ell)}A_s^T$ is
		$
		-\lambda(0,\ldots,0,1,a_{n-1}^{(1)},\ldots,a_{n-1}^{(\xi)}),
		$ while its first $s-\xi-1$ rows are the corresponding rows of $I_s$. Let $C=\operatorname{GRL}_k(\boldsymbol a,\boldsymbol v,A_s)$.
		
		Let a codeword represented by $f(x)$ belonging to $\operatorname{Hull}_{\ell}(C)$. By Proposition~\ref{prop:II.2}, there exists $g(x)$ with $\deg g(x)\le n-k+s-1$ such that $\lambda a_j^{\delta-1}f^{p^\ell}(a_j)=g(a_j)$ for $1\le j\le n$. Since $\delta-1+p^\ell(k-1)\le n-1$, we have $
		g(x)=\lambda x^{\delta-1}f^{p^\ell}(x).
		$
		Since $\delta-1+p^\ell(k-\xi-1)=n-k+\xi$ and
		$$
		\delta-1+p^\ell(k-\xi)>n-k+s-1\ge n-k+\xi\ge n-k>\delta-1+p^\ell(k-\xi-2),
		$$
		we have $f_{k-\xi}=\cdots=f_{k-1}=0$, and among $g_{n-k},\ldots,g_{n-k+s-1}$ only $
		g_{n-k+\xi}=\lambda f_{k-\xi-1}^{p^\ell}
		$ can be nonzero.
		
		Using $B_s=-T_s(A_s^{-1})^T$, the last $s$ coordinates give
		$$
		(f_{k-s}^{p^\ell},\ldots,f_{k-1}^{p^\ell})A_s^{(p^\ell)}A_s^T=-(g_{n-k},\ldots,g_{n-k+s-1})T_s.
		$$
		Since the $(\xi+1)$-st row of $T_s$ is $(0,\ldots,0,1,a_{n-1}^{(1)},\ldots,a_{n-1}^{(\xi)})$, by the choice of $A_s$ we obtain $f_{k-s}=\cdots=f_{k-\xi-2}=0$, and there is no further restriction on $f_{k-\xi-1}$. Hence $f_0,\ldots,f_{k-s-1}$ and $f_{k-\xi-1}$ can be chosen freely, and thus $\dim(\operatorname{Hull}_{\ell}(C))\le k-s+1$.
		
		Conversely, choose $f_0,\ldots,f_{k-s-1}$ and $f_{k-\xi-1}$ arbitrarily and let all the other coefficients be zero. Taking $g(x)=\lambda x^{\delta-1}f^{p^\ell}(x)$, we have $\deg g(x)\le n-k+\xi\le n-k+s-1$, and all the above equations are satisfied. Therefore $\dim(\operatorname{Hull}_{\ell}(C))=k-s+1$.
	\end{proof}
	
	\noindent\textbf{Remark}.
	If $a_{n-1}^{(1)}=\cdots=a_{n-1}^{(\xi)}=0$, one may also take $A_s=\mu I_s$ in Propositions~\ref{thm:common-plus-one-nonzero} and~\ref{thm:common-plus-one-shifted}, with $\mu^{p^\ell+1}=-1$ and $-\lambda$, respectively. The same coefficient comparison gives the stated hull dimensions. This condition is empty when $\xi=0$, and holds for the full evaluation set $\mathbb F_q$ in the ranges used below.

	\subsection{Multiplicative-type constructions}
	
	We first consider three families of nonzero evaluation sets. In each
	case, the associated coefficients satisfy
	$a_i^{-1}u_i\in\mathbb F_{p^\ell}^{*}$. Adjoining the evaluation point
	$0$ leads to a companion family for which the modified coefficients
	$w_i$ belong to $\mathbb F_{p^\ell}^{*}$.
	
	\subsubsection*{Construction A}\mbox{}\par
	\noindent Let $q=p^e$ be a prime power and assume that $\ell\mid e$.
	Consider the norm mapping
	\begin{equation}\label{eq:norm}
		\operatorname{Norm}:\mathbb F_q^{*}\longrightarrow
		\mathbb F_{p^\ell}^{*},\qquad
		x\longmapsto x^{\frac{q-1}{p^\ell-1}}.
	\end{equation}
	Label the elements of $\mathbb F_{p^\ell}^{*}$ as
	$b_1,b_2,\ldots,b_{p^\ell-1}$ and define $\mathbf{N}_j=\{x\in\mathbb F_q^{*}:\operatorname{Norm}(x)=b_j\}$, $1\le j\le p^\ell-1$.
	The norm mapping is surjective and
	$|\mathbf{N}_j|=(q-1)/(p^\ell-1)$. Let
	$1\le t\le p^\ell-1$ and set
	\begin{equation}\label{eq:Nset}
		\mathbf{N}=\bigcup_{j=1}^{t}\mathbf{N}_j=\{a_1,a_2,\ldots,a_n\},
		\qquad n=t\frac{q-1}{p^\ell-1}.
	\end{equation}
	
	\begin{lemma}\label{lem:3.4}\cite{9559992}
		Let $a_i$ and $u_i$ be defined by Equations~\eqref{eq:Nset}
		and~\eqref{eq:ui}, respectively. If $a_i\in N_j$, then
		\begin{equation}\label{eq:ui_norm}
			u_i=a_i^{1-\frac{q-1}{p^\ell-1}}
			\prod_{\substack{1\le j'\le t\\j'\ne j}}
			(b_j-b_{j'})^{-1}.
		\end{equation}
		Moreover, for every nonnegative integer $r$, $a_i^{r\frac{q-1}{p^\ell-1}-1}u_i
		\in\mathbb F_{p^\ell}^{*}$, $1\le i\le n$.
		In particular, $a_i^{-1}u_i\in\mathbb F_{p^\ell}^{*}$.
	\end{lemma}
	
	We first construct generalized Roth-Lempel codes with arbitrary
	$\ell$-Galois hull dimensions in this range.
	
	\begin{theorem}\label{thm:II.1}
		Let $q=p^e>4$, where $p$ is a prime, and assume that
		$2\ell\mid e$. Let
		$n=t\frac{q-1}{p^\ell-1}$, $1\le t\le p^\ell-1$.
		Then, for every
		$s+1\le k\le
		\left\lfloor
		\frac{p^\ell+n}{p^\ell+1}
		\right\rfloor$
		and every $0\le h\le k-s-1$, there exists an $[n+s,k]_q$
		generalized Roth-Lempel code $C$ such that
		\[
		\dim\bigl(\operatorname{Hull}_{\ell}(C)\bigr)=h.
		\]
	\end{theorem}
	
	\begin{proof}
		By Lemma~\ref{lem:3.4},
		$a_i^{-1}u_i\in\mathbb F_{p^\ell}^{*}$ for $1\le i\le n$.
		Therefore, all the hypotheses of Proposition~\ref{thm:common-nonzero}
		are satisfied, and the result follows.
	\end{proof}
	
	We next adjoin the evaluation point $0$. Put $a_{n+1}=0$ and define
	$w_i=\prod_{\substack{1\le j\le n+1\\j\ne i}}
	(a_i-a_j)^{-1}$, $1\le i\le n+1$.
	
	\begin{lemma}\label{lem:3.6}\cite{9559992}
		For $1\le i\le n$, we have
		\[
		w_i
		=a_i^{-1}
		\prod_{\substack{1\le j\le n\\j\ne i}}
		(a_i-a_j)^{-1}
		=a_i^{-1}u_i
		\in\mathbb F_{p^\ell}^{*}.
		\]
		Moreover,
		\[
		w_{n+1}
		=\prod_{j=1}^{n}(a_{n+1}-a_j)^{-1}
		=(-1)^{n+t\frac{q-p^\ell}{p^\ell-1}}
		\prod_{j=1}^{t}b_j^{-1}
		\in\mathbb F_{p^\ell}^{*}.
		\]
	\end{lemma}
	
	\begin{theorem}\label{thm:II.2}
		Let $q=p^e>4$, where $p$ is a prime, and assume that
		$2\ell\mid e$. Let $n=t\frac{q-1}{p^\ell-1}$, $1\le t\le p^\ell-1$.
		Then, for every
		$s<k\le
		\left\lfloor
		\frac{p^\ell+n+1}{p^\ell+1}
		\right\rfloor$ and every $0\le h\le k-s$, there exists an $[n+s+1,k]_q$
		generalized Roth-Lempel code $C$ such that
		\[
		\dim\bigl(\operatorname{Hull}_{\ell}(C)\bigr)=h.
		\]
	\end{theorem}
	
	\begin{proof}
		Let $N=n+1$ be the total number of evaluation points. By
		Lemma~\ref{lem:3.6}, $w_i\in\mathbb F_{p^\ell}^{*}$ for
		$1\le i\le N$. The stated bound on $k$ is precisely
		$k\le\lfloor(p^\ell+N)/(p^\ell+1)\rfloor$. The result now follows
		from Corollary~\ref{cor:common-direct-extended} applied with total length $N$,
		$\delta=1$, and $\lambda=1$.
	\end{proof}
	
	Moreover, a $(k-s)$-dimensional hull can be obtained for a special
	choice of $k$ and $A_s$.
	
	\begin{theorem}\label{thm:II.444}
		Let $q=p^e>4$, where $p$ is a prime, and assume that
		$2\ell\mid e$. Let
		$
		n=t\frac{q-1}{p^\ell-1},~1\le t\le p^\ell-1.
		$
		If
		$k=\xi+1+\frac{n}{p^\ell+1}$, 
		where
		\[
		0\le \xi\le
		\min\left\{
		p^\ell-1,\,
		\left\lfloor
		\frac{\frac{n}{p^\ell+1}-1}{p^\ell}
		\right\rfloor
		\right\},
		\qquad
		\xi+1\le s\le\min\{\xi+p^\ell,k-1\},
		\]
		then there exists an $[n+s,k]_q$ generalized Roth-Lempel code
		$C$ with a $(k-s)$-dimensional $\ell$-Galois hull.
	\end{theorem}
	
	\begin{proof}
		By Lemma~\ref{lem:3.4},
		$a_i^{-1}u_i\in\mathbb F_{p^\ell}^{*}$. Since $2\ell\mid e$,
		$n/(p^\ell+1)$ is a positive integer. The stated bounds on $\xi$ and
		$s$ verify the remaining hypotheses of
		Proposition~\ref{thm:common-plus-one-nonzero}. Hence the desired code
		exists.
	\end{proof}
	
	The zero-extended construction also gives a $(k-s+1)$-dimensional
	$\ell$-Galois hull for a special value of $k$.
	
	\begin{theorem}\label{thm:II.4}
		Let $q=p^e>4$, where $p$ is a prime, and assume that
		$2\ell\mid e$. Let $n=t\frac{q-1}{p^\ell-1}$, $1\le t\le p^\ell-1$. If $k=\xi+1+\frac{n}{p^\ell+1}$, where $0\le \xi\le
		\min\left\{
		p^\ell-1,\,
		\left\lfloor
		\frac{n}{p^\ell(p^\ell+1)}
		\right\rfloor
		\right\}$, and
		$\xi+1\le s\le\min\{\xi+p^\ell,k-1\}$, then there exists an $[n+s+1,k]_q$ generalized Roth-Lempel code
		$C$ with a $(k-s+1)$-dimensional $\ell$-Galois hull.
	\end{theorem}
	
	\begin{proof}
		Let $N=n+1$. Lemma~\ref{lem:3.6} gives
		$w_i\in\mathbb F_{p^\ell}^{*}$ for $1\le i\le N$. Taking
		$\delta=1$ and $\lambda=1$ in
		Proposition~\ref{thm:common-plus-one-shifted}, we have
		$(N-\delta)/(p^\ell+1)=n/(p^\ell+1)$, and the stated bounds on
		$\xi$ and $s$ verify the remaining hypotheses. Hence the result
		follows.
	\end{proof}
	\subsubsection*{Construction B}\mbox{}\par
	\noindent Let $q=p^e$ be a prime power and let $\omega$ be a primitive
	element of $\mathbb F_q$. Let $\xi_1=\omega^{x_1}$ and $\xi_2=\omega^{x_2}$,
	where $x_1,x_2$ are positive integers, and put
	$r_2=\operatorname{ord}(\xi_2)=\frac{q-1}{\gcd(x_2,q-1)}$.
	For $1\le j\le r_1$, define $R_j=\{\xi_1^j\xi_2^v:1\le v\le r_2\}$.
	Assume that $(q-1)\mid\operatorname{lcm}(x_1,x_2)$.
	Then the sets $R_j$ are pairwise disjoint whenever
	$1\le r_1\le\frac{q-1}{\gcd(x_1,q-1)}$. Put
	$R=\bigcup_{j=1}^{r_1}R_j=\{a_1,a_2,\ldots,a_n\}$, where $n=r_1r_2$.
	
	\begin{lemma}\label{lem:3.8}\cite{9559992}
		Let $u_i$ be defined by Equation~\eqref{eq:ui} and suppose that
		$a_i\in R_j$. Then
		\begin{equation}\label{eq:B-lagrange-auto}
			u_i=a_i\xi_1^{-jr_2}r_2^{-1}
			\prod_{\substack{1\le j'\le r_1\\j'\ne j}}
			\left(\xi_1^{jr_2}-\xi_1^{j'r_2}\right)^{-1}.
		\end{equation}
		If $\frac{q-1}{p^\ell-1}\mid x_1$, then, for every
		$\tau\in\mathbb Z_{\ge0}$, $a_i^{\tau r_2-1}u_i\in\mathbb F_{p^\ell}^{*}$,
		$1\le i\le n$.
		In particular, $a_i^{-1}u_i\in\mathbb F_{p^\ell}^{*}$.
	\end{lemma}
	
	\begin{theorem}\label{thm:II.14}
		Let $q=p^e>4$, where $p$ is a prime, and assume that
		$2\ell\mid e$. Let $x_1,x_2$ be positive integers satisfying
		$(q-1)\mid\operatorname{lcm}(x_1,x_2)$ and
		$\frac{q-1}{p^\ell-1}\mid x_1$. Let
		$n=r_1\frac{q-1}{\gcd(x_2,q-1)}$, where
		$1\le r_1\le\frac{q-1}{\gcd(x_1,q-1)}$.
		Then, for every $s+1\le k\le
		\left\lfloor\frac{p^\ell+n}{p^\ell+1}\right\rfloor$
		and every $0\le h\le k-s-1$, there exists an $[n+s,k]_q$
		generalized Roth-Lempel code $C$ such that
		\[
		\dim\bigl(\operatorname{Hull}_{\ell}(C)\bigr)=h.
		\]
	\end{theorem}
	
	\begin{proof}
		Lemma~\ref{lem:3.8} gives
		$a_i^{-1}u_i\in\mathbb F_{p^\ell}^{*}$ for $1\le i\le n$.
		The result follows immediately from
		Proposition~\ref{thm:common-nonzero}.
	\end{proof}
	We next adjoin the evaluation point $0$. Put $a_{n+1}=0$ and define
	$w_i=\prod_{\substack{1\le j\le n+1\\j\ne i}}
	(a_i-a_j)^{-1}$, $1\le i\le n+1$.
	For $1\le i\le n$, we have $w_i=a_i^{-1}u_i\in\mathbb F_{p^\ell}^{*}$.
	Moreover,
	\begin{equation}\label{eq:B-zero-lagrange-auto}
		w_{n+1}
		=(-1)^n
		\xi_1^{-\frac{r_1(r_1+1)r_2}{2}}
		\xi_2^{-\frac{r_2(r_2+1)r_1}{2}}
		\in\mathbb F_{p^\ell}^{*}.
	\end{equation}
	
	\begin{theorem}\label{thm:II.15}
		Under the field and evaluation-set assumptions of Theorem~\ref{thm:II.14}, for every
		$s<k\le\left\lfloor\frac{p^\ell+n+1}{p^\ell+1}\right\rfloor$
		and every $0\le h\le k-s$, there exists an $[n+s+1,k]_q$
		generalized Roth-Lempel code $C$ such that
		\[
		\dim\bigl(\operatorname{Hull}_{\ell}(C)\bigr)=h.
		\]
	\end{theorem}
	
	\begin{proof}
		Let $N=n+1$. The identities preceding the theorem show that
		$w_i\in\mathbb F_{p^\ell}^{*}$ for $1\le i\le N$. The stated
		bound on $k$ is precisely
		$k\le\lfloor(p^\ell+N)/(p^\ell+1)\rfloor$. Hence the result follows
		from Corollary~\ref{cor:common-direct-extended} with $\delta=1$ and
		$\lambda=1$.
	\end{proof}
	
	Finally, the same construction yields a $(k-s)$-dimensional hull
	for a special value of $k$.
	
	\begin{theorem}\label{thm:B-hull-plus-one}
		Under the field and evaluation-set assumptions of Theorem~\ref{thm:II.14}, assume that
		$\frac{n}{p^\ell+1}\in\mathbb Z_{>0}$ and put
		$k=\xi+1+\frac{n}{p^\ell+1}$. If
		\[
		0\le \xi\le
		\min\left\{
		p^\ell-1,\,
		\left\lfloor
		\frac{\frac{n}{p^\ell+1}-1}{p^\ell}
		\right\rfloor
		\right\},
		\qquad
		\xi+1\le s\le\min\{\xi+p^\ell,k-1\},
		\]
		then there exists an $[n+s,k]_q$ generalized Roth-Lempel code
		$C$ such that
		\[
		\dim\bigl(\operatorname{Hull}_{\ell}(C)\bigr)=k-s.
		\]
	\end{theorem}
	
	\begin{proof}
		By Lemma~\ref{lem:3.8},
		$a_j^{-1}u_j\in\mathbb F_{p^\ell}^{*}$ for $1\le j\le n$.
		The result follows from
		Proposition~\ref{thm:common-plus-one-nonzero}.
	\end{proof}
	
	The zero-extended construction gives the same additional hull
	dimension under a slightly different bound on $i$.
	
	\begin{theorem}\label{thm:B-zero-hull-plus-one}
		Under the field and evaluation-set assumptions of Theorem~\ref{thm:II.14}, put
		$a_{n+1}=0$. Assume that
		$\frac{n}{p^\ell+1}\in\mathbb Z_{>0}$ and let
		$k=\xi+1+\frac{n}{p^\ell+1}$. If
		\[
		0\le \xi\le
		\min\left\{
		p^\ell-1,\,
		\left\lfloor
		\frac{n}{p^\ell(p^\ell+1)}
		\right\rfloor
		\right\},
		\qquad
		\xi+1\le s\le\min\{\xi+p^\ell,k-1\},
		\]
		then there exists an $[n+s+1,k]_q$ generalized Roth-Lempel
		code $C$ such that
		\[
		\dim\bigl(\operatorname{Hull}_{\ell}(C)\bigr)=k-s+1.
		\]
	\end{theorem}
	
	\begin{proof}
		Let $N=n+1$. The zero-extended coefficients satisfy
		$w_j\in\mathbb F_{p^\ell}^{*}$ for $1\le j\le N$. Taking
		$\delta=1$ and $\lambda=1$ in
		Proposition~\ref{thm:common-plus-one-shifted} gives
		$(N-\delta)/(p^\ell+1)=n/(p^\ell+1)$, and the result follows.
	\end{proof}
	
	\subsubsection*{Construction C}\mbox{}\par
	\noindent Let $q=p^e$ be a prime power and assume that $\ell\mid e$. Put
	$y=\frac{q-1}{p^\ell-1}$. Let $m\mid(q-1)$ and write
	$m_2=\gcd(m,y)$ and $m_1=\frac{m}{m_2}$. Let $\omega$ be a primitive
	element of $\mathbb F_q$ and define
	$\vartheta_1=\omega^{\frac{q-1}{m}}$ and
	$\vartheta_2=\omega^{\frac{y}{m_2}}$. Set
	$H=\langle\vartheta_1\rangle$ and $G=\langle\vartheta_2\rangle$.
	Then $|H|=m$, $|G|=(p^\ell-1)m_2$, and $H$ is a subgroup of $G$.
	Hence $|G/H|=\frac{p^\ell-1}{m_1}$. Choose $r$ distinct cosets
	$\eta_1H,\ldots,\eta_rH$, where $1\le r\le\frac{p^\ell-1}{m_1}$,
	and put
	\begin{equation}\label{eq:C-evaluation-auto}
		\mathcal H=\bigcup_{j=1}^{r}\eta_jH
		=\{a_1,a_2,\ldots,a_n\},\qquad n=rm.
	\end{equation}

	\begin{lemma}\label{lem:3.12}\cite{9559992}
		Let $u_i$ be defined by Equation~\eqref{eq:ui} and suppose that
		$a_i\in\eta_jH$. Then
		\begin{equation}\label{eq:C-lagrange-auto}
			u_i=a_i\eta_j^{-m}m^{-1}
			\prod_{\substack{1\le j'\le r\\j'\ne j}}
			(\eta_j^m-\eta_{j'}^m)^{-1}.
		\end{equation}
		For every $\tau\in\mathbb Z_{\ge0}$,
		$a_i^{\tau m-1}u_i\in\mathbb F_{p^\ell}^{*}$.
		In particular, $a_i^{-1}u_i\in\mathbb F_{p^\ell}^{*}$.
	\end{lemma}
	
	\begin{theorem}\label{thm:II.18}
		Let $q=p^e>4$, where $p$ is a prime, and assume that
		$2\ell\mid e$. Keep the notation above and let $n=rm$, where
		$1\le r\le\frac{p^\ell-1}{m_1}$.
		Then, for every $s+1\le k\le
		\left\lfloor\frac{p^\ell+n}{p^\ell+1}\right\rfloor$
		and every $0\le h\le k-s-1$, there exists an $[n+s,k]_q$
		generalized Roth-Lempel code $C$ such that
		\[
		\dim\bigl(\operatorname{Hull}_{\ell}(C)\bigr)=h.
		\]
	\end{theorem}
	
	\begin{proof}
		Lemma~\ref{lem:3.12} gives
		$a_i^{-1}u_i\in\mathbb F_{p^\ell}^{*}$ for $1\le i\le n$.
		The result follows immediately from
		Proposition~\ref{thm:common-nonzero}.
	\end{proof}
	
	We next adjoin the evaluation point $0$. Put $a_{n+1}=0$ and define
	$w_i=\prod_{\substack{1\le j\le n+1\\j\ne i}}
	(a_i-a_j)^{-1}$, $1\le i\le n+1$.
	For $1\le i\le n$, we have $w_i=a_i^{-1}u_i\in\mathbb F_{p^\ell}^{*}$.
	Moreover,
	\begin{equation}\label{eq:C-zero-lagrange-auto}
		w_{n+1}
		=(-1)^n
		\vartheta_1^{-\frac{rm(m+1)}{2}}
		\prod_{i=1}^{r}\eta_i^{-m}
		\in\mathbb F_{p^\ell}^{*}.
	\end{equation}
	
	\begin{theorem}\label{thm:II.19}
		Under the field and evaluation-set assumptions of Theorem~\ref{thm:II.18}, for every
		$s<k\le\left\lfloor\frac{p^\ell+n+1}{p^\ell+1}\right\rfloor$
		and every $0\le h\le k-s$, there exists an $[n+s+1,k]_q$
		generalized Roth-Lempel code $C$ such that
		\[
		\dim\bigl(\operatorname{Hull}_{\ell}(C)\bigr)=h.
		\]
	\end{theorem}
	
	\begin{proof}
		Let $N=n+1$. The identities preceding the theorem show that
		$w_i\in\mathbb F_{p^\ell}^{*}$ for $1\le i\le N$. The stated
		bound on $k$ is precisely
		$k\le\lfloor(p^\ell+N)/(p^\ell+1)\rfloor$. Hence the result follows
		from Corollary~\ref{cor:common-direct-extended} with $\delta=1$ and
		$\lambda=1$.
	\end{proof}
	
	Finally, Construction C also yields a $(k-s)$-dimensional hull.
	
	\begin{theorem}\label{thm:C-hull-plus-one}
		Under the field and evaluation-set assumptions of Theorem~\ref{thm:II.18}, assume that $\frac{n}{p^\ell+1}\in\mathbb Z_{>0}$ and let
		$k=\xi+1+\frac{n}{p^\ell+1}$. If
		\[
		0\le \xi\le
		\min\left\{
		p^\ell-1,\,
		\left\lfloor
		\frac{\frac{n}{p^\ell+1}-1}{p^\ell}
		\right\rfloor
		\right\},
		\qquad
		\xi+1\le s\le\min\{\xi+p^\ell,k-1\},
		\]
		then there exists an $[n+s,k]_q$ generalized Roth-Lempel code
		$C$ such that
		\[
		\dim\bigl(\operatorname{Hull}_{\ell}(C)\bigr)=k-s.
		\]
	\end{theorem}
	
	\begin{proof}
		By Lemma~\ref{lem:3.12},
		$a_j^{-1}u_j\in\mathbb F_{p^\ell}^{*}$ for $1\le j\le n$.
		The result follows from
		Proposition~\ref{thm:common-plus-one-nonzero}.
	\end{proof}
	
	The zero-extended construction gives the same additional hull
	dimension under the corresponding shifted parameter bound.
	
	\begin{theorem}\label{thm:C-zero-hull-plus-one}
		Under the field and evaluation-set assumptions of Theorem~\ref{thm:II.18}, put
		$a_{n+1}=0$. Assume that $\frac{n}{p^\ell+1}\in\mathbb Z_{>0}$ and let
		$k=\xi+1+\frac{n}{p^\ell+1}$. If
		\[
		0\le \xi\le
		\min\left\{
		p^\ell-1,\,
		\left\lfloor
		\frac{n}{p^\ell(p^\ell+1)}
		\right\rfloor
		\right\},
		\qquad
		\xi+1\le s\le\min\{\xi+p^\ell,k-1\},
		\]
		then there exists an $[n+s+1,k]_q$ generalized Roth-Lempel
		code $C$ such that
		\[
		\dim\bigl(\operatorname{Hull}_{\ell}(C)\bigr)=k-s+1.
		\]
	\end{theorem}
	
	\begin{proof}
		Let $N=n+1$. The zero-extended coefficients satisfy
		$w_j\in\mathbb F_{p^\ell}^{*}$ for $1\le j\le N$. Taking
		$\delta=1$ and $\lambda=1$ in
		Proposition~\ref{thm:common-plus-one-shifted} gives
		$(N-\delta)/(p^\ell+1)=n/(p^\ell+1)$, and the result follows.
	\end{proof}
	
	\subsection{Additive-type constructions}
	
	We next turn to an additive evaluation set defined by trace fibers.
	Here the Lagrange coefficients already satisfy
	$u_i\in\mathbb F_{p^\ell}^{*}$, so the common argument applies without
	the factor $a_i^{-1}$ and yields the polynomial identity
	$f^{p^\ell}(x)=g(x)$ directly.
	
	\subsubsection*{Construction D}\mbox{}\par
	\noindent Let $q=p^e$ be a prime power and assume that $\ell\mid e$.
	Let $B\subseteq\mathbb F_{p^\ell}$ be an $\mathbb F_p$-linear
	subspace and choose distinct elements $b_1=0,b_2,\ldots,b_t$ of
	$B$, where $1\le t\le|B|$. Consider the trace mapping
	\begin{equation}\label{eq:trace}
		\operatorname{Tr}:\mathbb F_q\longrightarrow\mathbb F_{p^\ell},
		\qquad
		x\longmapsto x+x^{p^\ell}+\cdots+x^{p^{e-\ell}}.
	\end{equation}
	For $1\le j\le t$, define
	\[
	T_j=\{x\in\mathbb F_q:\operatorname{Tr}(x)=b_j\}.
	\]
	Then $|T_j|=p^{e-\ell}$ and the sets $T_j$ are pairwise disjoint.
	Put
	\begin{equation}\label{eq:Tset}
		T=\bigcup_{j=1}^{t}T_j=\{a_1,a_2,\ldots,a_n\},
		\qquad n=tp^{e-\ell},\qquad 1\le t\le|B|\le p^\ell.
	\end{equation}

	\begin{lemma}\label{lem:3.1}\cite{li2023several}
		Let $a_i$ and $u_i$ be defined by Equations~\eqref{eq:Tset}
		and~\eqref{eq:ui}, respectively. If $a_i\in T_{j_0}$, then
		\begin{equation}\label{eq:ui_trace}
			u_i=
			\begin{cases}
				1,&t=1,\\[2mm]
				\displaystyle
				\prod_{\substack{1\le j\le t\\j\ne j_0}}
				(b_{j_0}-b_j)^{-1},&2\le t\le p^\ell.
			\end{cases}
		\end{equation}
		In particular, $u_i\in\mathbb F_{p^\ell}^{*}$ for
		$1\le i\le n$.
	\end{lemma}
	
	\begin{theorem}\label{thm:II.t1}
		Let $q=p^e>4$, where $p$ is a prime, and assume that
		$2\ell\mid e$. Let
		$n=tp^{e-\ell}$, $1\le t\le|B|\le p^\ell$.
		Then, for every
		$
		s<k\le
		\left\lfloor\frac{p^\ell+n}{p^\ell+1}\right\rfloor
		$
		and every $0\le h\le k-s$, there exists an $[n+s,k]_q$
		generalized Roth-Lempel code $C$ such that
		\[
		\dim\bigl(\operatorname{Hull}_{\ell}(C)\bigr)=h.
		\]
	\end{theorem}
	
	\begin{proof}
		By Lemma~\ref{lem:3.1}, $u_i\in\mathbb F_{p^\ell}^{*}$ for
		$1\le i\le n$. Therefore, the result follows directly from
		Corollary~\ref{cor:common-direct-extended} with $\delta=1$ and $\lambda=1$.
	\end{proof}
	
	Moreover, the same construction yields a $(k-s+1)$-dimensional hull
	for a special value of $k$.
	
	\begin{theorem}\label{thm:D-hull-plus-one}
		Under the field and evaluation-set assumptions of Theorem~\ref{thm:II.t1}, assume that
		$t=p^\ell$, and hence $B=\mathbb F_{p^\ell}$ and $n=q$.
		Let $k=\xi+1+\frac{n-1}{p^\ell+1}$, where $0\le \xi\le
		\min\left\{
		p^\ell-1,\,
		\left\lfloor\frac{n-1}{p^\ell(p^\ell+1)}\right\rfloor
		\right\}$, $\xi+1\le s\le\min\{\xi+p^\ell,k-1\}$. Then there exists an $[n+s,k]_q$ generalized Roth-Lempel code
		$C$ such that
		\[
		\dim\bigl(\operatorname{Hull}_{\ell}(C)\bigr)=k-s+1.
		\]
	\end{theorem}
	
	\begin{proof}
		Lemma~\ref{lem:3.1} gives $u_i\in\mathbb F_{p^\ell}^{*}$.
		Since $2\ell\mid e$, we have
		$(n-1)/(p^\ell+1)=(q-1)/(p^\ell+1)\in\mathbb Z_{>0}$.
		Taking $\delta=1$ and $\lambda=1$ in
		Proposition~\ref{thm:common-plus-one-shifted}, the stated bounds on
		$\xi$ and $s$ verify all its remaining hypotheses. Hence the result
		follows.
	\end{proof}
	
	\subsubsection*{Construction E: additive-subspace cosets}\mbox{}\par
	\noindent Construction D can be extended by replacing the kernel of the trace
	map with an arbitrary additive subspace. 
	\begin{definition}
		Let $V$ be an $r$-dimensional
		$\mathbb F_p$-subspace of $\mathbb F_q$, and define its subspace
		polynomial by
		\[
		L_V(x)=\prod_{v\in V}(x-v).
		\]
	\end{definition}
	By \cite{Lidl_Niederreiter_1996}, $L_V(x)$ is a $p$-linearized polynomial of degree $p^r$ in $\mathbb{F}_{q}[x]$ and can be written as
	$L_V(x)=\sum_{i=0}^{r}c_{i}x^{p^i}$. Its formal derivative is the nonzero constant
	$c_V:=c_{0}=L_V'(x)=\prod_{v\in V\setminus\{0\}}(-v)\in\mathbb F_q^{*}$.

	Choose distinct elements
	$b_1,\ldots,b_t\in\operatorname{Im}(L_V)\cap\mathbb F_{p^\ell}$,
	and choose $\beta_j\in\mathbb F_q$ such that
	$L_V(\beta_j)=b_j$, where
	$1\le t\le|\operatorname{Im}(L_V)\cap\mathbb F_{p^\ell}|$.
	For $1\le j\le t$, put
	\[
	\begin{aligned}
		E_j&=\beta_j+V=\{x\in\mathbb F_q:L_V(x)=b_j\},\\
		E&=\bigcup_{j=1}^{t}E_j=\{a_1,\ldots,a_n\},
		\qquad n=tp^r.
	\end{aligned}
	\]
	The sets $E_j$ are pairwise disjoint because the $b_j$ are distinct.
	
	\begin{lemma}\label{lem:additive-subspace}
		Let $V$ be an $r$-dimensional
		$\mathbb F_p$-subspace of $\mathbb F_q$ and let $u_i$ be defined by Equation~\eqref{eq:ui}. If
		$a_i\in E_{j_0}$ for some $j_{0}\in\{1,\cdots,t\}$, then
		\begin{equation}\label{eq:ui-additive-subspace}
			u_i=c_V^{-1}
			\prod_{\substack{1\le j\le t\\j\ne j_0}}
			(b_{j_0}-b_j)^{-1}.
		\end{equation}
		Here $c_V=L_V'(x)$ is the formal derivative of the subspace
		polynomial $L_V(x)$. Consequently, $c_Vu_i\in\mathbb F_{p^\ell}^{*}$ for
		$1\le i\le n$.
	\end{lemma}
	
	\begin{proof}
		Since $L_V$ is $p$-linearized, we have
		$L_V(x+\beta_j)=L_V(x)+L_V(\beta_j)$. Hence, by definition $E_j$ is exactly the
		root set of $L_V(x)-b_j$, and the monic vanishing polynomial of
		$E=\bigcup_{j=1}^{t}E_j$ is
		\[
		P_E(x):=\prod_{a\in E}(x-a)=\prod_{j=1}^{t}\prod_{a\in E_{j}}(x-a)
		=\prod_{j=1}^{t}\bigl(L_V(x)-b_j\bigr).
		\]
		If $a_i\in E_{j_0}$, differentiation at $a_i$ gives
		\[
		u_i^{-1}=\prod_{1\le j\ne i\le n}(a_{i}-a_{j})=P_E'(a_i)
		=c_V\prod_{\substack{1\le j\le t\\j\ne j_0}}
		(b_{j_0}-b_j).
		\]
		Equation~\eqref{eq:ui-additive-subspace} follows. Since all the
		$b_j$ belong to $\mathbb F_{p^\ell}$ and are distinct, the final
		assertion is immediate.
	\end{proof}
	\noindent\textbf{Remark}.
	Taking $V=\ker(\operatorname{Tr})$ gives
	$L_V(x)=\operatorname{Tr}(x)$ and $c_V=1$, so Construction D is the
	corresponding special case of Construction E.
	
	\begin{theorem}\label{thm:additive-subspace-hull}
		Let $q=p^e>4$, where $p$ is a prime, and assume that
		$2\ell\mid e$. Keep the notation and assumptions of Construction E,
		so that $n=tp^r$. Then, for every
		$s<k\le\left\lfloor\frac{p^\ell+n}{p^\ell+1}\right\rfloor$ and
		every $0\le h\le k-s$, there exists an $[n+s,k]_q$ generalized
		Roth-Lempel code $C$ such that
		$\dim(\operatorname{Hull}_{\ell}(C))=h$.
	\end{theorem}
	
	\begin{proof}
		Lemma~\ref{lem:additive-subspace} gives
		$c_Vu_i\in\mathbb F_{p^\ell}^{*}$ for $1\le i\le n$.
		Taking $\delta=1$ and $\lambda=c_V$ in
		Corollary~\ref{cor:common-direct-extended} gives the result.
	\end{proof}
	
	Moreover, under an additional condition on $c_V$, Construction E
	also yields a $(k-s+1)$-dimensional hull.
	
	\begin{theorem}\label{thm:additive-subspace-plus-one}
		Under the field and evaluation-set assumptions of Theorem~\ref{thm:additive-subspace-hull},
		suppose that $c_V\in\mathbb F_{p^\ell}^{*}$ and
		$\frac{n-1}{p^\ell+1}\in\mathbb Z_{>0}$. Put
		$k=\xi+1+\frac{n-1}{p^\ell+1}$. If
		\[
		0\le \xi\le
		\min\left\{p^\ell-1,
		\left\lfloor\frac{n-1}{p^\ell(p^\ell+1)}\right\rfloor\right\},
		\qquad
		\xi+1\le s\le\min\{\xi+p^\ell,k-1\},
		\]
		then there exists an $[n+s,k]_q$ generalized Roth-Lempel code
		$C$ with a $(k-s+1)$-dimensional $\ell$-Galois hull.
	\end{theorem}
	
	\begin{proof}
		By Lemma~\ref{lem:additive-subspace},
		$c_Vu_i\in\mathbb F_{p^\ell}^{*}$ for $1\le i\le n$.
		Taking $\delta=1$ and $\lambda=c_V$ in
		Proposition~\ref{thm:common-plus-one-shifted} gives the result.
	\end{proof}
	
	\subsection{Mixed additive--multiplicative constructions}
	
	\subsubsection*{Construction F: mixed fibers of a subspace polynomial}\mbox{}\par
	\noindent Let $q=p^e$ be a prime power and assume that $\ell\mid e$. Let
	$V$ be an $r$-dimensional $\mathbb F_p$-subspace of
	$\mathbb F_{p^\ell}$, where $0\le r<\ell$, and put
	\[
	L_V(x)=\prod_{v\in V}(x-v).
	\]
	Then $L_V(x)$ is a $p$-linearized polynomial of degree $p^r$ with
	kernel $V$, and
	\[
	c_V:=L_V'(x)=\prod_{v\in V\setminus\{0\}}(-v)
	\in\mathbb F_{p^\ell}^{*}.
	\]
	In particular, $L_V$ induces an $\mathbb F_p$-linear map on
	$\mathbb F_{p^\ell}$ with
	$|L_V(\mathbb F_{p^\ell})|=p^{\ell-r}$.
	
	Let $m\mid(q-1)/(p^\ell-1)$ and put $\Phi(x)=L_V(x^m)$. Then
	$m\mid(q-1)$ and hence $p\nmid m$. If $\omega$ is a primitive
	element of $\mathbb F_q$, then
	\[
	\mathbb F_{p^\ell}^{*}
	=\left\langle\omega^{\frac{q-1}{p^\ell-1}}\right\rangle
	\subseteq\langle\omega^m\rangle
	=(\mathbb F_q^{*})^m.
	\]
	Thus, for every $c\in\mathbb F_{p^\ell}^{*}$, the polynomial
	$x^m-c$ splits completely over $\mathbb F_q$ with $m$ distinct
	roots.
	
	Choose distinct
	$b_1,\ldots,b_t\in L_V(\mathbb F_{p^\ell})\setminus\{0\}$, where
	$1\le t\le p^{\ell-r}-1$, and choose $\beta_j\in\mathbb F_{p^\ell}$
	such that $L_V(\beta_j)=b_j$. Since $L_V$ is $p$-linearized,
	\[
	\{y\in\mathbb F_{p^\ell}:L_V(y)=b_j\}=\beta_j+V.
	\]
	Moreover, $b_j\ne0$ implies
	$\beta_j+V\subseteq\mathbb F_{p^\ell}^{*}$. Therefore
	\[
	\Phi(x)-b_j
	=\prod_{v\in V}\bigl(x^m-(\beta_j+v)\bigr)
	\]
	The elements $\beta_j+v$, $v\in V$, are pairwise distinct, so
	the corresponding root sets are disjoint. Hence
	$\Phi(x)-b_j$ splits completely over $\mathbb F_q$ with $mp^r$
	distinct roots.
	
	For $1\le j\le t$, let
	$M_j=\{x\in\mathbb F_q:\Phi(x)=b_j\}$. Then the sets $M_j$ are
	pairwise disjoint and $|M_j|=mp^r$. Put
	\[
	M=\bigcup_{j=1}^tM_j=\{a_1,\ldots,a_n\},
	\qquad n=tmp^r.
	\]
	
	\begin{lemma}\label{lem:mixed-fibers}
		Let $V$ be an $r$-dimensional $\mathbb F_p$-subspace of
		$\mathbb F_{p^\ell}$ and assume that
		$m\mid(q-1)/(p^\ell-1)$. Let $u_i$ be defined by
		Equation~\eqref{eq:ui}. If $a_i\in M_{j_0}$ for some
		$j_0\in\{1,\ldots,t\}$, then
		\begin{equation}\label{eq:ui-mixed-fibers}
			u_i=(c_Vm)^{-1}a_i^{1-m}
			\prod_{\substack{1\le j\le t\\j\ne j_0}}
			(b_{j_0}-b_j)^{-1}.
		\end{equation}
		Here $c_V=L_V'(x)$ is the formal derivative of the subspace
		polynomial $L_V(x)$. Consequently,
		$c_Vma_i^{m-1}u_i\in\mathbb F_{p^\ell}^{*}$ for
		$1\le i\le n$. Since $a_i^m\in\mathbb F_{p^\ell}^{*}$, we also have
		$a_i^{-1}u_i\in\mathbb F_{p^\ell}^{*}$ for $1\le i\le n$.
	\end{lemma}
	
	\begin{proof}
		The monic vanishing polynomial of $M$ is
		\[
		P_M(x):=\prod_{a\in M}(x-a)
		=\prod_{j=1}^t\bigl(\Phi(x)-b_j\bigr).
		\]
		Since $L_V'(x)=c_V$, the chain rule gives
		$\Phi'(x)=c_Vmx^{m-1}$. If $a_i\in M_{j_0}$, differentiation
		at $a_i$ gives
		\[
		u_i^{-1}=P_M'(a_i)
		=c_Vma_i^{m-1}
		\prod_{\substack{1\le j\le t\\j\ne j_0}}
		(b_{j_0}-b_j).
		\]
		Equation~\eqref{eq:ui-mixed-fibers} follows. Since the $b_j$
		belong to $\mathbb F_{p^\ell}$ and are distinct, the final
		assertions follow from
		\[
		a_i^{-1}u_i
		=(c_Vma_i^m)^{-1}(c_Vma_i^{m-1}u_i)
		\in\mathbb F_{p^\ell}^{*}.
		\]
	\end{proof}
	
	\begin{theorem}\label{thm:mixed-fibers-hull}
		Let $q=p^e>4$, where $p$ is a prime, and assume that
		$2\ell\mid e$. Keep the notation and assumptions of Construction F.
		For every
		\[
		s+1\le k\le
		\left\lfloor\frac{n+p^\ell}{p^\ell+1}\right\rfloor
		\]
		and every $0\le h\le k-s-1$, there exists an $[n+s,k]_q$
		generalized Roth-Lempel code $C$ such that
		$\dim(\operatorname{Hull}_{\ell}(C))=h$.
	\end{theorem}
	
	\begin{proof}
		By Lemma~\ref{lem:mixed-fibers},
		$a_i^{-1}u_i\in\mathbb F_{p^\ell}^{*}$.
		The result follows from Proposition~\ref{thm:common-nonzero}.
	\end{proof}
	
	\noindent\textbf{Remark}.
	The condition $m\mid(q-1)/(p^\ell-1)$ is a convenient uniform
	sufficient condition for complete splitting. More generally, it is
	enough to assume that $m\mid(q-1)$ and
	$\beta_j+V\subseteq(\mathbb F_q^{*})^m$ for $1\le j\le t$.

	We may also include the zero fiber. Let
	$b_1=0,b_2,\ldots,b_t$ be distinct elements of
	$L_V(\mathbb F_{p^\ell})$, where $1\le t\le p^{\ell-r}$, and put
	\[
	M_1=\{x\in\mathbb F_q:\Phi(x)=0\},\qquad
	M_j=\{x\in\mathbb F_q:\Phi(x)=b_j\},\quad 2\le j\le t.
	\]
	Since $L_V^{-1}(0)=V$, we have
	$M_1=\{x\in\mathbb F_q:x^m\in V\}$. The equation $x^m=0$ has
	one root, while $x^m=v$ has $m$ distinct roots for every
	$v\in V\setminus\{0\}$. Hence
	$|M_1|=1+m(p^r-1)$, while the preceding argument gives
	$|M_j|=mp^r$ for $2\le j\le t$. Therefore
	\[
	M=\bigcup_{j=1}^tM_j=\{a_1,\ldots,a_n\},
	\qquad n=1+m(tp^r-1).
	\]
	
	\begin{lemma}\label{lem:mixed-fibers-zero}
		Assume that $m\mid(q-1)/(p^\ell-1)$, and let $u_i$ be defined
		by Equation~\eqref{eq:ui}. For the above
		evaluation set,
		\begin{equation}\label{eq:ui-mixed-fibers-zero}
			u_i=
			\begin{cases}
				\displaystyle(c_Vmb_j)^{-1}
				\prod_{\substack{2\le j'\le t\\j'\ne j}}
				(b_j-b_{j'})^{-1},&a_i\in M_j,\ 2\le j\le t,\\[3mm]
				\displaystyle(c_Vm)^{-1}
				\prod_{j=2}^t(-b_j)^{-1},
				&a_i\in M_1\setminus\{0\},\\[3mm]
				\displaystyle c_V^{-1}
				\prod_{j=2}^t(-b_j)^{-1},&a_i=0.
			\end{cases}
		\end{equation}
		Consequently, $u_i\in\mathbb F_{p^\ell}^{*}$ for
		$1\le i\le n$, where empty products are understood as $1$.
	\end{lemma}
	
	\begin{proof}
		The monic vanishing polynomial of $M_1$ is
		\[
		R_0(x)=x\prod_{v\in V\setminus\{0\}}(x^m-v)
		=\frac{\Phi(x)}{x^{m-1}},
		\]
		and hence
		\[
		P_M(x)=R_0(x)\prod_{j=2}^t(\Phi(x)-b_j).
		\]
		If $a_i\in M_j$ with $j\ge2$, then
		\[
		u_i^{-1}=P_M'(a_i)=c_Vm b_j
		\prod_{\substack{2\le j'\le t\\j'\ne j}}(b_j-b_{j'}).
		\]
		If $a_i\in M_1\setminus\{0\}$, then
		$R_0'(a_i)=c_Vm$, whereas $R_0'(0)=c_V$. Therefore
		\[
		u_i^{-1}=
		\begin{cases}
			\displaystyle c_Vm\prod_{j=2}^t(-b_j),
			&a_i\in M_1\setminus\{0\},\\[2mm]
			\displaystyle c_V\prod_{j=2}^t(-b_j),&a_i=0.
		\end{cases}
		\]
		Inverting these identities gives
		Equation~\eqref{eq:ui-mixed-fibers-zero}, and the final
		assertion follows.
	\end{proof}
	
	\begin{theorem}\label{thm:mixed-fibers-zero-hull}
		Let $q=p^e>4$, where $p$ is a prime, and assume that
		$2\ell\mid e$. Keep the notation above, so that
		$n=1+m(tp^r-1)$. Then, for every
		\[
		s<k\le
		\left\lfloor\frac{n+p^\ell}{p^\ell+1}\right\rfloor
		\]
		and every $0\le h\le k-s$, there exists an $[n+s,k]_q$
		generalized Roth-Lempel code $C$ such that
		$\dim(\operatorname{Hull}_{\ell}(C))=h$.
	\end{theorem}
	
	\begin{proof}
		By Lemma~\ref{lem:mixed-fibers-zero},
		$u_i\in\mathbb F_{p^\ell}^{*}$ for $1\le i\le n$.
		Taking $\delta=1$ and $\lambda=1$ in
		Corollary~\ref{cor:common-direct-extended} gives the result.
	\end{proof}
	
	Moreover, the nonzero-fiber construction yields a $(k-s)$-dimensional
	hull for a special value of $k$.
	
	\begin{theorem}\label{thm:mixed-fibers-plus-one}
		Under the field and evaluation-set assumptions of Theorem~\ref{thm:mixed-fibers-hull},
		assume that $\frac{n}{p^\ell+1}\in\mathbb Z_{>0}$ and put
		$k=\xi+1+\frac{n}{p^\ell+1}$. If
		\[
		0\le \xi\le
		\min\left\{p^\ell-1,
		\left\lfloor
		\frac{\frac{n}{p^\ell+1}-1}{p^\ell}
		\right\rfloor\right\},
		\qquad
		\xi+1\le s\le\min\{\xi+p^\ell,k-1\},
		\]
		then there exists an $[n+s,k]_q$ generalized Roth-Lempel code
		$C$ with a $(k-s)$-dimensional $\ell$-Galois hull.
	\end{theorem}
	
	\begin{proof}
		By Lemma~\ref{lem:mixed-fibers},
		$a_i^{-1}u_i\in\mathbb F_{p^\ell}^{*}$ for $1\le i\le n$.
		The result follows from
		Proposition~\ref{thm:common-plus-one-nonzero}.
	\end{proof}
	
	The zero-fiber construction also gives a $(k-s+1)$-dimensional hull
	under the corresponding shifted parameter bound.
	
	\begin{theorem}\label{thm:mixed-fibers-zero-plus-one}
		Under the field and evaluation-set assumptions of
		Theorem~\ref{thm:mixed-fibers-zero-hull}, assume that
		$\frac{n-1}{p^\ell+1}\in\mathbb Z_{>0}$ and put
		$k=\xi+1+\frac{n-1}{p^\ell+1}$. If
		\[
		0\le \xi\le
		\min\left\{p^\ell-1,
		\left\lfloor\frac{n-1}{p^\ell(p^\ell+1)}\right\rfloor\right\},
		\qquad
		\xi+1\le s\le\min\{\xi+p^\ell,k-1\},
		\]
		then there exists an $[n+s,k]_q$ generalized Roth-Lempel code
		$C$ with a $(k-s+1)$-dimensional $\ell$-Galois hull.
	\end{theorem}
	
	\begin{proof}
		By Lemma~\ref{lem:mixed-fibers-zero},
		$u_i\in\mathbb F_{p^\ell}^{*}$ for $1\le i\le n$.
		Taking $\delta=1$ and $\lambda=1$ in
		Proposition~\ref{thm:common-plus-one-shifted} gives the result.
	\end{proof}
	
	\noindent\textbf{Remark}.
	If $t=p^{\ell-r}$, choose all the elements of
	$L_V(\mathbb F_{p^\ell})$, and let
	$m=(q-1)/(p^\ell-1)$. Then $x^m\in\mathbb F_{p^\ell}$ for every
	$x\in\mathbb F_q$, so
	$\Phi(x)=L_V(x^m)\in L_V(\mathbb F_{p^\ell})$. Hence
	$M=\mathbb F_q$, and indeed
	\[
	n=1+\frac{q-1}{p^\ell-1}(p^\ell-1)=q.
	\]
	Thus the zero-fiber construction can attain the full evaluation
	set, whereas the nonzero-fiber construction omits the zero fiber.
	Moreover, if $2\ell\mid e$, then
	\[
	\frac{q-1}{p^\ell-1}
	=(p^\ell+1)(1+p^{2\ell}+\cdots+p^{e-2\ell}),
	\]
	so $m\mid p^\ell+1$ is a simple sufficient condition for
	$m\mid(q-1)/(p^\ell-1)$.
	
	\subsection{A unified viewpoint and summary}
	
	The preceding constructions are governed by the normalization of the
	Lagrange coefficients. Constructions A, B, C, and the nonzero-fiber
	version of Construction F use
	$a_i^{-1}u_i\in\mathbb F_{p^\ell}^{*}$ on nonzero multiplicative
	evaluation sets and lead to $f^{p^\ell}(x)=xg(x)$. Their zero-extended
	versions, Construction D, and the zero-fiber version of Construction F
	use coefficients lying directly in
	$\mathbb F_{p^\ell}^{*}$ and lead to $f^{p^\ell}(x)=g(x)$.
	Construction E differs only by the global factor $c_V$. The principal
	arbitrary-hull results
	are summarized in Table~\ref{tab:current-constructions}.
	
	\begin{table}[htbp]
		\centering
		\caption{Ordinary hull ranges for $[n+s,k]_q$ GRL codes with $h$-dimensional $\ell$-Galois hulls.}
		\label{tab:current-constructions}
		\small
		\setlength{\tabcolsep}{7pt}
		\renewcommand{\arraystretch}{1.45}
		\begin{tabular*}{0.98\linewidth}{@{\extracolsep{\fill}}cccc@{}}
			\toprule
			\multicolumn{4}{c}{$q=p^e>4$, $p$ prime, $2\ell\mid e$} \\
			\midrule
			$n$ & $k$ & $h$ & Ref. \\
			\midrule
			$t\frac{q-1}{p^\ell-1}$ & $s+1\le k\le\left\lfloor\frac{n+p^\ell}{p^\ell+1}\right\rfloor$ & $0\le h\le k-s-1$ & Theorem~\ref{thm:II.1} \\
			\midrule
			$t\frac{q-1}{p^\ell-1}+1$ & $s+1\le k\le\left\lfloor\frac{n+p^\ell-1}{p^\ell+1}\right\rfloor$ & $0\le h\le k-s$ & Theorem~\ref{thm:II.2} \\
			\midrule
			$r_1r_2$ & $s+1\le k\le\left\lfloor\frac{n+p^\ell}{p^\ell+1}\right\rfloor$ & $0\le h\le k-s-1$ & Theorem~\ref{thm:II.14} \\
			\midrule
			$r_1r_2+1$ & $s+1\le k\le\left\lfloor\frac{n+p^\ell-1}{p^\ell+1}\right\rfloor$ & $0\le h\le k-s$ & Theorem~\ref{thm:II.15} \\
			\midrule
			$rm$ & $s+1\le k\le\left\lfloor\frac{n+p^\ell}{p^\ell+1}\right\rfloor$ & $0\le h\le k-s-1$ & Theorem~\ref{thm:II.18} \\
			\midrule
			$rm+1$ & $s+1\le k\le\left\lfloor\frac{n+p^\ell-1}{p^\ell+1}\right\rfloor$ & $0\le h\le k-s$ & Theorem~\ref{thm:II.19} \\
			\midrule
			$tp^{e-\ell}$ & $s+1\le k\le\left\lfloor\frac{n+p^\ell-1}{p^\ell+1}\right\rfloor$ & $0\le h\le k-s$ & Theorem~\ref{thm:II.t1} \\
			\midrule
			$tp^r$ & $s+1\le k\le\left\lfloor\frac{n+p^\ell-1}{p^\ell+1}\right\rfloor$ & $0\le h\le k-s$ & Theorem~\ref{thm:additive-subspace-hull} \\
			\midrule
			$tmp^r$ & $\begin{gathered}\textstyle s+1\le k\le\left\lfloor\frac{n+p^\ell}{p^\ell+1}\right\rfloor\\\textstyle s+1\le k\le\left\lfloor\frac{n+p^\ell-m}{p^\ell+1}\right\rfloor\end{gathered}$ & $\begin{gathered}\vphantom{\textstyle \left\lfloor\frac{n+p^\ell}{p^\ell+1}\right\rfloor}0\le h\le k-s-1\\\vphantom{\textstyle \left\lfloor\frac{n+p^\ell}{p^\ell+1}\right\rfloor}0\le h\le k-s\end{gathered}$ & $\begin{gathered}\vphantom{\textstyle \left\lfloor\frac{n+p^\ell}{p^\ell+1}\right\rfloor}\text{Theorem~\ref{thm:mixed-fibers-hull}}\\\vphantom{\textstyle \left\lfloor\frac{n+p^\ell}{p^\ell+1}\right\rfloor}\text{Proposition~\ref{thm:common-direct}}\end{gathered}$ \\
			\midrule
			$1+m(tp^r-1)$ & $s+1\le k\le\left\lfloor\frac{n+p^\ell-1}{p^\ell+1}\right\rfloor$ & $0\le h\le k-s$ & Theorem~\ref{thm:mixed-fibers-zero-hull} \\
			\bottomrule
		\end{tabular*}
		\par\smallskip
		
	\end{table}
	
	In addition to the ranges in Table~\ref{tab:current-constructions},
	Theorems~\ref{thm:II.444}, \ref{thm:B-hull-plus-one},
	\ref{thm:C-hull-plus-one}, and \ref{thm:mixed-fibers-plus-one} give
	$(k-s)$-dimensional hulls, while
	Theorems~\ref{thm:II.4}, \ref{thm:B-zero-hull-plus-one},
	\ref{thm:C-zero-hull-plus-one}, \ref{thm:D-hull-plus-one},
	\ref{thm:additive-subspace-plus-one}, and
	\ref{thm:mixed-fibers-zero-plus-one} give
	$(k-s+1)$-dimensional hulls under the corresponding special parameter
	conditions. This organization
	separates the common hull-control argument from the evaluation-set
	calculations, so further families can be incorporated by verifying the appropriate Lagrange-coefficient condition.
	
	\section{Galois Hulls of GRL Codes When $s=2$}\label{sec:s2}
	
	In this section, we consider the MDS and AMDS properties of generalized Roth-Lempel codes for $s=2$ and combine the resulting criteria with the preceding hull constructions. Let $\boldsymbol a=(a_1,\ldots,a_n)$ consist of distinct elements of $\mathbb F_q$, let $\boldsymbol v\in(\mathbb F_q^*)^n$, and take $A_2=\begin{pmatrix}a&b\\c&d\end{pmatrix}\in\operatorname{GL}_2(\mathbb F_q)$. A codeword of $C=\operatorname{GRL}_k(\boldsymbol a,\boldsymbol v,A_2)$ has the form
	$$
	\bigl(v_1f(a_1),\ldots,v_nf(a_n),af_{k-2}+cf_{k-1},bf_{k-2}+df_{k-1}\bigr),\qquad \deg f<k.
	$$
	For $1\le t\le n$, put $\Delta_t(\boldsymbol a)=\{\sum_{i\in I}a_i:I\subseteq\{1,\ldots,n\},\ |I|=t\}$. We can describe the distance of $\operatorname{GRL}_k(\boldsymbol a,\boldsymbol v,A_2)$ as follows:

	\begin{proposition}\label{prop:s2-distance}\cite{liu2026generalized}
		Let $2\le k\le n$. Then $C$ is either MDS or AMDS. More precisely,
		\begin{enumerate}
			\item[(1)] $C$ is an $[n+2,k,n-k+3]_q$ MDS code if and only if $c-a\sigma\ne0$ and $d-b\sigma\ne0$ for every $\sigma\in\Delta_{k-1}(\boldsymbol a)$;
			\item[(2)] $C$ is an $[n+2,k,n-k+2]_q$ AMDS code if and only if $c=a\sigma$ or $d=b\sigma$ for some $\sigma\in\Delta_{k-1}(\boldsymbol a)$. In this case, $C$ is also NMDS.
		\end{enumerate}
	\end{proposition}

	For convenience, let $\zeta\in\mathbb F_q$ and take $A_\zeta=\begin{pmatrix}0&1\\1&\zeta\end{pmatrix}$. The preceding criterion immediately gives the following result.
	
	\begin{corollary}\label{cor:s2-Adelta}
		Let $2\le k\le n$ and $C=\operatorname{GRL}_k(\boldsymbol a,\boldsymbol v,A_\zeta)$. Then $C$ is an $[n+2,k,n-k+3]_q$ MDS code if $\zeta\notin\Delta_{k-1}(\boldsymbol a)$, and an $[n+2,k,n-k+2]_q$ AMDS code if $\zeta\in\Delta_{k-1}(\boldsymbol a)$. In the latter case, $C$ is also NMDS.
	\end{corollary}
	
	\subsection{Prescribed Galois hulls}
	
	Throughout the hull constructions below, let $q=p^e>4$, where $p$ is a prime, and assume that $2\ell\mid e$. As before, put $u_i=\prod_{j\ne i}(a_i-a_j)^{-1}$. The distance criteria above do not depend on the nonzero multipliers, so they can be combined directly with the ordinary ranges of the common constructions.
	
	\begin{proposition}\label{thm:s2-arbitrary-hull}
		Assume that one of the following conditions holds:
		\begin{enumerate}
			\item[(1)] $a_i\ne0$ and $a_i^{-1}u_i\in\mathbb F_{p^\ell}^*$ for $1\le i\le n$, with
			$$
			3\le k\le\left\lfloor\frac{n+p^\ell}{p^\ell+1}\right\rfloor,\qquad 0\le h\le k-3;
			$$
			\item[(2)] $\lambda a_i^{\delta-1}u_i\in\mathbb F_{p^\ell}^*$ for $1\le i\le n$, where $\delta\in\mathbb Z_{>0}$ and $\lambda\in\mathbb F_q^*$, with
			$$
			3\le k\le\left\lfloor\frac{n+p^\ell-\delta}{p^\ell+1}\right\rfloor,\qquad 0\le h\le k-2.
			$$
		\end{enumerate}
		Then there exists an $[n+2,k,n-k+2]_q$ AMDS code $C$ with $\dim(\operatorname{Hull}_\ell(C))=h$. If $\Delta_{k-1}(\boldsymbol a)\ne\mathbb F_q$, there also exists an $[n+2,k,n-k+3]_q$ MDS code with the same hull dimension.
	\end{proposition}
	
	\begin{proof}
		Propositions~\ref{thm:common-nonzero} and~\ref{thm:common-direct}, respectively, provide the prescribed hull dimension for every nonsingular $A_2$. Choose $A_2=A_\zeta$ with $\zeta\in\Delta_{k-1}(\boldsymbol a)$ for AMDS codes, or $\zeta\notin\Delta_{k-1}(\boldsymbol a)$ for MDS codes. The conclusion follows from Corollary~\ref{cor:s2-Adelta}.
	\end{proof}
	
	For the exceptional ranges, we use the matrices in Propositions~\ref{thm:common-plus-one-nonzero} and~\ref{thm:common-plus-one-shifted}. Their distance criterion depends on the value of $\xi$.
	
	\begin{proposition}\label{thm:s2-hull-plus-one}
		Let $3\le k\le n$. The following statements hold.
		\begin{enumerate}
			\item[(1)] Assume that $a_i\ne0$, $a_i^{-1}u_i\in\mathbb F_{p^\ell}^*$ for $1\le i\le n$, and $n/(p^\ell+1)\in\mathbb Z_{>0}$. If
			$$
			k=\xi+1+\frac{n}{p^\ell+1},\qquad 0\le \xi\le\min\left\{1,\left\lfloor\frac{\frac{n}{p^\ell+1}-1}{p^\ell}\right\rfloor\right\},
			$$
			then there exists an $[n+2,k]_q$ GRL code $C$ with $\dim(\operatorname{Hull}_\ell(C))=k-2$.
			\item[(2)] Assume that $\lambda a_i^{\delta-1}u_i\in\mathbb F_{p^\ell}^*$ for $1\le i\le n$, where $\delta\in\mathbb Z_{>0}$ and $\lambda\in\mathbb F_{p^\ell}^*$, and that $(n-\delta)/(p^\ell+1)\in\mathbb Z_{>0}$. If
			$$
			k=\xi+1+\frac{n-\delta}{p^\ell+1},\qquad 0\le \xi\le\min\left\{1,\left\lfloor\frac{n-\delta}{p^\ell(p^\ell+1)}\right\rfloor\right\},
			$$
			then there exists an $[n+2,k]_q$ GRL code $C$ with $\dim(\operatorname{Hull}_\ell(C))=k-1$.
		\end{enumerate}
		In both cases, $C$ is an $[n+2,k,n-k+3]_q$ MDS code if $a_{n-1}^{(1)}\notin\Delta_{k-1}(\boldsymbol a)$ for $\xi=1$, or $0\notin\Delta_{k-1}(\boldsymbol a)$ for $\xi=0$; otherwise it is an $[n+2,k,n-k+2]_q$ NMDS code. In case (2) with $\xi=0$, the same distance can also be obtained with every hull dimension $0\le h\le k-2$.
	\end{proposition}
	
	\begin{proof}
		Take $\mu^{p^\ell+1}=-1$ in (1) and $\mu^{p^\ell+1}=-\lambda$ in (2). The constructions give
		$$
		A_2=\begin{pmatrix}1&0\\0&\mu\end{pmatrix}\quad(\xi=0),\qquad
		A_2=\begin{pmatrix}\mu&0\\\mu a_{n-1}^{(1)}&1\end{pmatrix}\quad(\xi=1).
		$$
		Indeed, $-\lambda\mu^{-p^\ell}=\mu$ in (2), and the same identity with $\lambda=1$ applies in (1). Propositions~\ref{thm:common-plus-one-nonzero} and~\ref{thm:common-plus-one-shifted} give the stated hull dimensions. In Proposition~\ref{prop:s2-distance}, the only possible vanishing factor is $-\sigma$ for $\xi=0$ and $\mu(a_{n-1}^{(1)}-\sigma)$ for $\xi=1$. This proves the distance assertions. For the additional hull dimensions in case (2) with $\xi=0$, we have $\delta-1+p^\ell(k-1)=n-k$. Apply Corollary~\ref{cor:common-direct-extended} with $A_2=\mu I_2$ and $\mu^{p^\ell+1}=-\eta$, where $\eta\in(\mathbb F_q^*)^{p^\ell+1}$ and $\eta\ne\lambda$; Proposition~\ref{prop:s2-distance} gives the same distance criterion.
	\end{proof}
	
	\noindent\textbf{Remark}.
	At the additional boundary $\delta-1+p^\ell(k-1)=n-k$ in Corollary~\ref{cor:common-direct-extended}, take $A_2=\mu I_2$ with $\mu^{p^\ell+1}=-\eta$, where $\eta\in(\mathbb F_q^*)^{p^\ell+1}$ and $\eta\ne\lambda$. Then every $0\le h\le k-2$ is still obtained, and the MDS or AMDS property is determined by $0\notin\Delta_{k-1}(\boldsymbol a)$ or $0\in\Delta_{k-1}(\boldsymbol a)$, respectively. This boundary does not allow the unrestricted choice $A_2=A_\zeta$ used in Proposition~\ref{thm:s2-arbitrary-hull}.
	
	\subsection{Explicit families from the preceding constructions}\label{subsec:s2-explicit-MDS}
	
	We now list the evaluation sets and ordinary parameter ranges from Constructions A--F. In each case, $\boldsymbol a$ denotes the actual evaluation vector, including $0$ when it is adjoined; the subset sums are always taken with respect to this vector.
	
	\begin{theorem}\label{thm:s2-explicit-MDS}
		For each $n,k,h$ in Table~\ref{tab:s2-explicit-results}, choose the corresponding evaluation vector $\boldsymbol a$ under the stated field and evaluation-set assumptions. Then there exists an $[n+2,k,n-k+2]_q$ NMDS GRL code $C$ with $\dim(\operatorname{Hull}_\ell(C))=h$. If $\Delta_{k-1}(\boldsymbol a)\ne\mathbb F_q$, there also exists an $[n+2,k,n-k+3]_q$ MDS GRL code $C$ with the same hull dimension. These two choices are obtained by taking $A_2=A_\zeta$ with $\zeta\in\Delta_{k-1}(\boldsymbol a)$ and $\zeta\notin\Delta_{k-1}(\boldsymbol a)$, respectively.
	\end{theorem}
	
	\begin{table}[htbp]
		\centering
		\caption{Parameters of the $[n+2,k,d]_q$ GRL codes in Theorem~\ref{thm:s2-explicit-MDS}.}
		\label{tab:s2-explicit-results}
		\small
		\setlength{\tabcolsep}{7pt}
		\renewcommand{\arraystretch}{1.45}
		\begin{tabular*}{0.98\linewidth}{@{\extracolsep{\fill}}cccc@{}}
			\toprule
			\multicolumn{4}{c}{$q=p^e>4$, $p$ prime, $2\ell\mid e$} \\
			\midrule
			$n$ & $k$ & $h$ & Ref. \\
			\midrule
			$t\frac{q-1}{p^\ell-1}$ & $3\le k\le\left\lfloor\frac{n+p^\ell}{p^\ell+1}\right\rfloor$ & $0\le h\le k-3$ & Theorem~\ref{thm:II.1} \\
			\midrule
			$t\frac{q-1}{p^\ell-1}+1$ & $3\le k\le\left\lfloor\frac{n+p^\ell-1}{p^\ell+1}\right\rfloor$ & $0\le h\le k-2$ & Theorem~\ref{thm:II.2} \\
			\midrule
			$r_1r_2$ & $3\le k\le\left\lfloor\frac{n+p^\ell}{p^\ell+1}\right\rfloor$ & $0\le h\le k-3$ & Theorem~\ref{thm:II.14} \\
			\midrule
			$r_1r_2+1$ & $3\le k\le\left\lfloor\frac{n+p^\ell-1}{p^\ell+1}\right\rfloor$ & $0\le h\le k-2$ & Theorem~\ref{thm:II.15} \\
			\midrule
			$rm$ & $3\le k\le\left\lfloor\frac{n+p^\ell}{p^\ell+1}\right\rfloor$ & $0\le h\le k-3$ & Theorem~\ref{thm:II.18} \\
			\midrule
			$rm+1$ & $3\le k\le\left\lfloor\frac{n+p^\ell-1}{p^\ell+1}\right\rfloor$ & $0\le h\le k-2$ & Theorem~\ref{thm:II.19} \\
			\midrule
			$tp^{e-\ell}$ & $3\le k\le\left\lfloor\frac{n+p^\ell-1}{p^\ell+1}\right\rfloor$ & $0\le h\le k-2$ & Theorem~\ref{thm:II.t1} \\
			\midrule
			$tp^r$ & $3\le k\le\left\lfloor\frac{n+p^\ell-1}{p^\ell+1}\right\rfloor$ & $0\le h\le k-2$ & Theorem~\ref{thm:additive-subspace-hull} \\
			\midrule
			$tmp^r$ & $\begin{gathered}\textstyle 3\le k\le\left\lfloor\frac{n+p^\ell}{p^\ell+1}\right\rfloor\\\textstyle 3\le k\le\left\lfloor\frac{n+p^\ell-m}{p^\ell+1}\right\rfloor\end{gathered}$ & $\begin{gathered}\vphantom{\textstyle \left\lfloor\frac{n+p^\ell}{p^\ell+1}\right\rfloor}0\le h\le k-3\\\vphantom{\textstyle \left\lfloor\frac{n+p^\ell}{p^\ell+1}\right\rfloor}0\le h\le k-2\end{gathered}$ & $\begin{gathered}\vphantom{\textstyle \left\lfloor\frac{n+p^\ell}{p^\ell+1}\right\rfloor}\text{Theorem~\ref{thm:mixed-fibers-hull}}\\\vphantom{\textstyle \left\lfloor\frac{n+p^\ell}{p^\ell+1}\right\rfloor}\text{Proposition~\ref{thm:common-direct}}\end{gathered}$ \\
			\midrule
			$1+m(tp^r-1)$ & $3\le k\le\left\lfloor\frac{n+p^\ell-1}{p^\ell+1}\right\rfloor$ & $0\le h\le k-2$ & Theorem~\ref{thm:mixed-fibers-zero-hull} \\
			\bottomrule
		\end{tabular*}
		\par\smallskip
		
	\end{table}
	
	\begin{proof}
		Rows 1, 3, 5 and the first range in row~9 of Table~\ref{tab:s2-explicit-results} follow from Proposition~\ref{thm:common-nonzero}. The other rows follow from Proposition~\ref{thm:common-direct} with $\delta=1$. For the second range in row~9, Lemma~\ref{lem:mixed-fibers} and Proposition~\ref{thm:common-direct} apply with $\delta=m$. The global factors in Constructions E and F are $c_V$ and $mc_V$, respectively. Taking $s=2$ gives the stated hull ranges for every nonsingular $A_2$. Corollary~\ref{cor:s2-Adelta} then gives the distance and NMDS assertions.
	\end{proof}
	
	For the full evaluation set, the preceding constructions give AMDS codes of length $q+2$.
	
	\begin{theorem}\label{thm:s2-full-field-NMDS}
		Let $q=p^e>4$, where $p$ is a prime, and assume that $2\ell\mid e$. For every $3\le k\le1+(q-1)/(p^\ell+1)$ and $0\le h\le k-2$, there exists a $[q+2,k,q-k+2]_q$ AMDS code $C$ with $\dim(\operatorname{Hull}_\ell(C))=h$. These codes are also NMDS.
	\end{theorem}
	
	\begin{proof}
		Take $t=p^\ell-1$ in row~2 of Table~\ref{tab:s2-explicit-results} and apply Theorem~\ref{thm:s2-explicit-MDS}; the evaluation set is $\mathbb F_q$. Choose $\zeta\in\Delta_{k-1}(\mathbb F_q)$, which is nonempty. Then $A_2=A_\zeta$ gives an AMDS, and hence NMDS, code by Corollary~\ref{cor:s2-Adelta}. The hull range follows from the same construction for $k\le(q-1)/(p^\ell+1)$. At $k=1+(q-1)/(p^\ell+1)$, a multiplicative subgroup of $\mathbb F_q^*$ of order $k-1>1$ has sum zero, so $0\in\Delta_{k-1}(\mathbb F_q)$. The boundary assertion in Proposition~\ref{thm:s2-hull-plus-one}(2) with $n=q$ and $\delta=\lambda=1$ gives the remaining case.
	\end{proof}
	
	\begin{corollary}\label{cor:s2-full-field-plus-one}
		Under the field and evaluation-set assumptions of Theorem~\ref{thm:s2-full-field-NMDS}, let
		$$
		k=\xi+1+\frac{q-1}{p^\ell+1},\qquad 0\le \xi\le\min\left\{1,\left\lfloor\frac{q-1}{p^\ell(p^\ell+1)}\right\rfloor\right\}.
		$$
		Then there exists a $[q+2,k,q-k+2]_q$ AMDS code with a $(k-1)$-dimensional $\ell$-Galois hull. This code is also NMDS.
	\end{corollary}
	
	\begin{proof}
		For the evaluation set $\mathbb F_q$, we have $u_i=-1\in\mathbb F_{p^\ell}^*$ and $a_{q-1}^{(1)}=0$. A multiplicative subgroup of $\mathbb F_q^*$ of order $(q-1)/(p^\ell+1)>1$ has sum zero; adjoining $0$ gives a zero-sum subset with one more element. Thus $0\in\Delta_{k-1}(\mathbb F_q)$ for both permitted values of $\xi$. Apply Proposition~\ref{thm:s2-hull-plus-one}(2) with $n=q$ and $\delta=\lambda=1$.
	\end{proof}
	
	\noindent\textbf{Remark}.
	The same full evaluation set is obtained from Construction D with $t=p^\ell$, or from row~10 of Table~\ref{tab:s2-explicit-results} with $t=p^{\ell-r}$ and $m=(q-1)/(p^\ell-1)$. The AMDS conclusions above remain valid for these realizations. In characteristic $2$, the subset-sum set need not be all of $\mathbb F_q$; for example, $\Delta_2(\mathbb F_q)=\mathbb F_q^*$.
	
	\begin{example}
		We can construct some MDS GRL codes with certain Galois hull as Table \ref{tab:s2-mds-examples}.
		The examples are obtained from the evaluation
		sets in Constructions A--F. 
		\begin{table}[htbp]
			\centering
			\caption{Some MDS GRL codes with $s=2$ and prescribed $\ell$-Galois hulls.}
			\label{tab:s2-mds-examples}
			\small
			\setlength{\tabcolsep}{6pt}
			\renewcommand{\arraystretch}{1.25}
			\begin{tabular}{c c c c c c}
				\hline
				$q$ & $p^\ell$ & $n$ & $[n+2,k,d]_q$ & $h$ & Ref. \\
				\hline
				$169$ & $13$ & $43$ & $[45,3,43]_{169}$ & $0\le h\le1$ & Thm.~\ref{thm:s2-explicit-MDS} \\
				$81$ & $3$ & $11$ & $[13,3,11]_{81}$ & $0\le h\le1$ & Thm.~\ref{thm:s2-explicit-MDS} \\
				$64$ & $2$ & $9$ & $[11,5,7]_{64}$ & $h=3$ & Prop.~\ref{thm:s2-hull-plus-one} \\
				$256$ & $2$ & $16$ & $[18,5,14]_{256}$ & $0\le h\le3$ & Thm.~\ref{thm:s2-explicit-MDS} \\
				$64$ & $2$ & $32$ & $[34,11,24]_{64}$ & $0\le h\le9$ & Thm.~\ref{thm:s2-explicit-MDS} \\
				$81$ & $3$ & $27$ & $[29,7,23]_{81}$ & $0\le h\le5$ & Thm.~\ref{thm:s2-explicit-MDS} \\
				$256$ & $2$ & $128$ & $[130,43,88]_{256}$ & $0\le h\le41$ & Thm.~\ref{thm:s2-explicit-MDS} \\
				$256$ & $4$ & $64$ & $[66,13,54]_{256}$ & $0\le h\le11$ & Thm.~\ref{thm:s2-explicit-MDS} \\
				$625$ & $5$ & $125$ & $[127,21,107]_{625}$ & $0\le h\le19$ & Thm.~\ref{thm:s2-explicit-MDS} \\
				$1024$ & $2$ & $512$ & $[514,171,344]_{1024}$ & $0\le h\le169$ & Thm.~\ref{thm:s2-explicit-MDS} \\
				$2401$ & $7$ & $343$ & $[345,43,303]_{2401}$ & $0\le h\le41$ & Thm.~\ref{thm:s2-explicit-MDS} \\
				$4096$ & $2$ & $2048$ & $[2050,683,1368]_{4096}$ & $0\le h\le681$ & Thm.~\ref{thm:s2-explicit-MDS} \\
				$4096$ & $8$ & $512$ & $[514,57,458]_{4096}$ & $0\le h\le55$ & Thm.~\ref{thm:s2-explicit-MDS} \\
				$6561$ & $9$ & $729$ & $[731,73,659]_{6561}$ & $0\le h\le71$ & Thm.~\ref{thm:s2-explicit-MDS} \\
				$625$ & $5$ & $25$ & $[27,4,24]_{625}$ & $0\le h\le2$ & Thm.~\ref{thm:s2-explicit-MDS} \\
				$6561$ & $9$ & $60$ & $[62,6,57]_{6561}$ & $0\le h\le3$ & Thm.~\ref{thm:s2-explicit-MDS} \\
				\hline
			\end{tabular}
		\end{table}
	\end{example}
	\begin{example}
		We can also construct many AMDS GRL codes with certain Galois hull , see Table~\ref{tab:s2-amds-good-examples}.
		The extension matrix is chosen so that the resulting code is
		$[n+2,k,n-k+2]_q$ AMDS.
		{\small
			\setlength{\tabcolsep}{5pt}
			\renewcommand{\arraystretch}{1.18}
			\begin{longtable}{c c c c c c}
				\caption{Some AMDS GRL codes with $s=2$ and prescribed $\ell$-Galois hulls.}
				\label{tab:s2-amds-good-examples}\\
				\toprule
				$q$ & $p^\ell$ & $n$ & $[n+2,k,d]_q$ & $h$ & Ref. \\
				\midrule
				\endfirsthead
				
				\multicolumn{6}{c}{\tablename\ \thetable\ -- continued}\\
				\toprule
				$q$ & $p^\ell$ & $n$ & $[n+2,k,d]_q$ & $h$ & Ref. \\
				\midrule
				\endhead
				
				\midrule
				\multicolumn{6}{r}{Continued on next page}\\
				\endfoot
				
				\bottomrule
				\endlastfoot
				$64$ & $2$ & $32$ & $[34,11,23]_{64}$ & $0\le h\le 9$ & Thm.~\ref{thm:s2-explicit-MDS} \\
				$64$ & $2$ & $64$ & $[66,22,44]_{64}$ & $0\le h\le 21$ & Thm.~\ref{thm:s2-full-field-NMDS} \& Cor.~\ref{cor:s2-full-field-plus-one} \\
				$81$ & $3$ & $54$ & $[56,14,42]_{81}$ & $0\le h\le 12$ & Thm.~\ref{thm:s2-explicit-MDS} \\
				$81$ & $3$ & $81$ & $[83,21,62]_{81}$ & $0\le h\le 20$ & Thm.~\ref{thm:s2-full-field-NMDS} \& Cor.~\ref{cor:s2-full-field-plus-one} \\
				$256$ & $2$ & $16$ & $[18,6,12]_{256}$ & $0\le h\le 5$ & Prop.~\ref{thm:s2-hull-plus-one} \\
				$256$ & $2$ & $128$ & $[130,43,87]_{256}$ & $0\le h\le 41$ & Thm.~\ref{thm:s2-explicit-MDS} \\
				$256$ & $2$ & $256$ & $[258,86,172]_{256}$ & $0\le h\le 85$ & Thm.~\ref{thm:s2-full-field-NMDS} \& Cor.~\ref{cor:s2-full-field-plus-one} \\
				$256$ & $4$ & $192$ & $[194,39,155]_{256}$ & $0\le h\le 37$ & Thm.~\ref{thm:s2-explicit-MDS} \\
				$256$ & $4$ & $256$ & $[258,52,206]_{256}$ & $0\le h\le 51$ & Thm.~\ref{thm:s2-full-field-NMDS} \& Cor.~\ref{cor:s2-full-field-plus-one} \\
				$256$ & $16$ & $256$ & $[258,16,242]_{256}$ & $0\le h\le 15$ & Thm.~\ref{thm:s2-full-field-NMDS} \& Cor.~\ref{cor:s2-full-field-plus-one} \\
				$625$ & $5$ & $157$ & $[159,27,132]_{625}$ & $0\le h\le 26$ & Prop.~\ref{thm:s2-hull-plus-one} \\
				$625$ & $5$ & $469$ & $[471,79,392]_{625}$ & $0\le h\le 78$ & Prop.~\ref{thm:s2-hull-plus-one} \\
				$625$ & $5$ & $500$ & $[502,84,418]_{625}$ & $0\le h\le 82$ & Thm.~\ref{thm:s2-explicit-MDS} \\
				$625$ & $5$ & $625$ & $[627,105,522]_{625}$ & $0\le h\le 104$ & Thm.~\ref{thm:s2-full-field-NMDS} \& Cor.~\ref{cor:s2-full-field-plus-one} \\
				$625$ & $25$ & $625$ & $[627,25,602]_{625}$ & $0\le h\le 24$ & Thm.~\ref{thm:s2-full-field-NMDS} \& Cor.~\ref{cor:s2-full-field-plus-one} \\
				$729$ & $3$ & $486$ & $[488,122,366]_{729}$ & $0\le h\le 120$ & Thm.~\ref{thm:s2-explicit-MDS} \\
				$729$ & $3$ & $729$ & $[731,183,548]_{729}$ & $0\le h\le 182$ & Thm.~\ref{thm:s2-full-field-NMDS} \& Cor.~\ref{cor:s2-full-field-plus-one} \\
				$1024$ & $2$ & $512$ & $[514,171,343]_{1024}$ & $0\le h\le 169$ & Thm.~\ref{thm:s2-explicit-MDS} \\
				$1024$ & $2$ & $1024$ & $[1026,342,684]_{1024}$ & $0\le h\le 341$ & Thm.~\ref{thm:s2-full-field-NMDS} \& Cor.~\ref{cor:s2-full-field-plus-one} \\
				$1024$ & $32$ & $1024$ & $[1026,32,994]_{1024}$ & $0\le h\le 31$ & Thm.~\ref{thm:s2-full-field-NMDS} \& Cor.~\ref{cor:s2-full-field-plus-one} \\
				$2401$ & $7$ & $1601$ & $[1603,201,1402]_{2401}$ & $0\le h\le 200$ & Prop.~\ref{thm:s2-hull-plus-one} \\
				$2401$ & $7$ & $2058$ & $[2060,258,1802]_{2401}$ & $0\le h\le 256$ & Thm.~\ref{thm:s2-explicit-MDS} \\
				$2401$ & $7$ & $2401$ & $[2403,301,2102]_{2401}$ & $0\le h\le 300$ & Thm.~\ref{thm:s2-full-field-NMDS} \& Cor.~\ref{cor:s2-full-field-plus-one} \\
				$2401$ & $49$ & $2401$ & $[2403,49,2354]_{2401}$ & $0\le h\le 48$ & Thm.~\ref{thm:s2-full-field-NMDS} \& Cor.~\ref{cor:s2-full-field-plus-one} \\
				$4096$ & $2$ & $2048$ & $[2050,683,1367]_{4096}$ & $0\le h\le 681$ & Thm.~\ref{thm:s2-explicit-MDS} \\
				$4096$ & $2$ & $4096$ & $[4098,1366,2732]_{4096}$ & $0\le h\le 1365$ & Thm.~\ref{thm:s2-full-field-NMDS} \& Cor.~\ref{cor:s2-full-field-plus-one} \\
				$4096$ & $4$ & $1366$ & $[1368,274,1094]_{4096}$ & $0\le h\le 273$ & Prop.~\ref{thm:s2-hull-plus-one} \\
				$4096$ & $4$ & $3072$ & $[3074,615,2459]_{4096}$ & $0\le h\le 613$ & Thm.~\ref{thm:s2-explicit-MDS} \\
				$4096$ & $4$ & $4096$ & $[4098,820,3278]_{4096}$ & $0\le h\le 819$ & Thm.~\ref{thm:s2-full-field-NMDS} \& Cor.~\ref{cor:s2-full-field-plus-one} \\
				$4096$ & $8$ & $2926$ & $[2928,325,2603]_{4096}$ & $0\le h\le 323$ & Thm.~\ref{thm:s2-explicit-MDS} \\
				$4096$ & $8$ & $3584$ & $[3586,399,3187]_{4096}$ & $0\le h\le 397$ & Thm.~\ref{thm:s2-explicit-MDS} \\
				$4096$ & $8$ & $4096$ & $[4098,456,3642]_{4096}$ & $0\le h\le 455$ & Thm.~\ref{thm:s2-full-field-NMDS} \& Cor.~\ref{cor:s2-full-field-plus-one} \\
				$4096$ & $64$ & $4096$ & $[4098,64,4034]_{4096}$ & $0\le h\le 63$ & Thm.~\ref{thm:s2-full-field-NMDS} \& Cor.~\ref{cor:s2-full-field-plus-one} \\
				$6561$ & $3$ & $4374$ & $[4376,1094,3282]_{6561}$ & $0\le h\le 1092$ & Thm.~\ref{thm:s2-explicit-MDS} \\
				$6561$ & $3$ & $6561$ & $[6563,1641,4922]_{6561}$ & $0\le h\le 1640$ & Thm.~\ref{thm:s2-full-field-NMDS} \& Cor.~\ref{cor:s2-full-field-plus-one} \\
				$6561$ & $9$ & $821$ & $[823,83,740]_{6561}$ & $0\le h\le 82$ & Prop.~\ref{thm:s2-hull-plus-one} \\
				$6561$ & $9$ & $4101$ & $[4103,411,3692]_{6561}$ & $0\le h\le 410$ & Prop.~\ref{thm:s2-hull-plus-one} \\
				$6561$ & $9$ & $5832$ & $[5834,584,5250]_{6561}$ & $0\le h\le 582$ & Thm.~\ref{thm:s2-explicit-MDS} \\
				$6561$ & $9$ & $6561$ & $[6563,657,5906]_{6561}$ & $0\le h\le 656$ & Thm.~\ref{thm:s2-full-field-NMDS} \& Cor.~\ref{cor:s2-full-field-plus-one} \\
				$6561$ & $81$ & $6561$ & $[6563,81,6482]_{6561}$ & $0\le h\le 80$ & Thm.~\ref{thm:s2-full-field-NMDS} \& Cor.~\ref{cor:s2-full-field-plus-one} \\
				$15625$ & $5$ & $11719$ & $[11721,1954,9767]_{15625}$ & $0\le h\le 1953$ & Prop.~\ref{thm:s2-hull-plus-one} \\
				$15625$ & $5$ & $12500$ & $[12502,2084,10418]_{15625}$ & $0\le h\le 2082$ & Thm.~\ref{thm:s2-explicit-MDS} \\
				$15625$ & $5$ & $15625$ & $[15627,2605,13022]_{15625}$ & $0\le h\le 2604$ & Thm.~\ref{thm:s2-full-field-NMDS} \& Cor.~\ref{cor:s2-full-field-plus-one} \\
				$15625$ & $125$ & $15625$ & $[15627,125,15502]_{15625}$ & $0\le h\le 124$ & Thm.~\ref{thm:s2-full-field-NMDS} \& Cor.~\ref{cor:s2-full-field-plus-one} \\
			\end{longtable}
		}

	\end{example}

	\newpage
	\section{Galois Hulls of GRL Codes When $s=3$}\label{sec:s3}
	
	We now consider $C=\operatorname{GRL}_k(\boldsymbol a,\boldsymbol v,A_3)$, where $4\le k\le n$ and
	$$
	A_3=\begin{pmatrix}a_{11}&a_{12}&a_{13}\\a_{21}&a_{22}&a_{23}\\a_{31}&a_{32}&a_{33}\end{pmatrix}
	=(\boldsymbol a_1,\boldsymbol a_2,\boldsymbol a_3)\in\operatorname{GL}_3(\mathbb F_q),\qquad
	\boldsymbol a_j=(a_{1j},a_{2j},a_{3j})^T.
	$$
	The last three coordinates of the codeword associated with $f$ are $(f_{k-3},f_{k-2},f_{k-1})A_3$. For $S\subseteq\{1,\ldots,n\}$, put $e_1(S)=\sum_{i\in S}a_i$ and $e_2(S)=\sum_{i,j\in S,\,i<j}a_ia_j$. Throughout this section, $I,J\subseteq\{1,\ldots,n\}$ satisfy $|I|=k-2$ and $|J|=k-1$, respectively. Define
	$$
	\boldsymbol\omega_I=(1,e_1(I),e_1(I)^2-e_2(I))^T,\qquad
	\boldsymbol\gamma_J=(e_2(J),-e_1(J),1)^T.
	$$
	
	We can describe the distance of $\operatorname{GRL}_k(\boldsymbol a,\boldsymbol v,A_3)$ as follows:
	
	\begin{lemma}\label{thm:s3-MDS}\cite{LiangWanLiao2026GRL}
		The code $C$ is an $[n+3,k,n-k+4]_q$ MDS code if and only if
		\begin{enumerate}
			\item[(1)] $\det(\boldsymbol a_r,\boldsymbol a_t,\boldsymbol\omega_I)\ne0$ for every $I$ and every $1\le r<t\le3$;
			\item[(2)] $\boldsymbol a_r^T\boldsymbol\gamma_J\ne0$ for every $J$ and every $1\le r\le3$.
		\end{enumerate}
	\end{lemma}
	
	\begin{lemma}\label{thm:s3-AMDS}
		
		The code $C$ is an $[n+3,k,n-k+3]_q$ AMDS code if and only if both of the following conditions hold:
		\begin{enumerate}
			\item[(1)] $\det(\boldsymbol a_r,\boldsymbol a_t,\boldsymbol\omega_I)=0$ for some $I$ and $1\le r<t\le3$, or $\boldsymbol a_r^T\boldsymbol\gamma_J=0$ for some $J$ and $1\le r\le3$;
			\item[(2)] for every $J$, at most one of $\boldsymbol a_1^T\boldsymbol\gamma_J$, $\boldsymbol a_2^T\boldsymbol\gamma_J$ and $\boldsymbol a_3^T\boldsymbol\gamma_J$ is zero.
		\end{enumerate}
	\end{lemma}
	
	\begin{proof}
		This lemma follows from Lemma~\ref{thm:s3-MDS} and \cite{WuHengLiDing2024}. We give a short proof.
		
		By Lemma~\ref{thm:s3-MDS}, condition (1) is equivalent to $C$ not being MDS, and hence $d(C)\le n-k+3$. For any nonzero $f(x)\in\mathbb F_q[x]_{<k}$, if $f$ has at most $k-3$ evaluation zeros, its codeword has weight at least $n-k+3$. If $f$ has exactly $k-2$ evaluation zeros, then $\deg f\ge k-2$, so its last three coefficients are not all zero. Since $A_3$ is nonsingular, at least one appended coordinate is nonzero, and the weight is again at least $n-k+3$. If $f$ has $k-1$ evaluation zeros indexed by $J$, then $f(x)$ is a nonzero scalar multiple of $\prod_{j\in J}(x-a_j)$. Its appended coordinates are therefore a nonzero scalar multiple of $\boldsymbol\gamma_J^TA_3$. Thus its weight is at least $n-k+3$ if and only if condition (2) holds.
		
		Consequently, condition (2) is equivalent to $d(C)\ge n-k+3$. Combining this with condition (1) gives $d(C)=n-k+3$, as required.
	\end{proof}
	
	\subsection{Prescribed Galois hulls}
	
	Let $q=p^e>4$, where $p$ is a prime, and assume that $2\ell\mid e$. We first consider the ordinary ranges, in which $A_3$ is arbitrary.
	
	\begin{corollary}\label{cor:s3-arbitrary-hull}
		Assume that one of the following conditions holds:
		\begin{enumerate}
			\item[(1)] $a_i\ne0$ and $a_i^{-1}u_i\in\mathbb F_{p^\ell}^*$ for $1\le i\le n$, with
			$$
			4\le k\le\left\lfloor\frac{n+p^\ell}{p^\ell+1}\right\rfloor,\qquad 0\le h\le k-4;
			$$
			\item[(2)] $\lambda a_i^{\delta-1}u_i\in\mathbb F_{p^\ell}^*$ for $1\le i\le n$, where $\delta\in\mathbb Z_{>0}$ and $\lambda\in\mathbb F_q^*$, with
			$$
			4\le k\le\left\lfloor\frac{n+p^\ell-\delta}{p^\ell+1}\right\rfloor,\qquad 0\le h\le k-3.
			$$
		\end{enumerate}
		If $A_3$ satisfies Lemma~\ref{thm:s3-MDS}, there exists an $[n+3,k,n-k+4]_q$ MDS GRL code $C$ with $\dim(\operatorname{Hull}_\ell(C))=h$. If $A_3$ satisfies Lemma~\ref{thm:s3-AMDS}, there exists an $[n+3,k,n-k+3]_q$ AMDS GRL code with the same hull dimension.
	\end{corollary}
	
	\begin{proof}
		Apply Proposition~\ref{thm:common-nonzero} or~\ref{thm:common-direct} with $s=3$. The change of nonzero multipliers used to prescribe $h$ scales the evaluation columns and therefore preserves all the rank conditions in Lemmas~\ref{thm:s3-MDS} and~\ref{thm:s3-AMDS}.
	\end{proof}
	
	For a nonsingular diagonal matrix $A_3$, the distance criteria reduce to the following conditions on $e_1$ and $e_2$:
	\begin{equation}\label{eq:s3-diagonal-MDS}
		\begin{array}{ll}
			e_1(I)\ne0,\quad e_1(I)^2-e_2(I)\ne0,&\text{for every }I,\\
			e_1(J)\ne0,\quad e_2(J)\ne0,&\text{for every }J.
		\end{array}
	\end{equation}
	Namely, $C$ is MDS if and only if all the conditions in \eqref{eq:s3-diagonal-MDS} hold. It is AMDS if and only if at least one of them fails and $(e_1(J),e_2(J))\ne(0,0)$ for every $J$.

	For the exceptional ranges, the matrices in the common constructions can be written as
	\begin{equation}\label{eq:s3-exceptional-matrices}
		A_3=\begin{cases}
			\operatorname{diag}(1,1,\mu),&\xi=0,\\[2pt]
			\begin{pmatrix}1&0&0\\0&\mu&0\\0&\mu a_{n-1}^{(1)}&1\end{pmatrix},&\xi=1,\\[6pt]
			\begin{pmatrix}\mu&0&0\\\mu a_{n-1}^{(1)}&1&0\\\mu a_{n-1}^{(2)}&0&1\end{pmatrix},&\xi=2.
		\end{cases}
	\end{equation}
	
	Their distance conditions are summarized in Table~\ref{tab:s3-exceptional-distance}. 
	
	\begin{table}[htbp]
		\centering
		\caption{Distance criteria for the matrices in \eqref{eq:s3-exceptional-matrices}, with $|I|=k-2$ and $|J|=k-1$.}
		\label{tab:s3-exceptional-distance}
		\small
		\setlength{\tabcolsep}{6pt}
		\renewcommand{\arraystretch}{1.5}
		\begin{tabular}{@{}c c c@{}}
			\toprule
			$\xi$ & MDS conditions & Additional AMDS condition \\
			\midrule
			$0$ &
			$\begin{gathered}
				e_1(I)\ne0,\quad e_1(I)^2-e_2(I)\ne0,\\
				e_1(J)\ne0,\quad e_2(J)\ne0
			\end{gathered}$ &
			$(e_1(J),e_2(J))\ne(0,0)$ \\
			\midrule
			$1$ &
			$\begin{gathered}
				e_1(I)\ne0,\quad e_1(I)^2-e_2(I)-a_{n-1}^{(1)}e_1(I)\ne0,\\
				e_1(J)\ne a_{n-1}^{(1)},\quad e_2(J)\ne0
			\end{gathered}$ &
			$(e_1(J)-a_{n-1}^{(1)},e_2(J))\ne(0,0)$ \\
			\midrule
			$2$ &
			$\begin{gathered}
				e_1(I)\ne a_{n-1}^{(1)},\quad e_1(I)^2-e_2(I)\ne a_{n-1}^{(2)},\\
				e_1(J)\ne0,\quad e_2(J)-a_{n-1}^{(1)}e_1(J)+a_{n-1}^{(2)}\ne0
			\end{gathered}$ &
			$(e_1(J),e_2(J)+a_{n-1}^{(2)})\ne(0,0)$ \\
			\bottomrule
		\end{tabular}
	\end{table}
	
	\begin{corollary}\label{cor:s3-hull-plus-one}
		Let $4\le k\le n$. The following statements hold.
		\begin{enumerate}
			\item[(1)] Assume that $a_i\ne0$, $a_i^{-1}u_i\in\mathbb F_{p^\ell}^*$ for $1\le i\le n$, and $n/(p^\ell+1)\in\mathbb Z_{>0}$. If
			$$
			k=\xi+1+\frac{n}{p^\ell+1},\qquad \max\{0,3-p^\ell\}\le \xi\le\min\left\{2,p^\ell-1,\left\lfloor\frac{\frac{n}{p^\ell+1}-1}{p^\ell}\right\rfloor\right\},
			$$
			then choosing $A_3$ as in \eqref{eq:s3-exceptional-matrices} with $\mu^{p^\ell+1}=-1$ gives an $[n+3,k]_q$ GRL code $C$ with $\dim(\operatorname{Hull}_\ell(C))=k-3$.
			\item[(2)] Assume that $\lambda a_i^{\delta-1}u_i\in\mathbb F_{p^\ell}^*$ for $1\le i\le n$, where $\delta\in\mathbb Z_{>0}$ and $\lambda\in\mathbb F_{p^\ell}^*$, and that $(n-\delta)/(p^\ell+1)\in\mathbb Z_{>0}$. If
			$$
			k=\xi+1+\frac{n-\delta}{p^\ell+1},\qquad \max\{0,3-p^\ell\}\le \xi\le\min\left\{2,p^\ell-1,\left\lfloor\frac{n-\delta}{p^\ell(p^\ell+1)}\right\rfloor\right\},
			$$
			then choosing $A_3$ as in \eqref{eq:s3-exceptional-matrices} with $\mu^{p^\ell+1}=-\lambda$ gives an $[n+3,k]_q$ GRL code $C$ with $\dim(\operatorname{Hull}_\ell(C))=k-2$.
		\end{enumerate}
		In both cases, the MDS and AMDS properties are determined by the corresponding row of Table~\ref{tab:s3-exceptional-distance}.
	\end{corollary}
	
	\begin{proof}
		The hull dimensions follow from Propositions~\ref{thm:common-plus-one-nonzero} and~\ref{thm:common-plus-one-shifted} with $s=3$. The bounds on $\xi$ retain both $\xi\le p^\ell-1$ and $3\le \xi+p^\ell$; in particular, only $\xi=1$ is permitted when $p^\ell=2$. Substituting \eqref{eq:s3-exceptional-matrices} into Lemmas~\ref{thm:s3-MDS} and~\ref{thm:s3-AMDS} gives Table~\ref{tab:s3-exceptional-distance}. For example, when $\xi=1$, the three determinants are nonzero scalar multiples of $e_1(I)^2-e_2(I)-a_{n-1}^{(1)}e_1(I)$, $-e_1(I)$, and $1$, and the last-coordinate expressions are $e_2(J)$, $\mu(a_{n-1}^{(1)}-e_1(J))$, and $1$. The other two rows follow in the same way.
	\end{proof}
	
	\noindent\textbf{Remark}.
	At the boundary $\delta-1+p^\ell(k-1)=n-k$ of Corollary~\ref{cor:common-direct-extended}, take $A_3=\mu I_3$ with $\mu^{p^\ell+1}=-\eta$, where $\eta\in(\mathbb F_q^*)^{p^\ell+1}$ and $\eta\ne\lambda$. Then every $0\le h\le k-3$ is obtained, and the same diagonal distance criteria apply.
	
	\subsection{Explicit families from the preceding constructions}
	
	We now apply the preceding MDS and AMDS criteria to Constructions A--F. The following theorem gives the complete conclusions for each evaluation set.
	
	\begin{theorem}\label{thm:s3-explicit-families}
		For each $n,k,h$ in Table~\ref{tab:s3-explicit-results}, choose the corresponding evaluation vector $\boldsymbol a$ under the stated field and evaluation-set assumptions, and let $A_3\in\operatorname{GL}_3(\mathbb F_q)$. If $A_3$ satisfies the conditions of Lemma~\ref{thm:s3-MDS}, there exists an $[n+3,k,n-k+4]_q$ MDS GRL code $C$ with $\dim(\operatorname{Hull}_\ell(C))=h$. If $A_3$ satisfies the conditions of Lemma~\ref{thm:s3-AMDS}, there exists an $[n+3,k,n-k+3]_q$ AMDS GRL code $C$ with the same hull dimension. All distance conditions are taken with respect to $\boldsymbol a$.
	\end{theorem}
	
	\begin{table}[htbp]
		\centering
		\caption{Parameters of the $[n+3,k,d]_q$ GRL codes in Theorem~\ref{thm:s3-explicit-families}.}
		\label{tab:s3-explicit-results}
		\small
		\setlength{\tabcolsep}{7pt}
		\renewcommand{\arraystretch}{1.45}
		\begin{tabular*}{0.98\linewidth}{@{\extracolsep{\fill}}cccc@{}}
			\toprule
			\multicolumn{4}{c}{$q=p^e>4$, $p$ prime, $2\ell\mid e$} \\
			\midrule
			$n$ & $k$ & $h$ & Ref. \\
			\midrule
			$t\frac{q-1}{p^\ell-1}$ & $4\le k\le\left\lfloor\frac{n+p^\ell}{p^\ell+1}\right\rfloor$ & $0\le h\le k-4$ & Theorem~\ref{thm:II.1} \\
			\midrule
			$t\frac{q-1}{p^\ell-1}+1$ & $4\le k\le\left\lfloor\frac{n+p^\ell-1}{p^\ell+1}\right\rfloor$ & $0\le h\le k-3$ & Theorem~\ref{thm:II.2} \\
			\midrule
			$r_1r_2$ & $4\le k\le\left\lfloor\frac{n+p^\ell}{p^\ell+1}\right\rfloor$ & $0\le h\le k-4$ & Theorem~\ref{thm:II.14} \\
			\midrule
			$r_1r_2+1$ & $4\le k\le\left\lfloor\frac{n+p^\ell-1}{p^\ell+1}\right\rfloor$ & $0\le h\le k-3$ & Theorem~\ref{thm:II.15} \\
			\midrule
			$rm$ & $4\le k\le\left\lfloor\frac{n+p^\ell}{p^\ell+1}\right\rfloor$ & $0\le h\le k-4$ & Theorem~\ref{thm:II.18} \\
			\midrule
			$rm+1$ & $4\le k\le\left\lfloor\frac{n+p^\ell-1}{p^\ell+1}\right\rfloor$ & $0\le h\le k-3$ & Theorem~\ref{thm:II.19} \\
			\midrule
			$tp^{e-\ell}$ & $4\le k\le\left\lfloor\frac{n+p^\ell-1}{p^\ell+1}\right\rfloor$ & $0\le h\le k-3$ & Theorem~\ref{thm:II.t1} \\
			\midrule
			$tp^r$ & $4\le k\le\left\lfloor\frac{n+p^\ell-1}{p^\ell+1}\right\rfloor$ & $0\le h\le k-3$ & Theorem~\ref{thm:additive-subspace-hull} \\
			\midrule
			$tmp^r$ & $\begin{gathered}\textstyle 4\le k\le\left\lfloor\frac{n+p^\ell}{p^\ell+1}\right\rfloor\\\textstyle 4\le k\le\left\lfloor\frac{n+p^\ell-m}{p^\ell+1}\right\rfloor\end{gathered}$ & $\begin{gathered}\vphantom{\textstyle \left\lfloor\frac{n+p^\ell}{p^\ell+1}\right\rfloor}0\le h\le k-4\\\vphantom{\textstyle \left\lfloor\frac{n+p^\ell}{p^\ell+1}\right\rfloor}0\le h\le k-3\end{gathered}$ & $\begin{gathered}\vphantom{\textstyle \left\lfloor\frac{n+p^\ell}{p^\ell+1}\right\rfloor}\text{Theorem~\ref{thm:mixed-fibers-hull}}\\\vphantom{\textstyle \left\lfloor\frac{n+p^\ell}{p^\ell+1}\right\rfloor}\text{Proposition~\ref{thm:common-direct}}\end{gathered}$ \\
			\midrule
			$1+m(tp^r-1)$ & $4\le k\le\left\lfloor\frac{n+p^\ell-1}{p^\ell+1}\right\rfloor$ & $0\le h\le k-3$ & Theorem~\ref{thm:mixed-fibers-zero-hull} \\
			\bottomrule
		\end{tabular*}
		\par\smallskip
		
	\end{table}
	
	\begin{proof}
		Apply the same normalization as in Theorem~\ref{thm:s2-explicit-MDS}, now with $s=3$. The hull dimensions follow from Propositions~\ref{thm:common-nonzero} and~\ref{thm:common-direct}. Changing the nonzero multipliers scales only the evaluation columns, so the distance conditions in Lemmas~\ref{thm:s3-MDS} and~\ref{thm:s3-AMDS} are preserved.
	\end{proof}
	
	\noindent\textbf{Remark}.
	The exceptional constructions in A--F are covered by Corollary~\ref{cor:s3-hull-plus-one} after imposing their stated divisibility and coefficient conditions. Nonzero evaluation sets satisfying Proposition~\ref{thm:common-plus-one-nonzero} yield hull dimension $k-3$, whereas the shifted constructions satisfying Proposition~\ref{thm:common-plus-one-shifted} yield hull dimension $k-2$. In each case, the MDS and AMDS properties are determined by Table~\ref{tab:s3-exceptional-distance}. When the relevant $a_{n-1}^{(j)}$ vanish, these conditions reduce to the diagonal criteria.
	
	\begin{example}
		According to the above discussion, we can construct some MDS GRL codes with certain $\ell$-Galois hull as Table~\ref{tab:s3-mds-verified-examples}.
		\begin{table}[htbp]
			\centering
			\caption{Some MDS GRL codes with $s=3$ and prescribed $\ell$-Galois hulls.}
			\label{tab:s3-mds-verified-examples}
			\small\setlength{\tabcolsep}{5pt}\renewcommand{\arraystretch}{1.2}
			\begin{tabular}{@{}cccccc@{}}
				\toprule
				$q$ & $p^\ell$ & $n$ & $[n+3,k,d]_q$ & $h$ & Ref. \\
				\midrule
				$6561$ & $3$ & $27$ & $[30,5,26]_{6561}$ & $0\le h\le 2$ & Thm.~\ref{thm:s3-explicit-families} \\
				$6561$ & $3$ & $27$ & $[30,6,25]_{6561}$ & $0\le h\le 3$ & Thm.~\ref{thm:s3-explicit-families} \\
				$65536$ & $2$ & $32$ & $[35,9,27]_{65536}$ & $0\le h\le 6$ & Thm.~\ref{thm:s3-explicit-families} \\
				$65536$ & $2$ & $32$ & $[35,10,26]_{65536}$ & $0\le h\le 7$ & Thm.~\ref{thm:s3-explicit-families} \\
				$65536$ & $2$ & $32$ & $[35,11,25]_{65536}$ & $0\le h\le 8$ & Thm.~\ref{thm:s3-explicit-families} \\
				$117649$ & $7$ & $49$ & $[52,6,47]_{117649}$ & $0\le h\le 3$ & Thm.~\ref{thm:s3-explicit-families} \\
				$117649$ & $7$ & $49$ & $[52,7,46]_{117649}$ & $0\le h\le 4$ & Thm.~\ref{thm:s3-explicit-families} \\
				\bottomrule
			\end{tabular}
		\end{table}

	\end{example}
	\begin{example}
		
		According the above discussion, we can construct some AMDS GRL codes with certain $\ell$-Galois hull as Table\ref{tab:s3-amds-selected-examples}.
		{\small
			\setlength{\tabcolsep}{4pt}
			\renewcommand{\arraystretch}{1.12}
			\begin{longtable}{@{}cccccc@{}}
				\caption{Selected AMDS GRL codes with $s=3$ and prescribed $\ell$-Galois hulls.}
				\label{tab:s3-amds-selected-examples}\\
				\toprule
				$q$ & $p^\ell$ & $n$ & $[n+3,k,d]_q$ & $h$ & Ref. \\
				\midrule
				\endfirsthead
				\multicolumn{6}{c}{\tablename~\thetable\ (continued)}\\
				\toprule
				$q$ & $p^\ell$ & $n$ & $[n+3,k,d]_q$ & $h$ & Ref. \\
				\midrule
				\endhead
				\midrule
				\multicolumn{6}{r}{Continued on next page}\\
				\endfoot
				\bottomrule
				\endlastfoot
				$64$ & $2$ & $32$ & $[35,11,24]_{64}$ & $0\le h\le 8$ & Thm.~\ref{thm:s3-explicit-families} \\
				$64$ & $8$ & $32$ & $[35,4,31]_{64}$ & $0\le h\le 1$ & Thm.~\ref{thm:s3-explicit-families} \\
				$64$ & $2$ & $22$ & $[25,4,21]_{64}$ & $0\le h\le 1$ & Cor.~\ref{cor:s3-arbitrary-hull} \\
				$64$ & $2$ & $21$ & $[24,5,19]_{64}$ & $0\le h\le 1$ & Cor.~\ref{cor:s3-arbitrary-hull} \\
				$81$ & $3$ & $81$ & $[84,4,80]_{81}$ & $0\le h\le 1$ & Thm.~\ref{thm:s3-explicit-families} \\
				$81$ & $9$ & $81$ & $[84,4,80]_{81}$ & $0\le h\le 1$ & Thm.~\ref{thm:s3-explicit-families} \\
				$81$ & $9$ & $61$ & $[64,4,60]_{81}$ & $0\le h\le 1$ & Cor.~\ref{cor:s3-arbitrary-hull} \\
				$81$ & $3$ & $54$ & $[57,4,53]_{81}$ & $0\le h\le 1$ & Thm.~\ref{thm:s3-explicit-families} \\
				$81$ & $9$ & $51$ & $[54,4,50]_{81}$ & $0\le h\le 1$ & Thm.~\ref{thm:s3-explicit-families} \\
				$81$ & $3$ & $27$ & $[30,7,23]_{81}$ & $0\le h\le 4$ & Thm.~\ref{thm:s3-explicit-families} \\
				$81$ & $3$ & $21$ & $[24,5,19]_{81}$ & $0\le h\le 2$ & Cor.~\ref{cor:s3-arbitrary-hull} \\
				$256$ & $4$ & $171$ & $[174,4,170]_{256}$ & $0\le h\le 1$ & Thm.~\ref{thm:s3-explicit-families} \\
				$256$ & $16$ & $171$ & $[174,4,170]_{256}$ & $0\le h\le 1$ & Cor.~\ref{cor:s3-arbitrary-hull} \\
				$256$ & $2$ & $128$ & $[131,42,89]_{256}$ & $0\le h\le 39$ & Thm.~\ref{thm:s3-explicit-families} \\
				$256$ & $4$ & $128$ & $[131,5,126]_{256}$ & $0\le h\le 2$ & Thm.~\ref{thm:s3-explicit-families} \\
				$256$ & $16$ & $128$ & $[131,5,126]_{256}$ & $0\le h\le 2$ & Thm.~\ref{thm:s3-explicit-families} \\
				$256$ & $2$ & $86$ & $[89,5,84]_{256}$ & $0\le h\le 2$ & Cor.~\ref{cor:s3-arbitrary-hull} \\
				$256$ & $4$ & $64$ & $[67,12,55]_{256}$ & $0\le h\le 9$ & Thm.~\ref{thm:s3-explicit-families} \\
				$625$ & $5$ & $125$ & $[128,20,108]_{625}$ & $0\le h\le 17$ & Thm.~\ref{thm:s3-explicit-families} \\
				$625$ & $25$ & $125$ & $[128,4,124]_{625}$ & $0\le h\le 1$ & Thm.~\ref{thm:s3-explicit-families} \\
				$729$ & $3$ & $243$ & $[246,60,186]_{729}$ & $0\le h\le 57$ & Thm.~\ref{thm:s3-explicit-families} \\
				$729$ & $27$ & $243$ & $[246,4,242]_{729}$ & $0\le h\le 1$ & Thm.~\ref{thm:s3-explicit-families} \\
				$2401$ & $7$ & $343$ & $[346,42,304]_{2401}$ & $0\le h\le 39$ & Thm.~\ref{thm:s3-explicit-families} \\
			\end{longtable}
		}
	\end{example}

	\section{Applications to EAQECCs}\label{sec:eaqecc}

	In this section, we apply the preceding Galois hull constructions to provide entanglement-assisted quantum error-correcting codes (EAQECCs). A $q$-ary EAQECC with parameters $[[n,K,d_Q;c]]_q$ encodes $K$ logical qudits into $n$ physical qudits with the assistance of $c$ maximally entangled pairs. If $2d_Q\le n+2$, the entanglement-assisted Singleton bound gives $n+c-K\ge2(d_Q-1)$. In this range, equality gives an MDS EAQECC, while $n+c-K=2d_Q$ gives an AMDS EAQECC. For $2d_Q>n+2$, the same inequality is not a general bound on all EAQECCs.
	
	The following result is the tool we use to convert classical codes into EAQECCs; it also determines their minimum distance.

	\begin{lemma}\cite{GalindoHernandoMatsumotoRuano2019,WildeBrun2008}\label{lem}
		Let $C_1$ and $C_2$ be two $q$-ary linear codes with parameters $[n,k_1,d_1]_q$ and $[n,k_2,d_2]_q$, respectively. Then there exists an EAQECC with parameters
		$$[[n,k_1+k_2-n+c,\min\{d_1,d_2\};c]]_q,$$
		where $c=\operatorname{rank}(H_1H_2^T)$ and $H_i$ is a parity-check matrix of $C_i$, $i=1,2$.
	\end{lemma}
	
	Combining Lemma~\ref{lem} with Lemma~\ref{lem:galois-hull-rank}, we obtain the following proposition.
	
	\begin{proposition}\label{prop:EA-hull}\cite[Proposition~IV.1]{9559992}
		Let $q=p^e$, let $0\le\ell\le e-1$, and let $C$ be an $[n,k,d]_q$ linear code with $\dim(\operatorname{Hull}_\ell(C))=h$. Then there exists an $[[n,k-h,d;n-k-h]]_q$ EAQECC.
	\end{proposition}

	We first record the consequences valid for general $s$. They provide distance lower bounds without imposing an additional MDS or AMDS assumption.
	
	\begin{theorem}\label{thm:EA-general}\setlength{\emergencystretch}{2em}
		Let $q=p^e>4$, where $p$ is a prime, and assume that $2\ell\mid e$. Let $\boldsymbol a=(a_1,\ldots,a_n)$ consist of distinct elements of $\mathbb F_q$, put $u_j=\prod_{v\ne j}(a_j-a_v)^{-1}$, and let $2\le s<k\le n$. The following statements hold.
		\begin{enumerate}
			\item[(1)] Assume that $a_j\ne0$ and $a_j^{-1}u_j\in\mathbb F_{p^\ell}^*$ for $1\le j\le n$. For every $s<k\le\left\lfloor\frac{n+p^\ell}{p^\ell+1}\right\rfloor$ and every $0\le h\le k-s-1$, there exists an \mbox{$[[n+s,k-h,d;n+s-k-h]]_q$} EAQECC with $d\ge n-k+2$.
			
			\item[(2)] Let $\delta\in\mathbb Z_{>0}$ and $\lambda\in\mathbb F_q^*$, and assume that $\lambda a_j^{\delta-1}u_j\in\mathbb F_{p^\ell}^*$ for $1\le j\le n$. For every $s<k\le\left\lfloor\frac{n+p^\ell-\delta+1}{p^\ell+1}\right\rfloor$ and every $0\le h\le k-s$, there exists an \mbox{$[[n+s,k-h,d;n+s-k-h]]_q$} EAQECC with $d\ge n-k+2$.
			
			\item[(3)] Assume that $a_j\ne0$ and $a_j^{-1}u_j\in\mathbb F_{p^\ell}^*$ for $1\le j\le n$, and that $n/(p^\ell+1)\in\mathbb Z_{>0}$. If $k=\xi+1+n/(p^\ell+1)$, $0\le \xi\le\min\{p^\ell-1,\lfloor(n/(p^\ell+1)-1)/p^\ell\rfloor\}$ and $\xi+1\le s\le\min\{\xi+p^\ell,k-1\}$, there exists an \mbox{$[[n+s,s,d;n-2k+2s]]_q$} EAQECC with $d\ge n-k+2$.
			
			\item[(4)] Let $\delta\in\mathbb Z_{>0}$ and $\lambda\in\mathbb F_{p^\ell}^*$, and assume that $\lambda a_j^{\delta-1}u_j\in\mathbb F_{p^\ell}^*$ for $1\le j\le n$. Suppose that $(n-\delta)/(p^\ell+1)\in\mathbb Z_{>0}$. If $k=\xi+1+(n-\delta)/(p^\ell+1)$, $0\le \xi\le\min\{p^\ell-1,\lfloor(n-\delta)/(p^\ell(p^\ell+1))\rfloor\}$ and $\xi+1\le s\le\min\{\xi+p^\ell,k-1\}$, there exists an \mbox{$[[n+s,s-1,d;n-2k+2s-1]]_q$} EAQECC with $d\ge n-k+2$.
		\end{enumerate}
	\end{theorem}
	
	\begin{proof}
		For any $C=\operatorname{GRL}_k(\boldsymbol a,\boldsymbol v,A_s)$ with $A_s$ nonsingular, a nonzero polynomial $f$ of degree at most $k-s-1$ gives a codeword of weight at least $n-k+s+1$. If $\deg f\ge k-s$, its highest $s$ coefficients are not all zero, so at least one appended coordinate is nonzero, while the evaluation part has weight at least $n-k+1$. Hence $d(C)\ge n-k+2$ in both cases.
		
		The required hull dimensions follow, respectively, from Proposition~\ref{thm:common-nonzero}, Corollary~\ref{cor:common-direct-extended}, and Propositions~\ref{thm:common-plus-one-nonzero} and~\ref{thm:common-plus-one-shifted}. Apply Proposition~\ref{prop:EA-hull}; in (3) and (4), substitute $h=k-s$ and $h=k-s+1$, respectively.
	\end{proof}
	
	\noindent\textbf{Remark}.
	In Theorem~\ref{thm:EA-general}(2), the additional boundary $\delta-1+p^\ell(k-1)=n-k$ requires a suitable extension matrix as in Corollary~\ref{cor:common-direct-extended}. The ordinary families below use the smaller range in Proposition~\ref{thm:common-direct}, where the extension matrix is arbitrary. The boundary and the larger hull dimensions are then treated separately with their corresponding distance criteria.
	
	\subsection{EAQECCs from MDS and AMDS GRL codes for $s=2$}
	
	For $s=2$, choosing $A_2=A_\zeta$ with $\zeta\in\Delta_{k-1}(\boldsymbol a)$ always gives an AMDS code. If $\Delta_{k-1}(\boldsymbol a)\ne\mathbb F_q$, choosing $\zeta$ outside this set gives an MDS code. We apply these two choices to each family in Theorem~\ref{thm:s2-explicit-MDS}.
	
	\begin{theorem}\label{thm:EA-s2-explicit}
		For each $n,k,h$ in Table~\ref{tab:s2-explicit-results}, choose the corresponding evaluation vector $\boldsymbol a$ under the stated field and evaluation-set assumptions. Then there exists an $[[n+2,k-h,n-k+2;n+2-k-h]]_q$ AMDS EAQECC. If $\Delta_{k-1}(\boldsymbol a)\ne\mathbb F_q$, there also exists an $[[n+2,k-h,n-k+3;n+2-k-h]]_q$ MDS EAQECC.
		
		Moreover, let $\delta\in\mathbb Z_{>0}$ and $\lambda\in\mathbb F_q^*$, and assume that $\lambda a_j^{\delta-1}u_j\in\mathbb F_{p^\ell}^*$ for $1\le j\le n$. Suppose that $\delta-1+p^\ell(k-1)=n-k$. Then, for every $0\le h\le k-2$, if $0\notin\Delta_{k-1}(\boldsymbol a)$, there exists an $[[n+2,k-h,n-k+3;n+2-k-h]]_q$ MDS EAQECC; if $0\in\Delta_{k-1}(\boldsymbol a)$, there exists an $[[n+2,k-h,n-k+2;n+2-k-h]]_q$ AMDS EAQECC.
	\end{theorem}
	
	\begin{proof}
		For the parameter ranges in Table~\ref{tab:s2-explicit-results}, apply Proposition~\ref{prop:EA-hull} to the two choices in Theorem~\ref{thm:s2-explicit-MDS}. For the additional boundary, choose $\eta\in(\mathbb F_q^*)^{p^\ell+1}$ with $\eta\ne\lambda$ and take $A_2=\mu I_2$ with $\mu^{p^\ell+1}=-\eta$. Corollary~\ref{cor:common-direct-extended} gives every $0\le h\le k-2$, while Proposition~\ref{prop:s2-distance} gives the MDS and AMDS cases according as $0\notin\Delta_{k-1}(\boldsymbol a)$ or $0\in\Delta_{k-1}(\boldsymbol a)$. The result follows from Proposition~\ref{prop:EA-hull}.
	\end{proof}

	The next theorem gives the parameters for the two larger hull dimensions. Here $n$ is the number of entries of the actual evaluation vector, including $0$ whenever it is present.

\begin{theorem}\label{thm:EA-s2-exceptional}\setlength{\emergencystretch}{2em}
	Let $q=p^e>4$, where $p$ is a prime, and assume that $2\ell\mid e$. For the evaluation sets in Table~\ref{tab:s2-explicit-results}, choose the corresponding evaluation vector $\boldsymbol a$ under the stated field and evaluation-set assumptions. The following statements hold.
	\begin{enumerate}
		\item[(1)] For the nonzero evaluation sets in rows 1, 3, 5 and the first range of row 9, suppose that $n/(p^\ell+1)\in\mathbb Z_{>0}$ and let $k=\xi+1+n/(p^\ell+1)$, where $0\le\xi\le\min\{1,\lfloor(n/(p^\ell+1)-1)/p^\ell\rfloor\}$ in Table~\ref{tab:s2-explicit-results}. If the MDS condition in Proposition~\ref{thm:s2-hull-plus-one} holds for the given $\xi$, there exists an $[[n+2,2,n-k+3;n-2k+4]]_q$ MDS EAQECC. Otherwise, there exists an $[[n+2,2,n-k+2;n-2k+4]]_q$ AMDS EAQECC.
		
		\item[(2)] For the evaluation sets in rows 2, 4, 6, 7, 8 and 10 in Table~\ref{tab:s2-explicit-results}, suppose that the corresponding exceptional hull construction applies and $(n-1)/(p^\ell+1)\in\mathbb Z_{>0}$. Let $k=\xi+1+(n-1)/(p^\ell+1)$, where $0\le\xi\le\min\{1,\lfloor(n-1)/(p^\ell(p^\ell+1))\rfloor\}$. If the MDS condition in Proposition~\ref{thm:s2-hull-plus-one} holds for the given $\xi$, there exists an $[[n+2,1,n-k+3;n-2k+3]]_q$ MDS EAQECC. Otherwise, there exists an $[[n+2,1,n-k+2;n-2k+3]]_q$ AMDS EAQECC.
	\end{enumerate}
\end{theorem}

	\begin{proof}
		In (1) and (2), use the matrices and distance criteria in Proposition~\ref{thm:s2-hull-plus-one}, giving hull dimensions $k-2$ and $k-1$, respectively. Applying Proposition~\ref{prop:EA-hull} gives the stated parameters.
	\end{proof}
	
	The full evaluation set gives the following AMDS EAQECCs of length $q+2$.
	
	\begin{corollary}\label{thm:EA-full-field}
		Let $q=p^e>4$, where $p$ is a prime, and assume that $2\ell\mid e$. The following statements hold.
		\begin{enumerate}
			\item[(1)] For every $3\le k\le(q-1)/(p^\ell+1)$ and every $0\le h\le k-2$, there exists a $[[q+2,k-h,q-k+2;q+2-k-h]]_q$ AMDS EAQECC.
			
			\item[(2)] Let $k=1+(q-1)/(p^\ell+1)$. For every $0\le h\le k-2$, there exists a $[[q+2,k-h,q-k+2;q+2-k-h]]_q$ AMDS EAQECC.
			
			\item[(3)] Let $k=1+(q-1)/(p^\ell+1)$. Then there exists a $[[q+2,1,q-k+2;q+3-2k]]_q$ AMDS EAQECC.
			
			\item[(4)] If $(q-1)/(p^\ell+1)\ge p^\ell$ and $k=2+(q-1)/(p^\ell+1)$, then there exists a $[[q+2,1,q-k+2;q+3-2k]]_q$ AMDS EAQECC.
		\end{enumerate}
	\end{corollary}
	
	\begin{proof}
		For (1), apply Proposition~\ref{prop:EA-hull} to Theorem~\ref{thm:s2-full-field-NMDS}. For (2), all Lagrange coefficients equal $-1$ and $0\in\Delta_{k-1}(\mathbb F_q)$. Since $p^\ell(k-1)=q-k$, the boundary case in Theorem~\ref{thm:EA-s2-explicit} applies with $n=q$ and $\delta=\lambda=1$. For (3) and (4), apply Theorem~\ref{thm:EA-s2-exceptional}(2) with $n=q$, $\delta=\lambda=1$ and $\xi=0,1$, respectively. This proves the four conclusions.
	\end{proof}

	\subsection{EAQECCs from MDS and AMDS GRL codes for $s=3$}
	
	For $s=3$, the minimum distance is determined by the conditions in Lemmas~\ref{thm:s3-MDS} and~\ref{thm:s3-AMDS}. Since these conditions are unchanged by nonzero scaling of the evaluation columns, they can be combined with all the ordinary hull ranges in Theorem~\ref{thm:s3-explicit-families}.
	
	\begin{theorem}\label{thm:EA-s3-explicit}
		For each $n,k,h$ in Table~\ref{tab:s3-explicit-results}, choose the corresponding evaluation vector $\boldsymbol a$ under the stated field and evaluation-set assumptions. If $A_3$ satisfies the conditions of Lemma~\ref{thm:s3-MDS}, there exists an $[[n+3,k-h,n-k+4;n+3-k-h]]_q$ MDS EAQECC. If $A_3$ satisfies the conditions of Lemma~\ref{thm:s3-AMDS}, there exists an $[[n+3,k-h,n-k+3;n+3-k-h]]_q$ AMDS EAQECC.
		
		Moreover, let $\delta\in\mathbb Z_{>0}$ and $\lambda\in\mathbb F_q^*$, and assume that $\lambda a_j^{\delta-1}u_j\in\mathbb F_{p^\ell}^*$ for $1\le j\le n$. Suppose that $\delta-1+p^\ell(k-1)=n-k$. Then, for every $0\le h\le k-3$, if all the conditions in \eqref{eq:s3-diagonal-MDS} hold, there exists an $[[n+3,k-h,n-k+4;n+3-k-h]]_q$ MDS EAQECC; if at least one of those conditions fails and $(e_1(J),e_2(J))\ne(0,0)$ for every $J$, there exists an $[[n+3,k-h,n-k+3;n+3-k-h]]_q$ AMDS EAQECC.
	\end{theorem}
	
	\begin{proof}
		For the parameter ranges in Table~\ref{tab:s3-explicit-results}, apply Proposition~\ref{prop:EA-hull} to Theorem~\ref{thm:s3-explicit-families}. For the additional boundary, choose $\eta\in(\mathbb F_q^*)^{p^\ell+1}$ with $\eta\ne\lambda$ and take $A_3=\mu I_3$ with $\mu^{p^\ell+1}=-\eta$. Corollary~\ref{cor:common-direct-extended} gives every $0\le h\le k-3$, while \eqref{eq:s3-diagonal-MDS} and the corresponding AMDS criterion determine the distance. The result follows from Proposition~\ref{prop:EA-hull}.
	\end{proof}
	
	For the two larger hull dimensions, we use the matrices in \eqref{eq:s3-exceptional-matrices}. The notation $e_1(I)$, $e_2(I)$, $e_1(J)$ and $e_2(J)$, with $|I|=k-2$ and $|J|=k-1$, is the same as in Section~\ref{sec:s3}.

	\begin{theorem}\label{thm:EA-s3-exceptional}\setlength{\emergencystretch}{2em}
		Let $q=p^e>4$, where $p$ is a prime, and assume that $2\ell\mid e$. For the evaluation sets in Table~\ref{tab:s3-explicit-results}, choose the corresponding evaluation vector $\boldsymbol a$ under the stated field and evaluation-set assumptions. All distance conditions below are taken with respect to this evaluation vector. The following statements hold.
		\begin{enumerate}
			\item[(1)] For the nonzero evaluation sets in rows 1, 3, 5 and the first range of row 9 in Table~\ref{tab:s3-explicit-results}, suppose that $n/(p^\ell+1)\in\mathbb Z_{>0}$ and let $k=\xi+1+n/(p^\ell+1)$, where $\max\{0,3-p^\ell\}\le\xi\le\min\{2,p^\ell-1,\lfloor(n/(p^\ell+1)-1)/p^\ell\rfloor\}$. If the MDS conditions in row $\xi$ of Table~\ref{tab:s3-exceptional-distance} hold, there exists an $[[n+3,3,n-k+4;n-2k+6]]_q$ MDS EAQECC. If instead the AMDS criterion in that row holds, there exists an $[[n+3,3,n-k+3;n-2k+6]]_q$ AMDS EAQECC.
			
			\item[(2)] For the evaluation sets in rows 2, 4, 6, 7, 8 and 10 in Table~\ref{tab:s3-explicit-results}, suppose that the corresponding exceptional hull construction applies and $(n-1)/(p^\ell+1)\in\mathbb Z_{>0}$. Let $k=\xi+1+(n-1)/(p^\ell+1)$, where $\max\{0,3-p^\ell\}\le\xi\le\min\{2,p^\ell-1,\lfloor(n-1)/(p^\ell(p^\ell+1))\rfloor\}$. If the MDS conditions in row $\xi$ of Table~\ref{tab:s3-exceptional-distance} hold, there exists an $[[n+3,2,n-k+4;n-2k+5]]_q$ MDS EAQECC. If instead the AMDS criterion in that row holds, there exists an $[[n+3,2,n-k+3;n-2k+5]]_q$ AMDS EAQECC.
		\end{enumerate}
	\end{theorem}
	
	\begin{proof}
		In (1) and (2), use Corollary~\ref{cor:s3-hull-plus-one} with the matrices in \eqref{eq:s3-exceptional-matrices} and the corresponding criteria in Table~\ref{tab:s3-exceptional-distance}. Their hull dimensions are $k-3$ and $k-2$, respectively. Applying Proposition~\ref{prop:EA-hull} gives the stated parameters.
	\end{proof}


\subsection{An example from AMDS GRL codes}\label{subsec:EA-examples}

	\begin{example}\label{ex:EA-amds-unified}
	Let $C$ be an $[n,k,n-k]_q$ AMDS GRL code with an $h$-dimensional
	$\ell$-Galois hull. For the families below, the quantum distance equals
	the classical distance. By Proposition~\ref{prop:EA-hull}, this gives an
	EAQECC with parameters $[[n,k-h,n-k;n-k-h]]_q$.
	Associated with this result, the AMDS codes in
	Tables~\ref{tab:s2-amds-good-examples} and~\ref{tab:s3-amds-selected-examples},
	the constructions in Corollary~\ref{thm:EA-full-field},
	and some additional examples from Construction~B are collected in Table~\ref{tab:EA-amds-examples}.
	
	{\footnotesize
		\setlength{\tabcolsep}{3pt}
		\renewcommand{\arraystretch}{1.1}
		\begin{longtable}{@{}ccccc@{}}
			\caption{EAQECCs obtained from AMDS GRL codes.}\label{tab:EA-amds-examples}\\
			\toprule
			$q$ & $p^\ell$ & Classical AMDS code & $h$ & EAQECC parameters \\
			\midrule
			\endfirsthead
			\multicolumn{5}{c}{\tablename~\thetable\ (continued)}\\
			\toprule
			$q$ & $p^\ell$ & Classical AMDS code & $h$ & EAQECC parameters \\
			\midrule
			\endhead
			\midrule
			\multicolumn{5}{r}{Continued on next page}\\
			\endfoot
			\bottomrule
			\endlastfoot
			\multicolumn{5}{@{}l}{\emph{Case $s=2$}}\\
			\midrule
			
			$64$ & $2$ & $[34,11,23]_{64}$ & $0\le h\le 9$ & $[[34,11-h,23;23-h]]_{64}$ \\
			$64$ & $2$ & $[66,22,44]_{64}$ & $0\le h\le 21$ & $[[66,22-h,44;44-h]]_{64}$ \\
			$81$ & $3$ & $[56,14,42]_{81}$ & $0\le h\le 12$ & $[[56,14-h,42;42-h]]_{81}$ \\
			$81$ & $3$ & $[83,21,62]_{81}$ & $0\le h\le 20$ & $[[83,21-h,62;62-h]]_{81}$ \\
			$256$ & $2$ & $[18,6,12]_{256}$ & $0\le h\le 5$ & $[[18,6-h,12;12-h]]_{256}$ \\
			$256$ & $2$ & $[130,43,87]_{256}$ & $0\le h\le 41$ & $[[130,43-h,87;87-h]]_{256}$ \\
			$256$ & $2$ & $[258,86,172]_{256}$ & $0\le h\le 85$ & $[[258,86-h,172;172-h]]_{256}$ \\
			$256$ & $4$ & $[194,39,155]_{256}$ & $0\le h\le 37$ & $[[194,39-h,155;155-h]]_{256}$ \\
			$256$ & $4$ & $[258,52,206]_{256}$ & $0\le h\le 51$ & $[[258,52-h,206;206-h]]_{256}$ \\
			$256$ & $16$ & $[258,16,242]_{256}$ & $0\le h\le 15$ & $[[258,16-h,242;242-h]]_{256}$ \\
			$625$ & $5$ & $[159,27,132]_{625}$ & $0\le h\le 26$ & $[[159,27-h,132;132-h]]_{625}$ \\
			$625$ & $5$ & $[471,79,392]_{625}$ & $0\le h\le 78$ & $[[471,79-h,392;392-h]]_{625}$ \\
			$625$ & $5$ & $[502,84,418]_{625}$ & $0\le h\le 82$ & $[[502,84-h,418;418-h]]_{625}$ \\
			$625$ & $25$ & $[627,25,602]_{625}$ & $0\le h\le 24$ & $[[627,25-h,602;602-h]]_{625}$ \\
			$729$ & $3$ & $[488,122,366]_{729}$ & $0\le h\le 120$ & $[[488,122-h,366;366-h]]_{729}$ \\
			$729$ & $3$ & $[731,183,548]_{729}$ & $0\le h\le 182$ & $[[731,183-h,548;548-h]]_{729}$ \\
			$1024$ & $2$ & $[514,171,343]_{1024}$ & $0\le h\le 169$ & $[[514,171-h,343;343-h]]_{1024}$ \\
			$1024$ & $2$ & $[1026,342,684]_{1024}$ & $0\le h\le 341$ & $[[1026,342-h,684;684-h]]_{1024}$ \\
			$1024$ & $32$ & $[1026,32,994]_{1024}$ & $0\le h\le 31$ & $[[1026,32-h,994;994-h]]_{1024}$ \\
			$2401$ & $7$ & $[1603,201,1402]_{2401}$ & $0\le h\le 200$ & $[[1603,201-h,1402;1402-h]]_{2401}$ \\
			$2401$ & $7$ & $[2060,258,1802]_{2401}$ & $0\le h\le 256$ & $[[2060,258-h,1802;1802-h]]_{2401}$ \\
			$2401$ & $7$ & $[2403,301,2102]_{2401}$ & $0\le h\le 300$ & $[[2403,301-h,2102;2102-h]]_{2401}$ \\
			$2401$ & $49$ & $[2403,49,2354]_{2401}$ & $0\le h\le 48$ & $[[2403,49-h,2354;2354-h]]_{2401}$ \\
			$4096$ & $2$ & $[2050,683,1367]_{4096}$ & $0\le h\le 681$ & $[[2050,683-h,1367;1367-h]]_{4096}$ \\
			$4096$ & $2$ & $[4098,1366,2732]_{4096}$ & $0\le h\le 1365$ & $[[4098,1366-h,2732;2732-h]]_{4096}$ \\
			$4096$ & $4$ & $[1368,274,1094]_{4096}$ & $0\le h\le 273$ & $[[1368,274-h,1094;1094-h]]_{4096}$ \\
			$4096$ & $4$ & $[3074,615,2459]_{4096}$ & $0\le h\le 613$ & $[[3074,615-h,2459;2459-h]]_{4096}$ \\
			$4096$ & $4$ & $[4098,820,3278]_{4096}$ & $0\le h\le 819$ & $[[4098,820-h,3278;3278-h]]_{4096}$ \\
			$4096$ & $8$ & $[2928,325,2603]_{4096}$ & $0\le h\le 323$ & $[[2928,325-h,2603;2603-h]]_{4096}$ \\
			$4096$ & $8$ & $[3586,399,3187]_{4096}$ & $0\le h\le 397$ & $[[3586,399-h,3187;3187-h]]_{4096}$ \\
			$4096$ & $8$ & $[4098,456,3642]_{4096}$ & $0\le h\le 455$ & $[[4098,456-h,3642;3642-h]]_{4096}$ \\
			$4096$ & $64$ & $[4098,64,4034]_{4096}$ & $0\le h\le 63$ & $[[4098,64-h,4034;4034-h]]_{4096}$ \\
			$6561$ & $3$ & $[4376,1094,3282]_{6561}$ & $0\le h\le 1092$ & $[[4376,1094-h,3282;3282-h]]_{6561}$ \\
			$6561$ & $3$ & $[6563,1641,4922]_{6561}$ & $0\le h\le 1640$ & $[[6563,1641-h,4922;4922-h]]_{6561}$ \\
			$6561$ & $9$ & $[823,83,740]_{6561}$ & $0\le h\le 82$ & $[[823,83-h,740;740-h]]_{6561}$ \\
			$6561$ & $9$ & $[4103,411,3692]_{6561}$ & $0\le h\le 410$ & $[[4103,411-h,3692;3692-h]]_{6561}$ \\
			$6561$ & $9$ & $[5834,584,5250]_{6561}$ & $0\le h\le 582$ & $[[5834,584-h,5250;5250-h]]_{6561}$ \\
			$6561$ & $9$ & $[6563,657,5906]_{6561}$ & $0\le h\le 656$ & $[[6563,657-h,5906;5906-h]]_{6561}$ \\
			$6561$ & $81$ & $[6563,81,6482]_{6561}$ & $0\le h\le 80$ & $[[6563,81-h,6482;6482-h]]_{6561}$ \\
			$15625$ & $5$ & $[11721,1954,9767]_{15625}$ & $0\le h\le 1953$ & $[[11721,1954-h,9767;9767-h]]_{15625}$ \\
			$15625$ & $5$ & $[12502,2084,10418]_{15625}$ & $0\le h\le 2082$ & $[[12502,2084-h,10418;10418-h]]_{15625}$ \\
			$15625$ & $5$ & $[15627,2605,13022]_{15625}$ & $0\le h\le 2604$ & $[[15627,2605-h,13022;13022-h]]_{15625}$ \\
			$15625$ & $125$ & $[15627,125,15502]_{15625}$ & $0\le h\le 124$ & $[[15627,125-h,15502;15502-h]]_{15625}$ \\
			\midrule
			\multicolumn{5}{@{}l}{\emph{Case $s=3$}}\\
			\midrule
			$64$ & $2$ & $[35,11,24]_{64}$ & $0\le h\le 8$ & $[[35,11-h,24;24-h]]_{64}$ \\
			$64$ & $8$ & $[35,4,31]_{64}$ & $0\le h\le 1$ & $[[35,4-h,31;31-h]]_{64}$ \\
			$64$ & $2$ & $[25,4,21]_{64}$ & $0\le h\le 1$ & $[[25,4-h,21;21-h]]_{64}$ \\
			$64$ & $2$ & $[24,5,19]_{64}$ & $0\le h\le 1$ & $[[24,5-h,19;19-h]]_{64}$ \\
			$81$ & $3$ & $[84,4,80]_{81}$ & $0\le h\le 1$ & $[[84,4-h,80;80-h]]_{81}$ \\
			$81$ & $9$ & $[84,4,80]_{81}$ & $0\le h\le 1$ & $[[84,4-h,80;80-h]]_{81}$ \\
			$81$ & $9$ & $[64,4,60]_{81}$ & $0\le h\le 1$ & $[[64,4-h,60;60-h]]_{81}$ \\
			$81$ & $3$ & $[57,4,53]_{81}$ & $0\le h\le 1$ & $[[57,4-h,53;53-h]]_{81}$ \\
			$81$ & $9$ & $[54,4,50]_{81}$ & $0\le h\le 1$ & $[[54,4-h,50;50-h]]_{81}$ \\
			$81$ & $3$ & $[30,7,23]_{81}$ & $0\le h\le 4$ & $[[30,7-h,23;23-h]]_{81}$ \\
			$81$ & $3$ & $[24,5,19]_{81}$ & $0\le h\le 2$ & $[[24,5-h,19;19-h]]_{81}$ \\
			$256$ & $4$ & $[174,4,170]_{256}$ & $0\le h\le 1$ & $[[174,4-h,170;170-h]]_{256}$ \\
			$256$ & $16$ & $[174,4,170]_{256}$ & $0\le h\le 1$ & $[[174,4-h,170;170-h]]_{256}$ \\
			$256$ & $2$ & $[131,42,89]_{256}$ & $0\le h\le 39$ & $[[131,42-h,89;89-h]]_{256}$ \\
			$256$ & $4$ & $[131,5,126]_{256}$ & $0\le h\le 2$ & $[[131,5-h,126;126-h]]_{256}$ \\
			$256$ & $16$ & $[131,5,126]_{256}$ & $0\le h\le 2$ & $[[131,5-h,126;126-h]]_{256}$ \\
			$256$ & $2$ & $[89,5,84]_{256}$ & $0\le h\le 2$ & $[[89,5-h,84;84-h]]_{256}$ \\
			$256$ & $4$ & $[67,12,55]_{256}$ & $0\le h\le 9$ & $[[67,12-h,55;55-h]]_{256}$ \\
			$625$ & $5$ & $[128,20,108]_{625}$ & $0\le h\le 17$ & $[[128,20-h,108;108-h]]_{625}$ \\
			$625$ & $25$ & $[128,4,124]_{625}$ & $0\le h\le 1$ & $[[128,4-h,124;124-h]]_{625}$ \\
			$729$ & $3$ & $[246,60,186]_{729}$ & $0\le h\le 57$ & $[[246,60-h,186;186-h]]_{729}$ \\
			$729$ & $27$ & $[246,4,242]_{729}$ & $0\le h\le 1$ & $[[246,4-h,242;242-h]]_{729}$ \\
			$2401$ & $7$ & $[346,42,304]_{2401}$ & $0\le h\le 39$ & $[[346,42-h,304;304-h]]_{2401}$ \\
			\midrule
			\multicolumn{5}{@{}l}{\emph{Additional examples}}\\
			\midrule
			$9$ & $3$ & $[11,3,8]_9$ & $0\le h\le 2$ & $[[11,3-h,8;8-h]]_9$ \\
			$25$ & $5$ & $[27,5,22]_{25}$ & $0\le h\le 4$ & $[[27,5-h,22;22-h]]_{25}$ \\
			$49$ & $7$ & $[51,7,44]_{49}$ & $0\le h\le 6$ & $[[51,7-h,44;44-h]]_{49}$ \\
			$81$ & $3$ & $[83,22,61]_{81}$ & $21$ & $[[83,1,61;40]]_{81}$ \\
			$121$ & $11$ & $[123,11,112]_{121}$ & $0\le h\le 10$ & $[[123,11-h,112;112-h]]_{121}$ \\
			$625$ & $5$ & $[627,105,522]_{625}$ & $0\le h\le 104$ & $[[627,105-h,522;522-h]]_{625}$ \\
			$625$ & $5$ & $[627,106,521]_{625}$ & $105$ & $[[627,1,521;416]]_{625}$ \\
			$169$ & $13$ & $[45,4,41]_{169}$ & $0\le h\le 3$ & $[[45,4-h,41;41-h]]_{169}$ \\
			$169$ & $13$ & $[46,4,42]_{169}$ & $0\le h\le 2$ & $[[46,4-h,42;42-h]]_{169}$ \\
		\end{longtable}
	}
	
	For the full evaluation set, Corollary~\ref{thm:EA-full-field}(1)--(3)
	gives $0\le h\le k-1$ when $k=1+(q-1)/(p^\ell+1)$, whereas
	part (4) gives only $h=k-1$ when $k=2+(q-1)/(p^\ell+1)$.
\end{example}
	\section{LCD GRL codes and related EAQECCs}\label{sec:LCD}
	
	Taking $h=0$ in the preceding constructions gives $\ell$-Galois LCD GRL codes and maximally entangled EAQECCs. We collect the consequences in one corollary and compare the resulting Hermitian LCD families with those of Liang et al.~\cite{Liang2026LCD}.
	
	\begin{corollary}\label{cor:LCD-EA}
		Let $q=p^e>4$, where $p$ is a prime, and assume that $2\ell\mid e$.
		Every family in Section~\ref{sec:eaqecc} for which $h=0$ is allowed yields an $[n+s,k,d]_q$ $\ell$-Galois LCD GRL code $C$ and a maximally entangled $[[n+s,k,d;n+s-k]]_q$ EAQECC. If $e=2\ell$, then $C$ is Hermitian LCD and also yields a maximally entangled $[[n+s,k,d;n+s-k]]_{p^\ell}$ EAQECC. In particular, the full evaluation set gives an AMDS Hermitian LCD GRL code and an EAQECC with parameters
		$$[q+2,k,q-k+2]_q,\qquad [[q+2,k,q-k+2;q+2-k]]_{p^\ell},\qquad 3\le k\le p^\ell,\quad q=p^{2\ell}.$$
	\end{corollary}
	
	\begin{proof}
		Take $h=0$ in the corresponding hull construction and Proposition~\ref{prop:EA-hull}. When $e=2\ell$, the $\ell$-Galois inner product is Hermitian, and \cite[Corollary~3.2]{GuendaJitmanGulliver2018} gives the stated $p^\ell$-ary EAQECC. For the last assertion, take $n=q$, $\delta=\lambda=1$ and $h=0$ in Theorem~\ref{thm:EA-general}(2); the exact distance follows from Corollary~\ref{thm:EA-full-field}(1)--(2).
	\end{proof}
	
	For the comparison, fix $q=p^{2\ell}>4$. We compare the explicit Hermitian LCD constructions in \cite[Theorems~4.1, 4.4, 4.7, 4.8, 4.9, 4.10 and 4.12]{Liang2026LCD} with the family in Corollary~\ref{cor:LCD-EA}.
	
	Our construction has several advantages over the cited families. It applies to every prime $p$, including $p=2$, and provides length $q+s$, whereas the lengths in \cite{Liang2026LCD} are at most $q-1+s$. It also allows every $s<k\le p^\ell$, without requiring $k\mid(p^\ell-1)$, and permits all extension sizes $2\le s<k$, including $s>k/2$. For $s=2$, the codes with evaluation set $\mathbb{F}_{q}$ are AMDS with the exact distance $q-k+2$. Moreover, $A_s$ may be chosen arbitrarily when $k\le p^\ell-1$, while the boundary value $k=p^\ell$ is covered by the choice in Section~\ref{sec:eaqecc}. Consequently, the same advantages carry over to the resulting maximally entangled EAQECCs.
	
	\begin{example}\label{ex:LCD-comparison}
		Table~\ref{tab:LCD-examples} first compares our codes with the concrete Hermitian LCD codes listed in Appendix B of \cite{Liang2026LCD}, then with the largest lengths supplied by their cited family at the same $q,k,s=2$, and finally gives parameters not covered by their cited constructions.
	\end{example}
	
	\begin{table}[H]
		\centering
		\caption{Hermitian LCD GRL codes and their associated maximally entangled EAQECCs compared with Liang et al.~\cite{Liang2026LCD}.}
		\label{tab:LCD-examples}
		\small
		\setlength{\tabcolsep}{4pt}
		\renewcommand{\arraystretch}{1.16}
		\begin{tabular}{@{}ccccc@{}}
			\toprule
			$q$ & \multicolumn{2}{c}{This paper} & \multicolumn{2}{c}{Liang et al.~\cite{Liang2026LCD}} \\
			\cmidrule(lr){2-3}\cmidrule(lr){4-5}
			& Hermitian LCD GRL code & EAQECC & Hermitian LCD GRL code & EAQECC \\
			\midrule
			\multicolumn{5}{c}{\emph{Examples listed in Appendix B of \cite{Liang2026LCD}}} \\
			\midrule
			$81$ & $[83,8,75]_{81}$ & $[[83,8,75;75]]_9$ & $[10,8,3]_{81}$ & $[[10,8,3;2]]_9$ \\
			$625$ & $[627,8,619]_{625}$ & $[[627,8,619;619]]_{25}$ & $[13,8,5]_{625}$ & $[[13,8,5;5]]_{25}$ \\
			$169$ & $[171,6,165]_{169}$ & $[[171,6,165;165]]_{13}$ & $[14,6,8]_{169}$ & $[[14,6,8;8]]_{13}$ \\
			$121$ & $[123,5,118]_{121}$ & $[[123,5,118;118]]_{11}$ & $[37,5,32]_{121}$ & $[[37,5,32;32]]_{11}$ \\
			\midrule
			\multicolumn{5}{c}{\emph{Largest lengths from the cited family at the same $q,k,s=2$}} \\
			\midrule
			$49$ & $[51,6,45]_{49}$ & $[[51,6,45;45]]_7$ & $[32,6,26]_{49}$ & $[[32,6,26;26]]_7$ \\
			$81$ & $[83,8,75]_{81}$ & $[[83,8,75;75]]_9$ & $[58,8,50]_{81}$ & $[[58,8,50;50]]_9$ \\
			$121$ & $[123,10,113]_{121}$ & $[[123,10,113;113]]_{11}$ & $[72,10,62]_{121}$ & $[[72,10,62;62]]_{11}$ \\
			$169$ & $[171,12,159]_{169}$ & $[[171,12,159;159]]_{13}$ & $[134,12,122]_{169}$ & $[[134,12,122;122]]_{13}$ \\
			\midrule
			\multicolumn{5}{c}{\emph{Parameters not covered by the cited constructions}} \\
			\midrule
			$49$ & $[51,7,44]_{49}$ & $[[51,7,44;44]]_7$ & \multicolumn{2}{c}{Not covered ($k\nmid(p^\ell-1)$)} \\
			$64$ & $[66,8,58]_{64}$ & $[[66,8,58;58]]_8$ & \multicolumn{2}{c}{Not covered ($p=2$)} \\
			$81$ & $[84,4,80]_{81}$ & $[[84,4,80;80]]_9$ & \multicolumn{2}{c}{Not covered ($s>k/2$)} \\
			$625$ & $[128,4,124]_{625}$ & $[[128,4,124;124]]_{25}$ & \multicolumn{2}{c}{Not covered ($s>k/2$)} \\
			$64$ & $[35,4,31]_{64}$ & $[[35,4,31;31]]_8$ & \multicolumn{2}{c}{Not covered ($p=2$)} \\
			$256$ & $[131,5,126]_{256}$ & $[[131,5,126;126]]_{16}$ & \multicolumn{2}{c}{Not covered ($p=2$)} \\
			\bottomrule
		\end{tabular}
	\end{table}

	\section{Conclusion}\label{sec:conclusion}
	
	In this paper, we studied GRL codes with prescribed $\ell$-Galois hull dimensions and their applications to EAQECCs. After introducing the problem in Section~\ref{sec:introduction}, Section~\ref{sec:preliminaries} recalled the necessary Galois-duality tools and the polynomial description of GRL duals. Section~\ref{sec:construction} developed a common multiplier method and six explicit families based on multiplicative, additive, and mixed evaluation sets. Sections~\ref{sec:s2} and~\ref{sec:s3} specialized these constructions to $s=2$ and $s=3$, respectively, and combined the prescribed hull dimensions with MDS and AMDS criteria. Section~\ref{sec:eaqecc} converted the resulting classical codes into EAQECCs with explicit parameters and gave concrete AMDS examples. Finally, Section~\ref{sec:LCD} treated the zero-hull specialization, obtaining $\ell$-Galois and Hermitian LCD GRL codes together with maximally entangled EAQECCs.
	
	\printbibliography

@article{9559992,
	author  = {Cao, Meng},
	title   = {{MDS} codes with {Galois} hulls of arbitrary dimensions and the related entanglement-assisted quantum error correction},
	journal = {IEEE Transactions on Information Theory},
	volume  = {67},
	number  = {12},
	pages   = {7964--7984},
	year    = {2021}
}

@article{li2025new,
	author  = {Li, Fengwei and Jiang, Ruiyuan and Liu, Yuting},
	title   = {New {MDS} and self-dual generalized {Roth--Lempel} codes as well as their deep holes},
	journal = {Designs, Codes and Cryptography},
	volume  = {93},
	number  = {11},
	pages   = {5079--5096},
	year    = {2025}
}

@article{li2023several,
	author  = {Li, Yang and Su, Yunfei and Zhu, Shixin and Li, Shitao and Shi, Minjia},
	title   = {Several classes of {Galois} self-orthogonal {MDS} codes and related applications},
	journal = {Finite Fields and Their Applications},
	volume  = {91},
	pages   = {102267},
	year    = {2023}
}

@article{roth1989construction,
	author  = {Roth, Ron M. and Lempel, Abraham},
	title   = {A construction of non-{Reed--Solomon} type {MDS} codes},
	journal = {IEEE Transactions on Information Theory},
	volume  = {35},
	number  = {3},
	pages   = {655--657},
	year    = {1989}
}

@article{liu2026generalized,
	author  = {Liu, Qi and Wu, Xuefei and Cheng, Yingchun and Zhou, Haiyan},
	title   = {Generalized Roth--Lempel Codes: NMDS Characterization, Hermitian Self-Orthogonality, and Quantum Constructions},
	journal = {arXiv preprint arXiv:2604.11350},
	year    = {2026}
}

@book{Lidl_Niederreiter_1996,
	author    = {Lidl, Rudolf and Niederreiter, Harald},
	title     = {Finite Fields},
	edition   = {Second},
	series    = {Encyclopedia of Mathematics and its Applications},
	publisher = {Cambridge University Press},
	address   = {Cambridge},
	year      = {1996}
}

@misc{Liang2026LCD,
	author        = {Liang, Zhonghao and Huang, Dongmei and Liao, Qunying and Fan, Cuiling and Zhou, Zhengchun},
	title         = {Non-{GRS} type Euclidean and Hermitian {LCD} codes and Their Applications for {EAQECCs}},
	year          = {2026},
	eprint        = {2603.16187},
	archivePrefix = {arXiv},
	primaryClass  = {cs.IT}
}

@article{LiZhu2024,
	author  = {Li, Yang and Zhu, Shixin},
	title   = {Linear codes of larger lengths with {Galois} hulls of arbitrary dimensions and related entanglement-assisted quantum error-correcting codes},
	journal = {Discrete Mathematics},
	volume  = {347},
	number  = {2},
	pages   = {113760},
	year    = {2024}
}

@article{LiuChen2025,
	author  = {Liu, Jingge and Chen, Bocong},
	title   = {The intersection of two generalized {Reed--Solomon} codes},
	journal = {IEEE Transactions on Information Theory},
	volume  = {71},
	number  = {10},
	pages   = {7595--7608},
	year    = {2025}
}

@article{Shor1995,
	author  = {Shor, Peter W.},
	title   = {Scheme for reducing decoherence in quantum computer memory},
	journal = {Physical Review A},
	volume  = {52},
	number  = {4},
	pages   = {R2493--R2496},
	year    = {1995}
}

@article{Steane1996,
	author  = {Steane, Andrew M.},
	title   = {Error correcting codes in quantum theory},
	journal = {Physical Review Letters},
	volume  = {77},
	number  = {5},
	pages   = {793--797},
	year    = {1996}
}

@article{CalderbankShor1996,
	author  = {Calderbank, A. Robert and Shor, Peter W.},
	title   = {Good quantum error-correcting codes exist},
	journal = {Physical Review A},
	volume  = {54},
	number  = {2},
	pages   = {1098--1105},
	year    = {1996}
}

@article{Gottesman1996,
	author  = {Gottesman, Daniel},
	title   = {Class of quantum error-correcting codes saturating the quantum {Hamming} bound},
	journal = {Physical Review A},
	volume  = {54},
	number  = {3},
	pages   = {1862--1868},
	year    = {1996}
}

@article{KnillLaflamme1997,
	author  = {Knill, Emanuel and Laflamme, Raymond},
	title   = {Theory of quantum error-correcting codes},
	journal = {Physical Review A},
	volume  = {55},
	number  = {2},
	pages   = {900--911},
	year    = {1997}
}

@article{CRSS1998,
	author  = {Calderbank, A. Robert and Rains, Eric M. and Shor, Peter W. and Sloane, Neil J. A.},
	title   = {Quantum error correction via codes over {GF(4)}},
	journal = {IEEE Transactions on Information Theory},
	volume  = {44},
	number  = {4},
	pages   = {1369--1387},
	year    = {1998}
}

@article{AshikhminKnill2001,
	author  = {Ashikhmin, Alexei and Knill, Emanuel},
	title   = {Nonbinary quantum stabilizer codes},
	journal = {IEEE Transactions on Information Theory},
	volume  = {47},
	number  = {7},
	pages   = {3065--3072},
	year    = {2001}
}

@article{BrunDevetakHsieh2006,
	author  = {Brun, Todd and Devetak, Igor and Hsieh, Min-Hsiu},
	title   = {Correcting quantum errors with entanglement},
	journal = {Science},
	volume  = {314},
	number  = {5798},
	pages   = {436--439},
	year    = {2006}
}

@article{WildeBrun2008,
	author  = {Wilde, Mark M. and Brun, Todd A.},
	title   = {Optimal entanglement formulas for entanglement-assisted quantum coding},
	journal = {Physical Review A},
	volume  = {77},
	number  = {6},
	pages   = {064302},
	year    = {2008}
}

@article{GuendaJitmanGulliver2018,
	author  = {Guenda, Kenza and Jitman, Somphong and Gulliver, T. Aaron},
	title   = {Constructions of good entanglement-assisted quantum error correcting codes},
	journal = {Designs, Codes and Cryptography},
	volume  = {86},
	number  = {1},
	pages   = {121--136},
	year    = {2018}
}

@article{GalindoHernandoMatsumotoRuano2019,
	author  = {Galindo, Carlos and Hernando, Fernando and Matsumoto, Ryutaroh and Ruano, Diego},
	title   = {Entanglement-assisted quantum error-correcting codes over arbitrary finite fields},
	journal = {Quantum Information Processing},
	volume  = {18},
	number  = {4},
	pages   = {116},
	year    = {2019}
}

@article{FanZhang2017,
	author  = {Fan, Yun and Zhang, Liang},
	title   = {{Galois} self-dual constacyclic codes},
	journal = {Designs, Codes and Cryptography},
	volume  = {84},
	number  = {3},
	pages   = {473--492},
	year    = {2017}
}

@article{LiuPan2020,
	author  = {Liu, Hongwei and Pan, Xu},
	title   = {{Galois} hulls of linear codes over finite fields},
	journal = {Designs, Codes and Cryptography},
	volume  = {88},
	number  = {2},
	pages   = {241--255},
	year    = {2020}
}

@article{ReedSolomon1960,
	author  = {Reed, Irving S. and Solomon, Gustave},
	title   = {Polynomial codes over certain finite fields},
	journal = {Journal of the Society for Industrial and Applied Mathematics},
	volume  = {8},
	number  = {2},
	pages   = {300--304},
	year    = {1960}
}

@article{GrasslBethRoetteler2004,
	author  = {Grassl, Markus and Beth, Thomas and R{\"o}tteler, Martin},
	title   = {On optimal quantum codes},
	journal = {International Journal of Quantum Information},
	volume  = {2},
	number  = {1},
	pages   = {55--64},
	year    = {2004}
}

@article{LiXingWang2008,
	author  = {Li, Zhuo and Xing, Li-Juan and Wang, Xin-Mei},
	title   = {Quantum generalized {Reed--Solomon} codes: unified framework for quantum maximum-distance-separable codes},
	journal = {Physical Review A},
	volume  = {77},
	number  = {1},
	pages   = {012308},
	year    = {2008}
}

@article{JinLingLuoXing2010,
	author  = {Jin, Lingfei and Ling, San and Luo, Jinquan and Xing, Chaoping},
	title   = {Application of classical {Hermitian} self-orthogonal {MDS} codes to quantum {MDS} codes},
	journal = {IEEE Transactions on Information Theory},
	volume  = {56},
	number  = {9},
	pages   = {4735--4740},
	year    = {2010}
}

@article{JinXing2014,
	author  = {Jin, Lingfei and Xing, Chaoping},
	title   = {A construction of new quantum {MDS} codes},
	journal = {IEEE Transactions on Information Theory},
	volume  = {60},
	number  = {5},
	pages   = {2921--2925},
	year    = {2014}
}

@article{LuoCaoChen2019,
	author  = {Luo, Gaojun and Cao, Xiwang and Chen, Xiaojing},
	title   = {{MDS} codes with hulls of arbitrary dimensions and their quantum error correction},
	journal = {IEEE Transactions on Information Theory},
	volume  = {65},
	number  = {5},
	pages   = {2944--2952},
	year    = {2019}
}

@article{FangFuLiZhu2020,
	author  = {Fang, Weijun and Fu, Fang-Wei and Li, Lanqiang and Zhu, Shixin},
	title   = {{Euclidean} and {Hermitian} hulls of {MDS} codes and their applications to {EAQECCs}},
	journal = {IEEE Transactions on Information Theory},
	volume  = {66},
	number  = {6},
	pages   = {3527--3537},
	year    = {2020}
}

@article{FangJinLuoMa2022,
	author  = {Fang, Xiaolei and Jin, Renjie and Luo, Jinquan and Ma, Wen},
	title   = {New {Galois} hulls of {GRS} codes and application to {EAQECCs}},
	journal = {Cryptography and Communications},
	volume  = {14},
	number  = {1},
	pages   = {145--159},
	year    = {2022}
}

@article{WuLiYang2022,
	author  = {Wu, Yansheng and Li, Chengju and Yang, Shangdong},
	title   = {New {Galois} hulls of generalized {Reed--Solomon} codes},
	journal = {Finite Fields and Their Applications},
	volume  = {83},
	pages   = {102084},
	year    = {2022}
}

@article{LiZhuLi2023,
	author  = {Li, Yang and Zhu, Shixin and Li, Ping},
	title   = {On {MDS} codes with {Galois} hulls of arbitrary dimensions},
	journal = {Cryptography and Communications},
	volume  = {15},
	number  = {3},
	pages   = {565--587},
	year    = {2023}
}

@article{Cao2023,
	author  = {Cao, Meng},
	title   = {Several new families of {MDS} {EAQECCs} with much larger dimensions and related application to {EACQCs}},
	journal = {Quantum Information Processing},
	volume  = {22},
	pages   = {447},
	year    = {2023}
}

@article{WanZhu2025,
	author  = {Wan, Ruhao and Zhu, Shixin},
	title   = {Construction of {Galois} self-orthogonal {MDS} codes with larger dimensions},
	journal = {Finite Fields and Their Applications},
	volume  = {108},
	pages   = {102665},
	year    = {2025}
}

@article{WuHyunLee2021,
	author  = {Wu, Yansheng and Hyun, Jong Yoon and Lee, Yoonjin},
	title   = {New {LCD} {MDS} codes of non-{Reed--Solomon} type},
	journal = {IEEE Transactions on Information Theory},
	volume  = {67},
	number  = {8},
	pages   = {5069--5078},
	year    = {2021}
}

@article{HanFan2023,
	author  = {Han, Dongchun and Fan, Cuiling},
	title   = {{Roth--Lempel} {NMDS} codes of non-elliptic-curve type},
	journal = {IEEE Transactions on Information Theory},
	volume  = {69},
	number  = {9},
	pages   = {5670--5675},
	year    = {2023}
}

@misc{WuHengLiDing2024,
	author        = {Wu, Yansheng and Heng, Ziling and Li, Chengju and Ding, Cunsheng},
	title         = {More {MDS} Codes of Non-Reed--Solomon Type},
	year          = {2024},
	eprint        = {2401.03391},
	archivePrefix = {arXiv},
	primaryClass  = {cs.IT}
}

@misc{LiangWanLiao2026GRL,
	author        = {Liang, Zhonghao and Wan, Yongkang and Liao, Qunying},
	title         = {The equivalent condition for GRL codes to be MDS, AMDS or self-dual},
	year          = {2025},
	eprint        = {2506.03874},
	archivePrefix = {arXiv}
}

@article{LiangLiao2026NMDS,
	author  = {Liang, Zhonghao and Liao, Qunying},
	title   = {Two classes of {NMDS} codes from {Roth--Lempel} codes},
	journal = {Finite Fields and Their Applications},
	volume  = {111},
	pages   = {102779},
	year    = {2026}
}

@article{LiangLiao2026Extended,
	author  = {Liang, Zhonghao and Liao, Qunying},
	title   = {The extended code for a class of generalized {Roth--Lempel} codes and their properties},
	journal = {Discrete Mathematics},
	volume  = {349},
	number  = {8},
	pages   = {115084},
	year    = {2026}
}
\end{document}